\documentclass[12pt]{article}
\usepackage[colorlinks,linkcolor=Blue,citecolor=Blue,bookmarks,bookmarksnumbered]{hyperref}
\usepackage[scaled=0.85]{helvet}

\usepackage[LGR, T1]{fontenc}
\usepackage[english]{babel}

\usepackage{amsfonts,amsthm,cite,tcolorbox,xcolor,latexsym,dsfont,mathtools,amsmath,amssymb,accents,mathrsfs,XoohmE}
\usepackage[frozencache=true,cachedir=.]{minted}
\usepackage{graphicx,booktabs,multirow,placeins,subfigure}
\usepackage[subsection]{algorithm}
\usepackage{algpseudocode}

\definecolor{Green}  {rgb}{0.10,0.70,0.10} 
\definecolor{Orange} {rgb}{1.00,0.50,0.15} 
\definecolor{Red}    {rgb}{0.90,0.00,0.12} 
\definecolor{Purple} {rgb}{0.50,0.25,0.55} 
\definecolor{Turque} {rgb}{0.00,0.65,0.85} 
\definecolor{Blue}   {rgb}{0.00,0.00,1.00} 
\definecolor{Magenta}{rgb}{1.00,0.00,1.00} 
\definecolor{Gold}   {rgb}{1.00,0.75,0.25} 
\definecolor{Seaweed}{rgb}{0.01,0.24,0.09} 
\definecolor{Brown}  {rgb}{0.43,0.26,0.32} 
\definecolor{grey1}  {rgb}{0.20,0.20,0.20} 
\definecolor{grey2}  {rgb}{0.40,0.40,0.40} 
\definecolor{grey3}  {rgb}{0.60,0.60,0.60} 
\definecolor{grey4}  {rgb}{0.80,0.80,0.80} 
\definecolor{grey5}  {rgb}{0.90,0.90,0.90} 
\def\C#1#2{{\ifcase#1\or
             \color{Green}\or \color{Orange}\or \color{Red}\or
              \color{Purple}\or \color{Turque}\or \color{Blue}\or
               \color{Magenta}\or \color{Gold}\or \color{Seaweed}\or
                \color{Brown}\or\color{grey1}\or\color{grey2}\or
                 \color{grey3}\else\color{grey4}\fi#2}}

\definecolor{Slate} {rgb}{0.00,0.45,0.55}

\def\be{\begin{equation}}
\def\ee{\end{equation}}
\newcommand{\bea}{\begin{eqnarray}}
\newcommand{\eea}{\end{eqnarray}}
\newcommand{\ena}{\end{eqnarray}}

\def\pp{{\mathchoice
          {
              \kern 1pt%
              \raise 1pt
              \vbox{\hrule width5pt height0.4pt depth0pt
                    \kern -2pt
                    \hbox{\kern 2.3pt
                          \vrule width0.4pt height6pt depth0pt
                          }
                    \kern -2pt
                    \hrule width5pt height0.4pt depth0pt}%
                    \kern 1pt
           }
            {
              \kern 1pt%
              \raise 1pt
              \vbox{\hrule width4.3pt height0.4pt depth0pt
                    \kern -1.8pt
                    \hbox{\kern 1.95pt
                          \vrule width0.4pt height5.4pt depth0pt
                          }
                    \kern -1.8pt
                    \hrule width4.3pt height0.4pt depth0pt}%
                    \kern 1pt
            }
            {
              \kern 0.5pt%
              \raise 1pt
              \vbox{\hrule width4.0pt height0.3pt depth0pt
                    \kern -1.9pt  
                    \hbox{\kern 1.85pt
                          \vrule width0.3pt height5.7pt depth0pt
                          }
                    \kern -1.9pt
                    \hrule width4.0pt height0.3pt depth0pt}%
                    \kern 0.5pt
            }
            {
              \kern 0.5pt%
              \raise 1pt
              \vbox{\hrule width3.6pt height0.3pt depth0pt
                    \kern -1.5pt
                    \hbox{\kern 1.65pt
                          \vrule width0.3pt height4.5pt depth0pt
                          }
                    \kern -1.5pt
                    \hrule width3.6pt height0.3pt depth0pt}%
                    \kern 0.5pt
            }
        }}

\def\mm{{\mathchoice
                       {
                             \kern 1pt
               \raise 1pt    \vbox{\hrule width5pt height0.4pt depth0pt
                                  \kern 2pt
                                  \hrule width5pt height0.4pt depth0pt}
                             \kern 1pt}
                       {
                            \kern 1pt
               \raise 1pt \vbox{\hrule width4.3pt height0.4pt depth0pt
                                  \kern 1.8pt
                                  \hrule width4.3pt height0.4pt depth0pt}
                             \kern 1pt}
                       {
                            \kern 0.5pt
               \raise 1pt
                            \vbox{\hrule width4.0pt height0.3pt depth0pt
                                  \kern 1.9pt
                                  \hrule width4.0pt height0.3pt depth0pt}
                            \kern 1pt}
                       {
                           \kern 0.5pt
             \raise 1pt  \vbox{\hrule width3.6pt height0.3pt depth0pt
                                  \kern 1.5pt
                                  \hrule width3.6pt height0.3pt depth0pt}
                           \kern 0.5pt}
                       }}

\def\ad{{\kern0.5pt
                   \alpha \kern-5.05pt \raise5.8pt\hbox{$\textstyle.$}\kern
0.5pt}}

\def\bd{{\kern0.5pt
                   \beta \kern-5.05pt \raise5.8pt\hbox{$\textstyle.$}\kern
0.5pt}}

\def\qd{{\kern0.5pt
                   q \kern-5.05pt \raise5.8pt\hbox{$\textstyle.$}\kern
0.5pt}}
\def\Dot#1{{\kern0.5pt
     {#1} \kern-5.05pt \raise5.8pt\hbox{$\textstyle.$}\kern
0.5pt}}

\catcode`@=11
\def\un#1{\relax\ifmmode\@@underline#1\else
        $\@@underline{\hbox{#1}}$\relax\fi}
\catcode`@=12

\def\ca{{\cal A}}

\def\cf{{\cal F}}

\def\dslash{\not{\hbox{\kern-2pt $\partial$}}}
\def\Dslash{\not{\hbox{\kern-4pt $D$}}}
\def\pslash{\not{\hbox{\kern-2.3pt $p$}}}
 \newtoks\slashfraction
 \slashfraction={.13}
 \def\slash#1{\setbox0\hbox{$ #1 $}
 \setbox0\hbox to \the\slashfraction\wd0{\hss \box0}/\box0 }
 
\def\kcr{{\hbox{\ro \char'170}}}                
\def\ktl{{\hbox{\ro \char'170}}}        
\def\ktr{{\hbox{\ro \char'170}}}        
\def\kbl{{\hbox{\ro \char'170}}}        
\def\kbr{{\hbox{\ro \char'170}}}        

\def\plpl{\raise-2pt\hbox{$\raise3pt\hbox{$_+$}\hskip-6.67pt\raise0.0pt
\hbox{$^+$}\hskip 0.01pt$}}
\def\mimi{\raise-2pt\hbox{$\raise3pt\hbox{$_-$}\hskip-6.67pt\raise0.0pt
\hbox{$^-$}\hskip 0.01pt$}} 

\def\bo{{\raise.15ex\hbox{\large$\Box$}}}               
\def\TH{{\raise.2ex\hbox{$\displaystyle \bigodot$}\mskip-4.7mu \llap H \;}}
\def\face{{\raise.2ex\hbox{$\displaystyle \bigodot$}\mskip-2.2mu \llap {$\ddot
        \smile$}}}                                      

\def\dt#1{\on{\hbox{\bf .}}{#1}}                
\def\Dot#1{\dt{#1}}

\def\leftrightarrowfill{$\mathsurround=0pt \mathord\leftarrow \mkern-6mu
        \cleaders\hbox{$\mkern-2mu \mathord- \mkern-2mu$}\hfill
        \mkern-6mu \mathord\rightarrow$}
\def\dvec#1{\vbox{\ialign{##\crcr
        \leftrightarrowfill\crcr\noalign{\kern-1pt\nointerlineskip}
        $\hfil\displaystyle{#1}\hfil$\crcr}}}           
\def\dt#1{{\buildrel {\hbox{\LARGE .}} \over {#1}}}     

\def\sfrac#1#2{{\vphantom1\smash{\lower.5ex\hbox{\small$#1$}}\over
        \vphantom1\smash{\raise.4ex\hbox{\small$#2$}}}} 
\def\bfrac#1#2{{\vphantom1\smash{\lower.5ex\hbox{$#1$}}\over
        \vphantom1\smash{\raise.3ex\hbox{$#2$}}}}       
\def\afrac#1#2{{\vphantom1\smash{\lower.5ex\hbox{$#1$}}\over#2}}    

\def\ad{{\dot{\alpha}}}
\def\bd{{\dot{\beta}}}

 \font\rOpe=cmsy10                        
 \def\ktl{{\hbox{\rOpe\char'170}}}        
 \def\kbl{{\hbox{\rOpe\char'170}}}        
 \def\kcr{{\reflectbox{\rOpe\char'170}}}        
 \def\ktr{{\reflectbox{\rOpe\char'170}}}        
 \def\kbr{{\reflectbox{\rOpe\char'170}}}        
 \def\Border{\vbox{\hsize0pt
        \setlength{\unitlength}{1mm}
        \newcount\xco
        \newcount\yco
        \xco=-21
        \yco=12
        \begin{picture}(0,0)(-7.5,0)
        \put(\xco,\yco){$\ktl$}
        \advance\yco by-1
        {\loop
        \put(\xco,\yco){$\kcr$}
        \advance\yco by-2
        \ifnum\yco>-240
        \repeat
        \put(\xco,\yco){$\kbl$}}
        \xco=170
        \yco=12
        \put(\xco,\yco){$\ktr$}
        \advance\yco by-1
        {\loop
        \put(\xco,\yco){$\kcr$}
        \advance\yco by-2
        \ifnum\yco>-240
        \repeat
        \put(\xco,\yco){$\kbr$}}
        \put(-19.5,13){\scalebox{.6065}{%
         University of Maryland Center for String and Particle  Theory \&\ Physics Department%
        |University of Maryland Center for String and Particle  Theory \&\ Physics Department}}
        \put(-19.5,-241.5){\scalebox{.5835}{%
         ****University of Maryland * Center for String and
         Particle  Theory* Physics Department****University of Maryland *Center
        for String and Particle  Theory* Physics Department}}
        \end{picture}
        \par\vskip-8mm}}
\definecolor{UMred}{rgb}{.9,.05,.2}
\definecolor{HUblue}{rgb}{.0,.3,.7}

\newdimen\parshift\parshift=\parindent
\catcode`@=11
 \long\def\@footnotetext#1{\insert\footins{\reset@font\footnotesize
           \interlinepenalty\interfootnotelinepenalty\splittopskip%
            \footnotesep\splitmaxdepth\dp\strutbox\floatingpenalty\@MM%
             \hsize\columnwidth\addtolength{\hsize}{-2\parindent}
              \@parboxrestore\protected@edef\@currentlabel%
              {\csname p@footnote\endcsname\@thefnmark}%
                \color@begingroup%
                 \@makefntext{\rule\z@\footnotesep\ignorespaces#1%
                  \@finalstrut\strutbox}%
                \color@endgroup}}
 \long\def\@makefntext#1{\hglue\parshift%
           \vbox{\noindent\baselineskip=11pt plus.5pt minus.5pt\hb@xt@0em{\hss\@makefnmark\kern1pt}#1}}
\catcode`@=12

\newskip\humongous \humongous=0pt plus 1000pt minus 1000pt

\newif\ifdtup

\makeatletter
\def\section{\@startsection{section}{1}{\z@}
        {3ex plus-1ex minus-.2ex}{1pt plus1pt}{\large\sf\bfseries\boldmath}}
\def\subsection{\@startsection{subsection}{2}{\z@}
         {1.5ex plus-1ex minus-.2ex}{0.01pt plus1pt}{\sf\slshape}}
\def\subsubsection{\@startsection{subsubsection}{3}{\z@}
          {1.5ex plus-1ex minus-.2ex}{0.01pt plus0.2pt}{\sf\boldmath}}
\def\paragraph{\@startsection{paragraph}{4}{\z@}
           {.75ex \@plus.5ex \@minus.2ex}{-2mm}{\sf\bfseries\boldmath}}
\makeatother

\usepackage{pifont}
\newcommand{\cmark}{\ding{51}}
\newcommand{\xmark}{\ding{55}}

\newcommand\numberthis{\addtocounter{equation}{1}\tag{\theequation}}
\newcommand{\Mod}[1]{\ (\mathrm{mod}\ #1)}

\numberwithin{equation}{section}
\allowdisplaybreaks

\newlength{\leftside}
\newlength{\rightside}
\newcommand*{\leftterm}{}
\newcommand*{\rightterm}{}
\newcommand*{\term}[1]{$\displaystyle#1$}

\newtheorem{theorem}{Theorem}[section]
\theoremstyle{definition}
\newtheorem{corollary}[theorem]{Corollary}
\newtheorem{definition}[theorem]{Definition}

\newtheorem{proposition}[theorem]{Proposition}

\DeclareMathOperator{\tr}{tr}
\DeclarePairedDelimiter\floor{\lfloor}{\rfloor}
\DeclarePairedDelimiter\ceil{\lceil}{\rceil}

\definecolor{LightGray}{HTML}{f6f8fa}
\definecolor{DarkGreen}{RGB}{46, 125, 50}

\setminted[python]{bgcolor=LightGray, fontsize=\footnotesize, breaklines, breaksymbolleft=\color{black}\tiny\ensuremath{\hookrightarrow}}

\makeatletter
\algrenewcommand\ALG@beginalgorithmic{\scriptsize}
\makeatother
\algrenewcommand\alglinenumber[1]{\scriptsize #1:}

\begin{document}

\thispagestyle{empty}
\noindent{\small
\hfill{  \\ 
$~~~~~~~~~~~~~~~~~~~~~~~~~~~~~~~~~~~~~~~~~~~~~~~~~~~~~~~~~~~~~~~~~$
$~~~~~~~~~~~~~~~~~~~~~~~~~~~~~~~~~~~~~~~~~~~~~~~~~~~~~~~~~~~~~~~~~$
{}
}}
\vspace*{0mm}
\begin{center}
{\large \bf
Modern Tensor-Spinor Symbolic Algebra Algorithms and
Computing  \\[2pt]
Non-Closure Geometry
\& Holoraumy in 11D,
\(\mathcal{N} = 1\) Supergravity  \\[2pt]
}   \vskip0.3in
{\large {
S. James Gates, Jr.,\footnote{gatess@umd.edu}\(^{,a}\)
Isaiah B. Hilsenrath,\footnote{ibh@sas.upenn.edu}\(^{,b}\) and 
Saul Hilsenrath\footnote{saulhils@sas.upenn.edu}\(^{,b}\)
}}
\\*[8mm]
{\normalsize \emph{
\centering
$^{a} $Department of Physics, University of Maryland,
\\[1pt]
College Park, MD 20742, USA
 \\[10pt]
$^{b} $Department of Mathematics, University of Pennsylvania,
\\[1pt]
Philadelphia, PA 19104, USA
\\[10mm]
}}
{\textbf{ABSTRACT}}\\[4mm]
\parbox{142mm}{\parindent=2pc\indent\baselineskip=14pt plus1pt
The Supersymmetry Genomics project aims to classify supermultiplets by properties like adinkras and holoraumy. The project's protocol is: 1) set up the SUSY transformation rules (and action) with unknown numerical coefficients, 2) solve those coefficients, and 3) compute desired properties, e.g., non-closure geometry (the non-closure functions in the anticommutator of supercovariant derivatives) and holoraumy (the commutator of supercovariant derivatives). This paper provides the first broadly applicable computer algorithms for completing these computations, comprising a new Cadabra module we call ``SusyPy.'' We provide a significant extension of the available tensor-arithmetic/canonicalization to include spinor-indexed expressions and NW-SE convention, and we provide a new Fierz expansion algorithm for spinor-indexed expressions. On top of this new tensor-spinor symbolic algebra, we have built algorithms for solving for the coefficients of a multiplet and its action, and for computing holoraumy, for both off-shell and on-shell \(\mathcal{N} = 1\) multiplets of any dimension. We apply our tools to assimilate linearized 11D, \(\mathcal{N} = 1\) supergravity into the Genomics project. We demonstrate solving the 11D, \(\mathcal{N} = 1\) multiplet, as in Cremmer, Julia, and Scherk, in a few lines of code, and we provide a comprehensive picture of the multiplet's non-closure geometry. Further, we provide the first-ever computation of the holoraumy of 11D, \(\mathcal{N} = 1\) supergravity.}
\end{center}
\vfill
\vfill
\clearpage
%

%

{
    \hypersetup{hidelinks}
}

\section{Introduction}
For roughly 15 years, the Supersymmetry Genomics \cite{Gates2009,Gates2011,Chappell2013,Mak2019} project has aimed to classify supermultiplets according to their adinkras, and more recently, according to their holoraumy. Throughout this project, linearized multiplets, typically characterized by their sets of fields, have been studied systematically according to the following procedure. We use index/dimension comparison to abstractly write supersymmetry transformation rules for the multiplet, each of the general form
\begin{equation}\label{susy_rule_ex}
    {\rm D}_{\alpha} A = \sum_{j} u_j B_{\alpha}^j,
\end{equation}
where \({\rm D}_{\alpha}\) is a supercovariant derivative, \(A\) is a field, the \(B_{\alpha}^j\) are terms in the fields of opposite type (e.g., fermions if \(A\) is a boson), and the \(u_j\) are a priori unknown scalar coefficients. We also set up an action consisting of a Lagrangian density which is a linear combination of terms in the fields with unknown scalar coefficients. Then, we grind through tensor-spinor algebra and Fierz identities to obtain constraints on the scalar coefficients from closure of the multiplet, supersymmetry-invariance of the action, and normalization of the action, and we use these constraints to establish numerical values of the coefficients (usually up to an overall scaling and phase). We call this process ``solving the multiplet.''

Once a multiplet is solved, we can consider its properties (its ``genetics''), including non-closure geometry (in the on-shell case) and holoraumy. Non-closure geometry \cite{Gates2020a} is the set of terms by which the anticommutator \(\{{\rm D}_{\alpha}, {\rm D}_{\beta}\}\) of supercovariant derivatives differs from a standard translation. These non-closure terms are a reflection of fields' gauge symmetry and equations of motion. The problem of eliminating equation-of-motion terms by adjoining auxiliary fields to get an off-shell multiplet (the so-called ``off-shell SUSY problem'') was dubbed by supersymmetry's fundamental challenge by the first author \cite{Gates2002} and was a key motivation for the Genomics project \cite{Gates2009}. On the other hand, holoraumy \cite{Gates2015} is an analogue of holonomy, in the sense that holoraumy is related to the commutator of supercovariant derivatives, like holonomy is related to the commutator of covariant derivatives. Holoraumy was established as an empirically holography-related tensor when for a variety of off-shell 4D, \(\mathcal{N} = 1\) multiplets, it was shown that the holoraumy associated with the 1D, \(\mathcal{N} = 4\) adinkras obtained by dimensional reduction carried manifestations of 3D spatial rotations and extended \(R\)-symmetry \cite{Gates2015}. More recently, it was discovered \cite{Gates2020a} that for on-shell multiplets, the holoraumy featured a natural rotational symmetry similar to the Maxwellian symmetry between the electric and magnetic fields, although the rotations turn out not to be present in the holoraumy for 10D, \(\mathcal{N} = 1\) super Maxwell theory \cite{Gates2022}.

As the Genomics project has ventured to treat more complicated multiplets, the need for software to handle elements of the aforementioned process has become apparent. Multiplet-specific Mathematica code has been used to solve certain on-shell 4D, \(\mathcal{N} = 1\) multiplets and compute their holoraumy \cite{Gates2020a}, and Adinkra.m has been demonstrated for the computation of adinkras and holoraumy of off-shell 4D, \(\mathcal{N} = 1\) multiplets \cite{Mak2019}. At the same time, new software have appeared that handle tensor arithmetic/canonicalization and the algebra of gamma matrices without spinor indices, including (in roughly chronological order) GAMMA \cite{Gran2001}, Cadabra \cite{Peeters2007}, Redberry \cite{Poslavsky2015}, GammaMaP \cite{Kuusela2019}, and FieldsX \cite{Frob2021}. However, no tools currently exist for handling the arithmetic/canonicalization of spinor-indexed expressions. No broadly applicable tools exist to solve new multiplets or compute non-closure geometry. No general tools exist to compute non-closure geometry for arbitrary dimension, off-shell and on-shell.

The aim of the present work is to fill these gaps via a new Cadabra module we have named ``SusyPy.'' SusyPy includes a robust procedure for automatically canonicalizing tensor-spinor expressions, augmenting existing algorithms which effectively only treat canonicalization of pure Lorentz-tensor arithmetic. This point merits a bit of clarification: by tensor-spinor canonicalization, we do not mean the mere sorting of spinor indices on a tensor via the exploitation of the tensor's symmetries, which functions the same whether the indices are vectorial or spinorial. Rather, we refer to the correct manipulation of NW-SE convention, the spinor metric, and dimension-specific dualities and symmetries of spinor-indexed expressions to bring a complicated tensor-spinor formula expressed using spinor-indices (rather than spinor bilinears), as is ubiquitous in the genomics papers and subsequent work, to a standard (i.e., canonical) simplified form. For example, reducing \(C_{\eta \rho} (\gamma^{a b})_{\alpha \beta} C^{\gamma \beta} (\gamma_{c d})^{\alpha}{}_{\gamma} (\gamma^{c d})^{\rho}{}_{\sigma}\) to \(64 (\gamma^{a b})_{\eta \sigma}\) in 11D involves 1) sorting factors so that factors sharing dummy indices are adjacent, 2) exploiting the dimension-specific spinor-index symmetries of gamma matrices and the spinor metric to manipulate the expression into NW-SE convention, and 3) combining tensor-spinors and contracting the spinor metric appropriately. Previous tensor algebra systems have not handled these tasks, but our tensor-spinor algebra system does. We also develop a tool for Fierz-expanding spinor-indexed expressions using knowledge of only two of the expressions' spinor-indices, as is necessary when considering Fierz identities of expressions obtained from applying two supercovariant derivatives.

Building on this tensor-spinor symbolic algebra, we develop algorithms for solving general \(\mathcal{N} = 1\) multiplets and computing their properties. In particular, we provide two ``multiplet solvers'' that take a multiplet's supersymmetry transformation rules (in the form \eqref{susy_rule_ex}) and provide constraints on the coefficients obtained by requiring closure of the algebra. One follows a gamma-splitting procedure that filters out non-closure functions by analyzing gauge-invariance of expressions, and it provides a complete closure structure that can be used to easily compute non-closure geometry. The other uses Feynman propagators to filter out non-closure functions and is more widely applicable (succeeding on every tested multiplet). In addition, we provide a ``SUSY-invariant action solver'' that provides additional constraints on the SUSY-rules' and action's coefficients obtained by requiring supersymmetry-invariance of the action. Finally, we adapt the procedure of the first multiplet-solver to provide an algorithm that computes holoraumy. Our tools are expressly flexible, enabling the user to define a symbolic algebra workspace adapted to the dimension of interest and (when applicable) the choice of Clifford algebra representation, as well as to solve a multiplet whether off-shell or on-shell. Importantly, no choice of matrix representation of tensors is used at all in our system: we treat all tensor-spinor procedures completely symbolically.

As an application, we consider linearized (on-shell) 11D, \(\mathcal{N} = 1\) supergravity. Since the early superspace formulations of 11D supergravity \cite{Cremmer1980,Brink1980}, work on the multiplet has been focused on the off-shell SUSY problem, exploring how we can introduce a finite and minimal number of auxiliary fields into 11D supergravity so that the supersymmetry algebra closes off-shell, i.e., without application of the Bianchi identities and classical dynamics. Recent work on this multiplet \cite{Gates2019, Gates2020c} has made progress on the problem. It was found that a scalar superfield in 11D superspace contains the conformal graviton at the sixteenth level of the $\theta$-expansion of the superfield. In particular, this superfield could act as the supergravity prepotential leading to an off-shell 11D supergravity, and it places an upper bound on the number of component fields required: the maximum number of bosonic plus fermionic degrees is capped at 4,294,967,296. Since the Genomics project shares the broad aim of studying the off-shell problem, it is natural to assimilate 11D supergravity into the project. We accomplish just that in this paper. Using our module, we demonstrate how to solve the linearized 11D, \(\mathcal{N} = 1\) supergravity multiplet in a few lines of code, recasting the solution of \cite{Cremmer1978} in the formalism typical of the Genomics project and computing the non-closure geometry. We build on this solution to provide the first-ever computation of the holoraumy of 11D, \(\mathcal{N} = 1\) supergravity.

The plan of the rest of the paper is as follows. We begin with the physics results: in \S\ref{closure_sect}, we detail the aforementioned application to 11D supergravity. Then, in \S\ref{susypy}, we provide a comprehensive overview of our Cadabra module SusyPy for tensor-spinor symbolic algebra and the solution of multiplets, complete with explanations and pseudo-code for our algorithms, as well as examples to aid further work.

\numberwithin{equation}{subsection}
\section{11D Supergravity}\label{closure_sect}
\subsection{Solving the Multiplet}
In order to study non-closure geometry in the manner of \cite{Gates2020a}, we cast the linearized 11D, \(\mathcal{N} = 1\) supergravity multiplet in the following form akin to that of on-shell 4D supergravity as formulated in that paper.

\begin{definition}\label{11d_sugra_def}
The on-shell 11D, \(\mathcal{N} = 1\) \textit{supergravity multiplet} consists of a symmetric boson field \(h_{a b}\) -- the ``graviton'' -- a fermion field \(\Psi_{a}{}^{\alpha}\) -- the ``gravitino'' -- and an antisymmetric 3-tensor boson field \(A_{a b c}\) -- the ``three-form'' (or ``3-form''). The action of the supercovariant derivative (i.e., the linearized set of supersymmetry transformation rules) is given as follows.
\begin{subequations}\label{11d_sugra}
    \begin{align}
        {\rm D}_{\alpha}h_{a b} &= u \left(\gamma_{(a}\right)_{\alpha \beta} \Psi_{b)}{}^{\beta} \\
        {\rm D}_{\alpha}\Psi_{b}{}^{\beta} &= v ~ \omega_{b d e} \left(\gamma^{d e}\right)_{\alpha}{}^{\beta} + x \left(\gamma_{b} \gamma^{c d e f}\right)_{\alpha}{}^{\beta} \partial_{c}A_{d e f} + y \left(\gamma^{c d e f} \gamma_{b}\right)_{\alpha}{}^{\beta} \partial_{c}A_{d e f} \label{11d_sugra:2} \\
        {\rm D}_{\alpha}A_{b c d} &= z \left(\gamma_{[b c}\right)_{\alpha \beta} \Psi_{d]}{}^{\beta}
    \end{align}
\end{subequations}
One can check that these rules follow from index, symmetry, and engineering dimension considerations. Here, \({\rm D}_{\alpha}\) is the supercovariant derivative which effects the supersymmetry transformation, \(u, v, x, y, z\) are a priori unknown parameters,
\begin{equation}\label{spin_conn_def}
    \omega_{a b c} = \frac12 (c_{a b c} - c_{a c b} - c_{b c a})
\end{equation}
is the spin connection, and
\begin{equation}\label{anholonomy_def}
    c_{a b c} = \partial_{a}h_{b c} - \partial_{b}h_{a c}
\end{equation}
is the anholonomy. The corresponding (free-field) action takes the form
\begin{equation}\label{action}
    \begin{split}
        S &= \ell \int d^{11}x \left[-\frac14 c^{a b c} c_{a b c} + \frac12 c^{a b c} c_{c a b} + (c^{a}{}_{b}{}^{b})^2\right] 
        + \frac{m}{12} \int d^{11}x ~ \Psi_{a}{}^{\alpha} \left(\gamma^{a b c}\right)_{\alpha \beta}\partial_{b}\Psi_{c}{}^{\beta} \\
        &\quad+ \frac{n}{48} \int d^{11}x \left[\partial_{[a} A_{b c d]}\right] \left[\partial^{[a}A^{b c d]}\right],
    \end{split}
\end{equation}
where the first term (expressed quadratically using the anholonomy) corresponds to the graviton, the second to the gravitino, and the third to the three-form. Here, \(\ell, m, n\) are again a priori unknown parameters.
\end{definition}

Notice that up to a total derivative, the graviton term in our action \eqref{action} equals the fully contracted Riemann curvature scalar. This aligns with superfield considerations. The search for the supergravity prepotential superfield for 11D, \({\cal N} = 1\) superspace \cite{Gates2020c} uncovered that the scalar superfield contains the conformal graviton at its \(\theta\)-expansion's sixteenth level. If one wishes to include the degrees of freedom associated with Lorentz-rotations, the simplest route is to introduce a spinorial superfield \({\cal U}^{\a}\), which can include all of the graviton at the seventeenth level. The 4D analogue of the 11D supergravity superfield \({\cal U}^{\a}\) is not a spinor, but a bosonic vector superfield \({\cal U}^{a}\) \cite{Siegel1979}, and when one calculates the determinant of the supervielbein, the graviton portion of the action is exactly the three-term integrand seen in the first term of \eqref{action}, the reason being that \({\cal U}^{a}\) contains no field corresponding to the spin-connection. It turns out that the 4D, \({\cal N} = 2\) case is similar \cite{Grisaru1985,Belluci1989}, and it seems reasonable to conjecture that this extends to all minimal and irreducible higher-dimensional superfield supergravity (SFSG) theories and their SFSG prepotentials.

Now, we apply the multiplet solver tools to compute the various coefficients in definition \ref{11d_sugra_def}. The multiplet from \eqref{11d_sugra} can be entered in code as shown below.\footnote{We include at the end of the SUSY rules ``susy'' a substitution rule which records \eqref{anholonomy_def}.}
\begin{minted}{python}
>>> susy = r'D_{\alpha}(h_{a b}) -> u ((\Gamma_{a})_{\alpha \beta} (\Psi_{b})^{\beta} + (\Gamma_{b})_{\alpha \beta} (\Psi_{a})^{\beta}), D_{\alpha}((\Psi_{b})^{\beta}) -> 2 v \partial_{e}(h_{b d}) (\Gamma^{d e})_{\alpha}^{\beta} + x (\Gamma_{b} \Gamma^{c d e f})_{\alpha}^{\beta} \partial_{c}(A_{d e f}) + y (\Gamma^{c d e f} \Gamma_{b})_{\alpha}^{\beta} \partial_{c}(A_{d e f}), D_{\alpha}(A_{b c d}) -> 2 z (\Gamma_{b c})_{\alpha \beta} (\Psi_{d})^{\beta} - 2 z (\Gamma_{b d})_{\alpha \beta} (\Psi_{c})^{\beta} + 2 z (\Gamma_{c d})_{\alpha \beta} (\Psi_{b})^{\beta}, c_{a b c} -> \partial_{a}(h_{b c}) - \partial_{b}(h_{a c})'
\end{minted}
On-shell closure, i.e., the requirement that upon enforcing the equations of motion, we have\footnote{In 11D, our convention is to take \(c = 1\) in \eqref{closure_condition}.}
\begin{equation}\label{desired}
    \{{\rm D}_{\alpha}, {\rm D}_{\beta}\} = i \left(\gamma^{a}\right)_{\alpha \beta} \partial_{a}
\end{equation}
up to gauge transformations, gives us a first set of constraints on the coefficients. Note that the gauge transformation of the gravitino is of the form
\begin{equation}
    \delta_{G} \Psi_{a}{}^{\gamma} = \partial_{a} \zeta^{\gamma}
\end{equation}
for a gauge parameter \(\zeta^{\gamma}\), and the chosen basis for the Clifford algebra is that in \eqref{11d_basis_redux}. The code below applies ``SUSYSolve'' from \S\ref{multiplet_solver1}.
\begin{minted}{python}
>>> bosons = [Ex('h_{a b}'), Ex('A_{a b c}')]
>>> fermions = [Ex(r'(\Psi_{a})^{\gamma}')]
>>> gauge_transs = [Ex(r'\partial_{a}((\zeta)^{\gamma})')]
>>> basis = [Ex(r'C_{\alpha \beta}'), Ex(r'(\Gamma^{a})_{\alpha \beta}'), Ex(r'(\Gamma^{a b})_{\alpha \beta}'), Ex(r'(\Gamma^{a b c})_{\alpha \beta}'), Ex(r'(\Gamma^{a b c d})_{\alpha \beta}'), Ex(r'(\Gamma^{a b c d e})_{\alpha \beta}')]
>>> consts = ['u', 'v', 'x', 'y', 'z']
>>> indices = [r'_{\alpha}', r'_{\beta}']
>>> susy_solve(bosons, fermions, gauge_transs, susy, basis, consts, indices)
\end{minted}
Alternatively, we can use ``SUSYSolvePropagator'' from \S\ref{multiplet_solver2}, with the 11D Rarita-Schwinger propagator.
\begin{minted}{python}
>>> fermion_propagators = [rarita_schwinger_prop()]
>>> sol = susy_solve_propagator(bosons, fermions, fermion_propagators, susy, basis, consts, indices)
\end{minted}
A second set of constraints comes from SUSY-invariance of the action, obtained via ``MakeActionSUSYInv'' (\S\ref{action_solver}).
\begin{minted}{python}
>>> L = Ex(r'D_{\gamma}(l (-1/4 c^{a b c} c_{a b c} + 1/2 c^{a b c} c_{c a b} + c^{a}_{b}^{b} c_{a c}^{c}) + m (1/12) (\Psi_{a})^{\alpha} (\Gamma^{a b c})_{\alpha \beta} \partial_{b}((\Psi_c)^{\beta}) + n (3/4) (\partial_{a}(A_{b c d}) - \partial_{b}(A_{a c d}) + \partial_{c}(A_{a b d}) - \partial_{d}(A_{a b c})) (\partial^{a}(A^{b c d}) - \partial^{b}(A^{a c d}) + \partial^{c}(A^{a b d}) - \partial^{d}(A^{a b c})))')
>>> make_action_susy_inv(L, susy, ['u', 'v', 'x', 'y', 'z', 'l', 'm', 'n'], [r'\Psi'])
\end{minted}
The remaining constraints come from action normalization; see appendix \ref{11d_sugra_norm}. Finally, we introduce an arbitrary choice of overall scaling of the multiplet. The ultimate constraints are as follows.
\[
\renewcommand*{\leftterm}{m y}
\renewcommand*{\rightterm}{-12n z \quad}
\settowidth{\leftside}{\term{\leftterm}}
\settowidth{\rightside}{\term{\rightterm}}
\begin{array}{l}
    \left.
    \begin{aligned}
        uv &= -\frac{i}{8} \\[0.2em]
        xz &= \frac{i}{96} \\[0.2em]
        \makebox[\leftside][r]{\term{yz}} &= \makebox[\rightside][l]{\term{-\frac{i}{288}}}
    \end{aligned}
    \right\} \quad \text{from closure of the superalgebra} \\[3.8em]
    \left.
    \begin{aligned}
        \ell u &= \frac{1}{12} m v \\[1em]
        m x &= 36n z \\[1em]
        \leftterm &= \rightterm
    \end{aligned}
    \right\} \quad \text{from SUSY-invariance of the action} \\[4em]
    \left.
    \begin{aligned}
        \ell &= \frac14 \\[0.2em]
        \makebox[\leftside][r]{\term{n}} &= \makebox[\rightside][l]{\term{-\frac{1}{36}}}
    \end{aligned}
    \right\} \quad \text{from normalization of the action} \\[2.1em]
    \left.
    \begin{aligned}
        \makebox[\leftside][r]{\term{z}} &= \makebox[\rightside][l]{\term{\frac12}} \\
    \end{aligned}
    \right\} \quad \text{from a convenient (but arbitrary) scaling}
\end{array}
\]
It is an easy algebra exercise to confirm that (modulo the assumption that \(u,z\) are positive real) we find
\begin{equation}\label{coef_sol}
    \boxed{u = 1, \quad v = -\frac{i}{8}, \quad x = \frac{i}{48}, \quad y = -\frac{i}{144}, \quad z = \frac{1}{2}, \quad \ell = \frac14, \quad m = 24i, \quad n = -\frac{1}{36}.}
\end{equation}

\subsection{Non-Closure Geometry of 11D Supergravity}
Here, we summarize the non-closure geometry obtained by substituting \eqref{coef_sol} into the second output of ``SUSYSolve.'' The anticommutator of supercovariant derivatives applied to the supergravity fields takes the following form.
\begin{subequations}
    \begin{align}
        \{{\rm D}_\alpha, {\rm D}_\beta\}h_{a b} &= i \left(\gamma^{c}\right)_{\alpha \beta} \partial_{c}h_{a b} - \partial_{(a}\xi_{b) \alpha \beta} \\
        \{{\rm D}_\alpha, {\rm D}_\beta\}A_{a b c} &= i \left(\gamma^{d}\right)_{\alpha \beta} \partial_{d}A_{a b c} - \partial_{[a}\zeta_{b c] \alpha \beta} \\
        \{{\rm D}_\alpha, {\rm D}_\beta\}\Psi_{a}{}^{\gamma} &= i \left(\gamma^{b}\right)_{\alpha \beta} \partial_{b}\Psi_{a}{}^{\gamma} - \partial_{a}\epsilon_{\alpha \beta}{}^{\gamma} - Z_{a \alpha \beta}{}^{\gamma}
    \end{align}
\end{subequations}
The gauge transformations are
\begin{subequations}
    \begin{align}
        \xi_{b \alpha \beta} &= \frac{i}{2} \left(\gamma^{[1]}\right)_{\alpha \beta} h_{[1] b} \\
        \zeta_{b c \alpha \beta} &= \frac{i}{2} \left(\gamma^{[1]}\right)_{\alpha \beta} A_{[1] b c} + \frac{i}{2} \left(\gamma^{[1]}{}_{b}\right)_{\alpha \beta} h_{c [1]} \\
        \begin{split}
            \epsilon_{\alpha \beta}{}^{\gamma} &= \frac{27i}{32} \left(\gamma^{[1]}\right)_{\alpha \beta} \Psi_{[1]}{}^{\gamma} - \frac{i}{8} \left(\gamma_{[1]}\right)_{\alpha \beta} \left(\gamma^{[1] [\bar{1}]}\right)_{\eta}{}^{\gamma} \Psi_{[\bar{1}]}{}^{\eta} \\
            &\quad- \frac{i}{8} \left(\gamma^{[1] [\bar{1}]}\right)_{\alpha \beta} \left(\gamma_{[1]}\right)_{\eta}{}^{\gamma} \Psi_{[\bar{1}]}{}^{\eta} - \frac{3i}{64} \left(\gamma_{[2]}\right)_{\alpha \beta} \left(\gamma^{[2] [1]}\right)_{\eta}{}^{\gamma} \Psi_{[1]}{}^{\eta} \\
            &\quad- \frac{i}{768} \left(\gamma^{[4] [1]}\right)_{\alpha \beta} \left(\gamma^{[4]}\right)_{\eta}{}^{\gamma} \Psi_{[1]}{}^{\eta}.
        \end{split}
    \end{align}
\end{subequations}
Finally, and most importantly, the only non-closure function is fermionic and consists of the equation-of-motion terms
\begin{equation}
    Z_{a \alpha \beta}{}^{\gamma} = Z^{(1)}_{a \alpha \beta}{}^{\gamma \zeta} R_{\zeta} + Z^{(2)}_{a}{}^{b}{}_{\alpha \beta}{}^{\gamma \zeta} E_{b \zeta},
\end{equation}
where
\begin{subequations}\label{rarita_schwinger}
    \begin{align}
        R_{\alpha} &= \left(\gamma^{b c}\right)_{\alpha \eta} \partial_{b} \Psi_{c}{}^{\eta} \\
        E_{c \alpha} &= \left(\gamma^{b}\right)_{\alpha \eta} \partial_{[b} \Psi_{c]}{}^{\eta}.
    \end{align}
\end{subequations}
are forms of the Rarita-Schwinger equation and\footnote{Here, \(C_{\alpha \beta}\) is the spinor metric, i.e., the charge-conjugation matrix expressed with spinor indices. Note that \(C_{\alpha}{}^{\beta} = \delta_{\alpha}{}^{\beta}\).}
\begin{subequations}
    \begin{align}
        \begin{split}
            Z^{(1)}_{a \alpha \beta}{}^{\gamma \zeta} &= \frac{41i}{192} \left(\gamma_{a}\right)_{\alpha \beta} C^{\gamma \zeta} - \frac{5i}{192} \left(\gamma^{[1]}\right)_{\alpha \beta} \left(\gamma_{[1] a}\right)^{\gamma \zeta} \\
            &\quad- \frac{29i}{192} \left(\gamma_{a [1]}\right)_{\alpha \beta} \left(\gamma^{[1]}\right)^{\gamma \zeta} + \frac{7i}{384} \left(\gamma^{[2]}\right)_{\alpha \beta} \left(\gamma_{[2] a}\right)^{\gamma \zeta} \\
            &\quad- \frac{7i}{4608} \left(\gamma_{a [4]}\right)_{\alpha \beta} \left(\gamma^{[4]}\right)^{\gamma \zeta} \\
            &\quad+ \frac{i}{552960} \epsilon_{a [5] [\bar{5}]}\left(\gamma^{[5]}\right)_{\alpha \beta} \left(\gamma^{[\bar{5}]}\right)^{\gamma \zeta}
        \end{split} \\
        \begin{split}
            Z^{(2)}_{a}{}^{b}{}_{\alpha \beta}{}^{\gamma \zeta} &= \frac{5i}{48} \left(\gamma^{b}\right)_{\alpha \beta} \left(\gamma_{a}\right)^{\gamma \zeta} - \frac{3i}{8} \left(\gamma^{[1]}\right)_{\alpha \beta} \left(\gamma_{[1]}\right)^{\gamma \zeta} \delta_{a}{}^{b} \\
            &\quad- \frac{19i}{48} \left(\gamma_{a}{}^{b}\right)_{\alpha \beta} C^{\gamma \zeta} + \frac{i}{24} \left(\gamma^{[1] b}\right)_{\alpha \beta} \left(\gamma_{a [1]}\right)^{\gamma \zeta} \\
            &\quad+ \frac{i}{32} \left(\gamma^{[2]}\right)_{\alpha \beta} \left(\gamma_{[2]}\right)^{\gamma \zeta} \delta_{a}{}^{b} + \frac{i}{144} \left(\gamma_{a}{}^{b [3]}\right)_{\alpha \beta} \left(\gamma_{[3]}\right)^{\gamma \zeta} \\
            &\quad+ \frac{i}{1152} \left(\gamma^{[4] b}\right)_{\alpha \beta} \left(\gamma_{a [4]}\right)^{\gamma \zeta}.
        \end{split}
    \end{align}
\end{subequations}

\subsection{Holoraumy of 11D Supergravity}
Finally, we apply ``Holoraumy'' from \S\ref{holoraumy_alg} to find the holomoraumy.
\begin{minted}{python}
>>> fields = [Ex('h_{a b}'), Ex('A_{a b c}'), Ex(r'(\Psi_{a})^{\gamma}')]
>>> subs = Ex('u -> 1, v -> (-1/8) I, x -> (1/48) I, y -> (-1/144) I, z -> 1/2, l -> 1/4, m -> 24 I, n -> -1/36', False)
>>> holoraumy(fields, susy, basis, subs, indices)
\end{minted}
We find the bosonic holoraumy to be
\begin{align}
    \begin{split}
        [{\rm D}_\alpha, {\rm D}_\beta]h_{a b} &= -\frac{i}{4} \left(\gamma^{[1] [2]}{}_{(a|}\right)_{\alpha \beta} \partial_{[1]}A_{|b) [2]} + \frac{i}{18} \eta_{a b} \left(\gamma^{[1] [3]}\right)_{\alpha \beta} \partial_{[1]}A_{[3]} \\
        &\quad+ \frac{i}{2} \left(\gamma^{[1] [\bar{1}]}{}_{(a|}\right)_{\alpha \beta} \partial_{[1]}h_{|b) [\bar{1}]} - \partial_{(a}\xi'_{b) \alpha \beta}
    \end{split} \\
    \begin{split}
        [{\rm D}_\alpha, {\rm D}_\beta]A_{a b c} &= \frac{i}{4} \left(\gamma^{[1] [\bar{1}]}{}_{[a b|}\right)_{\alpha \beta} \partial_{[1]}h_{|c] [\bar{1}]} + \frac{i}{4} \left(\gamma^{[1] [\bar{1}]}{}_{[a|}\right)_{\alpha \beta} \partial_{[1]}A_{|b c] [\bar{1}]} \\
        &\quad+ \frac{i}{288} \epsilon_{a b c}{}^{[4] [1] [3]} \left(\gamma_{[4]}\right)_{\alpha \beta} \partial_{[1]}A_{[3]} - \partial_{[a}\zeta'_{b c] \alpha \beta}
    \end{split}
\end{align}
and the fermionic holoraumy to be
\begin{equation}
    \begin{split}
        [{\rm D}_\alpha, {\rm D}_\beta]\Psi_{a}{}^{\gamma} &= \frac{i}{2} \left(\gamma_{a}{}^{[1] [\bar{1}]}\right)_{\alpha \beta} \partial_{[1]} \Psi_{[\bar{1}]}{}^{\gamma} + \frac{i}{3} \left(\gamma_{a}{}^{[1] [\bar{1}] [\bar{\bar{1}}]}\right)_{\alpha \beta} \left(\gamma_{[1]}\right)_{\eta}{}^{\gamma} \partial_{[\bar{1}]} \Psi_{[\bar{\bar{1}}]}{}^{\eta} \\
        &\quad+ \frac{i}{8} \left(\gamma^{[2] [1]}\right)_{\alpha \beta} \left(\gamma_{[2]}\right)_{\eta}{}^{\gamma} \partial_{[1]} \Psi_{a}{}^{\eta} + \frac{i}{12} \left(\gamma^{[2] [1] [\bar{1}]}\right)_{\alpha \beta} \left(\gamma_{a [2]}\right)_{\eta}{}^{\gamma} \partial_{[1]} \Psi_{[\bar{1}]}{}^{\eta} \\
        &\quad+ \frac{i}{24} \left(\gamma^{[3] [1]}\right)_{\alpha \beta} \left(\gamma_{[3]}\right)_{\eta}{}^{\gamma} \partial_{[1]} \Psi_{a}{}^{\eta} - \partial_{a} \epsilon'_{\alpha \beta}{}^{\gamma} - \mathscr{Z}_{a \alpha \beta}{}^{\gamma}.
    \end{split}
\end{equation}
The gauge transformations are
\begin{align}
    \xi'_{b \alpha \beta} &= \frac{i}{12} \left(\gamma_{b}{}^{[3]}\right)_{\alpha \beta} A_{[3]} \\
    \zeta'_{b c \alpha \beta} &= \frac{i}{4} \left(\gamma^{[2]}{}_{b}\right)_{\alpha \beta} A_{c [2]} \\
    \begin{split}
        \epsilon'_{\alpha \beta}{}^{\gamma} &= -\frac{5i}{32} C_{\alpha \beta} \left(\gamma^{[1]}\right)_{\eta}{}^{\gamma} \Psi_{[1]}{}^{\eta} + \frac{5i}{64} \left(\gamma^{[2] [1]}\right)_{\alpha \beta} \left(\gamma_{[2]}\right)_{\eta}{}^{\gamma} \Psi_{[1]}{}^{\eta} \\
        &\quad- \frac{i}{96} \left(\gamma_{[3]}\right)_{\alpha \beta} \left(\gamma^{[3] [1]}\right)_{\eta}{}^{\gamma} \Psi_{[1]}{}^{\eta} + \frac{i}{32} \left(\gamma^{[3] [1]}\right)_{\alpha \beta} \left(\gamma_{[3]}\right)_{\eta}{}^{\gamma} \Psi_{[1]}{}^{\eta} \\
        &\quad- \frac{i}{768} \left(\gamma_{[4]}\right)_{\alpha \beta} \left(\gamma^{[4] [1]}\right)_{\eta}{}^{\gamma} \Psi_{[1]}{}^{\eta}.
    \end{split}
\end{align}
Finally, the off-shell equation-of-motion terms are
\begin{equation}
    \mathscr{Z}_{a \alpha \beta}{}^{\gamma} = \mathscr{Z}^{(1)}_{a \alpha \beta}{}^{\gamma \zeta} R_{\zeta} + \mathscr{Z}^{(2)}_{a}{}^{b}{}_{\alpha \beta}{}^{\gamma \zeta} E_{b \zeta},
\end{equation}
where
\begin{align}
    \begin{split}
        \mathscr{Z}^{(1)}_{a \alpha \beta}{}^{\gamma \zeta} &= \frac{3i}{64} C_{\alpha \beta} \left(\gamma_{a}\right)^{\gamma \zeta} - \frac{3i}{128} \left(\gamma_{a [2]}\right)_{\alpha \beta} \left(\gamma^{[2]}\right)^{\gamma \zeta} \\
        &\quad- \frac{i}{384} \left(\gamma^{[3]}\right)_{\alpha \beta} \left(\gamma_{[3] a}\right)^{\gamma \zeta} + \frac{5i}{1152} \left(\gamma_{a [3]}\right)_{\alpha \beta} \left(\gamma^{[3]}\right)^{\gamma \zeta} \\
        &\quad- \frac{7i}{4608} \left(\gamma^{[4]}\right)_{\alpha \beta} \left(\gamma_{[4] a}\right)^{\gamma \zeta}
    \end{split} \\
    \begin{split}
        \mathscr{Z}^{(2)}_{a}{}^{b}{}_{\alpha \beta}{}^{\gamma \zeta} &= \frac{i}{4} C_{\alpha \beta} C^{\gamma \zeta} \delta_{a}{}^{b} + \frac{i}{6} \left(\gamma_{a}{}^{b [1]}\right)_{\alpha \beta} \left(\gamma_{[1]}\right)^{\gamma \zeta} \\
        &\quad+ \frac{i}{96} \left(\gamma^{[2] b}\right)_{\alpha \beta} \left(\gamma_{a [2]}\right)^{\gamma \zeta} + \frac{5i}{96} \left(\gamma_{a}{}^{b [2]}\right)_{\alpha \beta} \left(\gamma_{[2]}\right)^{\gamma \zeta} \\
        &\quad- \frac{i}{72} \left(\gamma^{[3] b}\right)_{\alpha \beta} \left(\gamma_{a [3]}\right)^{\gamma \zeta} - \frac{i}{192} \left(\gamma^{[4]}\right)_{\alpha \beta} \left(\gamma_{[4]}\right)^{\gamma \zeta} \delta_{a}{}^{b} \\
        &\quad+ \frac{i}{96} \left(\gamma^{[3]}\right)_{\alpha \beta} \left(\gamma_{[3]}\right)^{\gamma \zeta} \delta_{a}{}^{b}.
    \end{split}
\end{align}
Notice that no analogue to the electromagnetic-duality rotations of \cite{Gates2020a} is present in the 11D, \(\mathcal{N} = 1\) supergravity holoraumy above, providing the second counterexample (after 10D, \(\mathcal{N} = 1\) super Maxwell theory \cite{Gates2022}) to the conjectured ubiquity of these rotations.

\section{SusyPy: A Supersymmetry Multiplet Solver}\label{susypy}
\subsection{Basic Use}\label{using_susypy}
In this chapter, we expound and demonstrate the symbolic algebra algorithms which we utilize for solving the 11D supergravity multiplet; we have christened these algorithms collectively ``SusyPy.''\footnote{SusyPy = \textbf{Su}per\textbf{sy}mmetry \textbf{Py}thon.} SusyPy is a Python module built on Cadabra for the purpose of handling explicit spinor-index arithmetic/canonicalization and crucial supersymmetry calculations. Cadabra has its own interpreted language ``cdb,'' but its functions are also accessible via Python. Since the authors are devout Python-ers, we have elected to continue purely in Python. However, documentation on Cadabra's pure Python interface is rather scanty, so before moving forward, we briefly consider its use in Python concomitantly with SusyPy. After downloading both software, Cadabra and SusyPy can be imported via the code below.
\begin{minted}{python}
>>> import cadabra2 as cdb
>>> import susypy as susy
\end{minted}
Input into Cadabra is in the form of pseudo-\LaTeX{} strings into \mintinline{python}{Ex()} wrappers. Just like in \LaTeX, sums are written via ``+'' and products via proximity, and Greek letters can even be entered by writing their name after a backslash as in \LaTeX. In addition, covariant indices are written with a simple ``\_'' and contravariant indices with a ``\^{}''. The \mintinline{python}{Ex()} wrappers themselves can also be added and multiplied after they are defined. The code below creates two Cadabra expressions and adds them. Cadabra's output is printed as unicode.
\begin{minted}{python}
>>> ex1 = cdb.Ex(r'(A_{a} B_{b} + E_{a b}) \Psi^{c}') 
>>> ex2 = cdb.Ex(r'G_{a} \Theta_{b}^{c}')
>>> ex = ex1 + ex2
>>> ex
'(A_{a} B_{b} + E_{a b}) Ψ^{c} + G_{a} ϴ_{b}^{c}'
\end{minted}
An important Cadabra tool is substitution, which we will employ frequently in the algorithms later in this chapter. A series of substitution rules is entered as a comma-separated string of individual tensor substitution rules, with each substitution indicated by an arrow ``->'', inputted into an \mintinline{python}{Ex()} wrapper with the optional parameter ``False'' signifying that it is not an ordinary tensor expression. Entering an expression together with a substitution rule into Cadabra's function \mintinline{python}{substitute()} applies the latter to the former, as in the below example.
\begin{minted}{python}
>>> sub = cdb.Ex(r'A_{a} -> G_{a}^{d} H_{d}, \Theta_{e}^{f} -> \Psi_{e}^{f}', False)
>>> cdb.substitute(ex, sub)
'(G_{a}^{d} H_{d} B_{b} + E_{a b}) Ψ^{c} + G_{a} Ψ_{b}^{c}'
\end{minted}
Notice that indices in the substitution rule dynamically adapt to the context of the expression. Having briefly summarized the use of Cadabra, in the rest of this section, we illustrate the basic use of SusyPy.

\subsubsection{Setting up the Environment}
One fixes a computational environment in SusyPy via the function \mintinline{python}{susy_env()}, which takes as input the following information:\footnote{A call of \mintinline{python}{susy_env()} with none of the optional parameters entered yields the 11D Clifford algebra with the conventions in Appendix \ref{conventions}, since this is the environment needed for the 11D supergravity calculations in the previous chapter.} the dimension \(D\), the array ``lorentz\_indices'' of Lorentz indices with which one intends to work (which must be of length at least the dimension),\footnote{These Lorentz indices are also used for internal purposes. The requirement that the length of ``lorentz\_indices'' be at least \(D\) is in order to allow the construction of Levi-Civita tensors.} the array ``spinor\_indices'' of spinor indices with which one intends to work, a length-2 array ``desired\_syms'' which determines the spinor-index symmetries of the gamma matrices, and a representation sign ``rep'' which determines the sign of the gamma-matrix reduction formula.\footnote{``desired\_syms'' is only a relevant input in even dimensions, for the symmetries have fixed values in odd dimensions. Similarly, ``rep'' is only relevant in odd dimensions, for there is only a single representation in even dimensions.} (See \S\ref{traces_and_dimension} for explanations of these parameters.) The line of code below creates a 4D environment with the Lorentz indices \(a, b, c, d\) and spinor indices \(\alpha, \beta, \gamma, \eta\), values \((t_0, t_1) = (1, 1)\) which make the spinor metric and 1-index gamma matrix antisymmetric, and the negative Clifford algebra representation (an input which is ignored because of the dimension but included for completeness). Note that whenever \mintinline{python}{susy_env()} is called, its output must be stored in a variable named \mintinline{python}{__cdbkernel__} in order for subsequent arithmetic to be run; this derives from an idiosyncracy of Cadabra.
\begin{minted}{python}
>>> __cdbkernel__ = susy.susy_env(D = 4, lorentz_indices=['a', 'b', 'c', 'd'], spinor_indices=[r'\alpha', r'\beta', r'\gamma', r'\eta'], desired_syms=[1,1], rep=-1)
\end{minted}
The environment created by \mintinline{python}{susy_env()} defines a variety of Cadabra objects, tabulated in Table \ref{tab:table1}, including two new tensors corresponding to the spinor metric \(C_{\alpha \beta}\) and highest-rank element \(\gamma_{*}\) of the Clifford algebra.
\begin{table}[ht!]
    \begin{center}
        \begin{tabular}{| c | c | c |}
            \hline
            \textbf{SusyPy Name} & \textbf{Cadabra Class} & \textbf{Use} \\
            \hline
            \hline
            ``\textbackslash Gamma'' & GammaMatrix & Dirac Gamma Matrix \\ \hline
            ``\textbackslash delta'' & KroneckerDelta & Kronecker Delta \\  \hline
            ``\textbackslash epsilon'' & EpsilonTensor & Levi-Civita Symbol \\  \hline
            ``D'' & Derivative & Supercovariant Derivative \\ \hline
            ``\textbackslash partial'' & PartialDerivative & Partial Derivative \\ \hline
            ``I'' & ImaginaryI & Imaginary Unit \\ \hline
            ``Tr'' & Trace & Trace Function \\ \hline
            ``C'' & N/A & Spinor Metric \\ \hline
            ``\textbackslash Gamma'\,'' & N/A & Highest-Rank Element \\ [0.5ex]
            \hline
        \end{tabular}
        \caption{Tensor objects created by \mintinline{python}{susy_env()}.}
        \label{tab:table1}
    \end{center}
\end{table}

The below code verifies that the dimension, gamma-matrix symmetry, and pools of indices are as selected.
\begin{minted}{python}
>>> ex1 = cdb.Ex(r'(\Gamma_{a})_{\alpha}^{\beta} (\Gamma^{b})_{\beta}^{\alpha}')
>>> susy.evaluate(ex1)
'4δ_{a}^{b}'
>>> ex2 = cdb.Ex(r'(\Gamma^{a})_{\beta \alpha}')
>>> cdb.canonicalise(ex2)
'-\indexbracket(Γ^{a})_{α β}'
>>> ex3 = cdb.Ex(r'(A_{d} B^{d})_{\gamma}^{\gamma}')
>>> cdb.rename_dummies(ex3)
'\indexbracket(A_{a} B^{a})_{α}^{α}'
\end{minted}

\subsubsection{Simplifying Tensor Expressions}
At its foundation, SusyPy is a spinor-index arithmetic tool; all of the procedures which form the multiplet solvers are constructed on the base of explicit spinor-index canonicalization and multiplication adhering to the NW-SE convention. To distinguish spinor and vector indices, SusyPy latches on a Cadabra feature which interprets parentheses carrying indices as objects called ``indexbrackets''; the object inside that indexbracket is interpreted as an argument.\footnote{It was possibly the hope for Cadabra that indexbrackets would eventually be used for this purpose. A very rudimentary Cadabra2 \cite{Peeters2018b} function \mintinline{python}{combine()} enables the combination of indexbrackets, but in a na\"ive fashion which pays no attention to NW-SE convention and has no utility for actual spinor-index arithmetic.} We write a tensor-spinor with the underlying Lorentz tensor surrounded by parentheses which carry the spinor indices, e.g., SusyPy will interpret the Cadabra object below as a tensor-spinor \(\Psi_{a \eta}\), with \(a\) a Lorentz index and \(\eta\) a spinor index.
\begin{minted}{python}
>>> cdb.Ex(r'''(\Psi_{a})_{\eta}''')
\end{minted}
The critical function powering SusyPy's handling of spinor indices is \mintinline{python}{spinor_combine()}, which maximally combines and canonicalizes indexbrackets in a manner adherent to the NW-SE convention, as seen in the below example.\footnote{N.B. The spinor metric with its first index covariant and its second contravariant is often written simply as \(\delta_{\alpha}{}^{\beta}\). This is not permissable when using SusyPy, as Kronecker deltas are considered pure Lorentz tensors by the algorithm. One must use \(C_{\alpha}{}^{\beta}\). Incidentally, the spinor metric is the only tensor in SusyPy whose indices are interpreted as spinor indices without the tensor having an indexbracket.}
\begin{minted}{python}
>>> ex = cdb.Ex(r'(\Gamma_{a})_{\alpha \beta} (\Gamma^{b})^{\beta \gamma} C_{\gamma \eta}')
>>> susy.spinor_combine(ex)
'-\indexbracket(Γ_{a} Γ^{b})_{α η}'
\end{minted}
We now consider our broad simplification function \mintinline{python}{evaluate()}. This function is built to accomplish several goals with respect to an inputted expression:
\begin{enumerate}
    \item apply \mintinline{python}{spinor_combine()} in order to maximally combine indexbrackets and order chains of multiplied indexbrackets in lexicographic order (i.e., canonicalize the spinor-index structure) to the greatest extent permissable by symmetries,
    \item multiply all gamma matrices using Cadabra's Lorentz-tensor simplification sequence,
    \item simplify all traces implicit in products of tensor-spinors, all derivatives, and all Kronecker-delta contractions,
    \item draw all constants and pure Lorentz tensors (e.g., Kronecker deltas and Levi-Civita tensors) out of derivatives, Fourier transforms, and indexbrackets, and
    \item flatten the expression into the sum of products of single-term factors.
\end{enumerate}
Below is an example of its application.
\begin{minted}{python}
>>> ex = Ex(r'(\Gamma_{a})_{\alpha \beta} (I \Gamma^{b})^{\beta \gamma} C_{\gamma \eta} (5 \delta_{b}{}^{a} \Psi_{d})_{\zeta}')
>>> susy.evaluate(ex)
'-55I C_{α η} \indexbracket(Ψ_{d})_{ζ}'
\end{minted}

\subsubsection{Compatibility with Cadabra Functions}
Virtually all functions of SusyPy take in expression input as Cadabra \mintinline{python}{Ex} objects, as well as give expression output as Cadabra \mintinline{python}{Ex} objects.\footnote{Also, like Cadabra functions, virtually all functions of SusyPy actually modify the inputted expressions rather than simply outputting the result. This means that the output of a SusyPy function need not be stored in a new variable.} This means that expressions outputted by Cadabra functions can be inputted into SusyPy functions and vice versa without the need to reenter the expression. In addition, SusyPy's algorithms work well with Cadabra's ordinary means of defining new tensors and imparting them properties. For example, in the 4D code below, three tensors \(A_{a}, \lambda_{\gamma}, d\) (the last one is actually a pseudo-scalar) are defined via Cadabra's \mintinline{python}{Depends} to be ``differentiable'' by the supercovariant derivative and partial derivative defined in \mintinline{python}{susy_env()}, or more precisely, to have nonzero derivatives. In addition, \(\lambda_{\gamma}\) is defined to anticommute with the supercovariant derivative, reflecting that both are fermionic. Notice the Cadabra idiosyncrasy that properties assigned to a tensor with spinor indices must be explicitly assigned both to the underlying Lorentz tensor and to the tensor within indexbrackets; this arises from the fact that indexbrackets are viewed as somewhat distinct objects. Also, notice that we must explicitly write out \mintinline{python}{'\indexbracket'} rather than simply use parentheses.
\begin{minted}{python}
>>> __cdbkernel__ = susy.susy_env(D = 4, lorentz_indices=['a', 'b', 'c', 'd'], spinor_indices=[r'\alpha', r'\beta', r'\gamma'])
>>> cdb.Depends(Ex(r'''A{#}'''), Ex(r'''D{#}, \partial{#}'''))
>>> cdb.Depends(Ex(r'''\lambda{#}'''), Ex(r'''D{#}, \partial{#}'''))
>>> cdb.Depends(Ex(r'''\indexbracket{\lambda{#}}{#}'''), Ex(r'''D{#}, \partial{#}'''))
>>> cdb.AntiCommuting(Ex(r'''\lambda, D{#}'''))
>>> cdb.AntiCommuting(Ex(r'''\indexbracket{\lambda{#}}{#}, D{#}'''))
>>> cdb.Depends(Ex(r'''d'''), Ex(r'''D{#}, \partial{#}'''))
\end{minted}

Now that we have adumbrated the elementary use of SusyPy, the remainder of the chapter will explain the core algorithms. Principal among these are the algorithms obtaining constraints from SUSY algebra closure (\S\ref{multiplet_solver1} and \S\ref{multiplet_solver2}), the algorithm obtaining constraints from SUSY-invariance of the action will be presented in \S\ref{action_solver}, and the algorithm for computing holoraumy (\S\ref{holoraumy_alg}).

From here on out, the results of all code will be typeset, and occasionally modifications will be done implicitly where such modifications offer a tangible improvement in clarity and only a small algebraic jump from the code's actual output. In addition, for brevity, we omit calls to \mintinline{python}{susy_env()} and explicit tensor definitions Nevertheless, we \textit{do} offer the full code corresponding to each computational calculation on GitHub at \href{https://github.com/IsaiahBHilz/susypy/tree/main/paper}{https://github.com/IsaiahBHilz/susypy/tree/main/paper}.

\subsection{Procedures for Dimension-Specific Symmetries and Substitutions}\label{traces_and_dimension}
The first algorithms which must be considered in setting the foundation for spinor-index arithmetic are those that treat dimension-dependent gamma-matrix symmetries. (3.63) in \cite{Freedman2013} gives that
\begin{equation}\label{t_sym_rel}
    \left(\gamma_{a_1 a_2 \cdots a_r}\right)_{\alpha \beta} = -t_r \left(\gamma_{a_1 a_2 \cdots a_r}\right)_{\beta \alpha},
\end{equation}
where \(t_r\) is the \(r\)th element of a sequence defined by fixing \(t_0, t_1 \in \{-1, 1\}\) and setting \(t_{r+2} = -t_r\). \(t_0, t_1\) depend on the dimension \(D\) and satisfy\footnote{Note that (3.106) of \cite{Freedman2013} erroneously includes a factor of \(2^{m-1}\) in the \(t_1\) factor of \eqref{sym_trig}.}
\begin{equation}\label{sym_trig}
    t_0 \cos\left(\frac{m \pi}{2}\right) + t_1 \sin\left(\frac{m \pi}{2}\right) = -1,
\end{equation}
where \(m = \lfloor D/2 \rfloor\). Hence, either \(m \pi/2\) is an integer multiple of \(\pi\) and \(t_0 = -\cos(m \pi/2)\), or \(m \pi/2\) is a half-integer multiple of \(\pi\) and \(t_1 = -\sin(m \pi/2)\). That is, one of \(t_0, t_1\) is fixed by \eqref{sym_trig}. If \(D\) is odd, then the other can be fixed by the fact that \(\gamma^{2m}\) is constructed by adjoining the highest-rank element \(\pm \gamma_{*}\) in the (\(D-1\))-dimensional Clifford algebra. Since \(\gamma^{2m}\) has the symmetry \(t_1\) while \(\pm \gamma_{*}\) has the symmetry \(t_{2m}\) (\(\gamma_{*}\) is roughly proportional to a \(2m\)-indexed gamma matrix), it follows that
\begin{equation}\label{1eq2m}
    t_1 = t_{2m} = (-1)^m t_0.
\end{equation}
Hence, if \(D\) is odd and \(t_0\) is fixed by \eqref{sym_trig}, then \(t_1 = (-1)^m t_0\), while if \(t_1\) is fixed by \eqref{sym_trig}, then \(t_0 = (-1)^m t_1\). If, on the other hand, \(D\) is even, then \eqref{1eq2m} does not hold and the value of the unfixed symmetry is arbitrary. In general, though, for even \(D\), one typically chooses \(t_0, t_1\) equal if \(D \equiv 0 \Mod 8\) and opposite otherwise.

The algorithm ``GenSyms'' (see Algorithm \ref{alg:gensyms}) uses the foregoing considerations to determine \(t_0, t_1\).\footnote{In the program, ``GenSyms'' is manifested as a function \mintinline{python}{gen_syms()}. In this paper, the convention is that inline code represents the Python function while double-quoted camel-script text represents the algorithm abstractly.} Optional desired values for \(t_0, t_1\) can be provided as input. If the desired values are impossible, an error is raised. Note that \(\textbf{int}(D \% 8 \ne 0)\) represents the boolean condition \(D \not\equiv 0 \Mod 8\) viewed as a numeric value 0 or 1, so that \((-1)^{\textbf{int}(D \% 8 \ne 0)}\) is short-hand for
\begin{equation}
    (-1)^{\textbf{int}(D \% 8 \ne 0)} = \begin{cases} -1, & \textnormal{if } D \not\equiv 0 \Mod 8 \\ 1, & \textnormal{if } D \equiv 0 \Mod 8. \end{cases}
\end{equation}

\begin{algorithm}[ht!]
    \caption{Find Values of \(t_0, t_1\) per Dimension}\label{alg:gensyms}
    \begin{algorithmic}[1]
        \Require \textit{desired\_syms} equals [0,0] (i.e., null), [1,1], [1,-1], [-1,1], or [-1,-1]
        \Function{GenSyms}{\(D\), \textit{desired\_syms}[0..1]}

            \State \(m \gets \lfloor D/2 \rfloor\)
            \State \textit{syms} \(\gets\) [\(-\cos(m \pi/2)\), \(-\sin(m \pi/2)\)]

            \If{\(D\) is odd \& \textit{syms}[0] = 0} \Comment{If \(D\) is odd, use \eqref{1eq2m}.}
                \State \textit{syms}[0] \(\gets (-1)^m~\cdot\) \textit{syms}[1]
            \ElsIf{\(D\) is odd \& \textit{syms}[1] = 0}
                \State \textit{syms}[1] \(\gets (-1)^m~\cdot\) \textit{syms}[0]
            \ElsIf{\(D\) is even \& \textit{syms}[0] = 0 \& \textit{syms}[1] = \textit{desired\_syms}[1]} \Comment{If \(D\) is even, try desired values.}
                \State \textit{syms}[0] \(\gets\) \text{desired\_syms}[0]
            \ElsIf{\(D\) is even \& \textit{syms}[1] = 0 \& \textit{syms}[0] = \textit{desired\_syms}[0]}
                \State \textit{syms}[1] \(\gets\) \text{desired\_syms}[1]
            \ElsIf{\(D\) is even \& \textit{syms}[0] = 0 \& \textit{desired\_syms} = [0,0]} \Comment{If desired values null, use conventions.}
                \State \textit{syms}[0] \(\gets (-1)^{\textbf{int}(D \% 8 \ne 0)}~\cdot\) \textit{syms}[1]
            \ElsIf{\(D\) is even \& \textit{syms}[1] = 0 \& \textit{desired\_syms} = [0,0]}
                \State \textit{syms}[1] \(\gets (-1)^{\textbf{int}(D \% 8 \ne 0)}~\cdot\)\textit{syms}[0]
            \EndIf

            \If{\textit{syms} \(\ne\) \textit{desired\_syms} \& \textit{desired\_syms} \(\ne\) [0,0]} \Comment{If desired symmetries impossible, raise error.}
                \State \textbf{raise error} because \textit{desired\_syms} must equal \textit{syms} if \textit{desired\_syms} is not null
            \EndIf

            \State \Return \textit{syms}
        \EndFunction
    \end{algorithmic}
\end{algorithm}

The values \(t_r\), \(r > 1\), can be determined via \(t_{r+2} = -t_r\), and \eqref{t_sym_rel} shows that this determines the gamma-matrix symmetries. The algorithm ``AntisymmetricGammas'' (see Algorithm \ref{alg:agamma}) leverages these results to list all of the values \(r\) from 1 to \(2 \lfloor D/2 \rfloor\)\footnote{Gamma matrices can have at most \(2 \lfloor D/2 \rfloor\) indices in dimension \(D\) (while maintaining properties like tracelessness).} for which \(r\)-index gamma matrices are antisymmetric, given the dimension \(D\) and a length-two array ``syms'' containing \(t_0, t_1\) (in that order).

\begin{algorithm}[ht!]
    \caption{Tabulate Antisymmetric Gamma Matrices}\label{alg:agamma}
    \begin{algorithmic}[1]
        \Function{AntisymmetricGammas}{\(D\), \textit{syms}[0..1]}

            \State \textit{antisym\_gammas} \(\gets\) [ ]
            \State \(n \gets 2 \lfloor D/2 \rfloor\)

            \For{\(i \in 1..n\)}
                \If{(\(i \equiv 0 \Mod{4}\) \& \textit{syms}[0] = 1) or (\(i \equiv 1 \Mod{4}\) \& \textit{syms}[1] = 1) or (\(i \equiv 2 \Mod{4}\) \& \textit{syms}[0] = -1) or (\(i \equiv 3 \Mod{4}\) \& \textit{syms}[1] = -1)}
                    \State \textbf{append} \(i\) to \textit{antisym\_gammas}
                \EndIf
            \EndFor

            \State \Return \textit{antisym\_gammas}

        \EndFunction
    \end{algorithmic}
\end{algorithm}

Besides the gamma matrices, two other tensors of import are the highest-rank element \(\gamma_{*}\) of the Clifford algebra (for \(D\) even) and the spinor metric (charge conjugation matrix) \(C_{\alpha \beta}\) which is used to raise and lower spinor indices. \(\gamma_{*}\) is antisymmetric if and only if \(t_{2m} = 1\), and \(C_{\alpha \beta}\) is antisymmetric if and only if \(t_0 = 1\). The symmetries of these two tensors and those of the gamma matrices can be employed to discern the symmetry of \(\gamma^{a_1 a_2 \cdots a_r} \gamma_{*}\) for a particular \(r\), which proves necessary in spinor-index canonicalization for even \(D\). It is easy to verify that
\begin{equation}
    \left(\gamma^{a_1 a_2 \cdots a_r} \gamma_{*}\right)_{\alpha \beta} = (-t_{2m}) (-t_0) (-t_r) (-1)^r \left(\gamma^{a_1 a_2 \cdots a_r} \gamma_{*}\right)_{\beta \alpha},
\end{equation}
so a negative is introduced in the symmetry equation above for each of the following conditions: \(t_{2m} = 1\), \(t_0 = 1\), \(t_r = 1\), and \(r \equiv 1 \Mod 2\). The algorithm ``isGammaGammaStarProdSym'' (see Algorithm \ref{alg:iggsps}) uses this to evaluate the boolean expression that \(\gamma^{a_1 a_2 \cdots a_r} \gamma_{*}\) is symmetric rather than antisymmetric.\footnote{We take into account the symmetry of \(\gamma_{*}\) via \(t_{2m}\) even if \(D\) is odd, as \(\gamma_{*} = \pm \gamma^{2m}\) by convention for odd \(D\), so it is well-defined and carries the same symmetry.}

\begin{algorithm}[ht!]
    \caption{Check whether the Product of an \(r\)-Indexed Gamma Matrix with \(\gamma_{*}\) is Symmetric}\label{alg:iggsps}
    \begin{algorithmic}[1]
        \Function{isGammaGammaStarProdSym}{\(r\), \(D\), \textit{syms}[0..1]}

            \State \textit{antisym\_gammas} = \textbf{AntisymmetricGammas}(\(D\), \textit{syms})
            \State \(C \gets 1\)

            \If{\(2\lfloor D/2 \rfloor \in\) \textit{antisym\_gammas}} \Comment{i.e., if \(\gamma_*\) is antisymmetric}
                \State \(C \gets -C\)
            \EndIf
            \If{\textit{syms}[0] = 1} \Comment{i.e., if C is antisymmetric}
                \State \(C \gets -C\)
            \EndIf
            \If{\(r \in\) \textit{antisym\_gammas}} \Comment{i.e., if the \(r\)-index \(\gamma\) is antisymmetric}
                \State \(C \gets -C\)
            \EndIf
            \If{\(r \equiv 1 \Mod{2}\)} \Comment{i.e., if \(\gamma_*\) and the \(r\)-index \(\gamma\) anticommute}
                \State \(C \gets -C\)
            \EndIf

            \State \Return \(C = 1\)

        \EndFunction
    \end{algorithmic}
\end{algorithm}

Besides symmetries, other dimension-specific identities of import are the reductions of gamma matrices with more than \( m = \lfloor D/2 \rfloor\) indices to those with at most \(m\) indices (multiplied by \(\gamma_{*}\) if \(D\) is even). In particular, it follows from (3.41) and (3.42) in \cite{Freedman2013} that
\begin{subequations}
    \begin{align}
        \gamma_{\pm}^{a_1 \cdots a_r} &= \pm \frac{i^{m+1}}{(D-r)!} \epsilon^{a_1 \cdots a_D} \gamma_{\pm~a_D \cdots a_{r+1}}, && D~\textnormal{odd} \label{gamma_red:1} \\
        \gamma^{a_1 \cdots a_r} &= (-1)^r \frac{i^{m+1}}{(D-r)!} \epsilon^{a_1 \cdots a_D} \gamma_{a_D \cdots a_{r+1}} \gamma_{*}, \quad D~\textnormal{even} \label{gamma_red:2},
    \end{align}
\end{subequations}
where the \(\pm\) in the identity for odd \(D\) reflects that there are two choices of representation for the Clifford algebra in dimension \(D\) odd. The algorithm ``GenSubs'' (see Algorithm \ref{alg:gensubs}) takes the dimension and the representation \(+1\) or \(-1\) (only relevant for \(D\) odd) as input and generates the substitution rule for each \(r\) (strictly, if \(D\) is odd) greater than \(m\).\footnote{These identities are intended for use in reducing gamma matrices via repeated substitutions, so it is desirable that the substitutions be carried out in descending number of indices from \(2m\) to \(m+1\). Hence, it is preferable to prepend rather than append the substitution rules to the list as \(r\) increases.}

\begin{algorithm}[ht!]
    \caption{Generate Substitution Rules for Gamma Matrices with More than \(\lfloor D/2 \rfloor\) Indices}\label{alg:gensubs}
    \begin{algorithmic}[1]
        \Function{GenSubs}{\(D\), \textit{rep}}

            \State \(m \gets \lfloor D/2 \rfloor\)
            \State \textit{sub\_exs} \(\gets\) [ ]

            \For{\(r \in (m+1)..2m\)}
                \State \textit{exstr} \(\gets\) \(\gamma^{a_1 \cdots a_r}\)
                \State \textit{res\_str\_base} \(\gets\) \(\frac{i^{m+1}}{(D-r)!} \epsilon^{a_1 \cdots a_r a_{r+1} \cdots a_D} \gamma_{a_D \cdots a_{r+1}}\)
                \If{\(D \equiv 0 \Mod{2}\)}:
                    \State \textit{sub\_rule} \(\gets\) ``\textit{exstr} \(\rightarrow (-1)^r~\cdot\) \textit{res\_str\_base} \(\cdot~\gamma_*\)'' \Comment{i.e., \eqref{gamma_red:2}}
                \Else
                    \State \textit{sub\_rule} \(\gets\) ``\textit{exstr} \(\rightarrow rep~\cdot\) \textit{res\_str\_base}'' \Comment{i.e., \eqref{gamma_red:1}}
                \EndIf
                \State \textbf{prepend} \textit{sub\_rule} to \textit{sub\_exs}
            \EndFor

            \If{\(D \equiv 0 \Mod{2}\)}
                \State \textit{exstr} \(\gets\) \(\gamma^{a_1 \cdots a_m} \gamma_*\)
                \State \textit{res\_str\_base} \(\gets\) \((-1)^m \frac{i^{m+1}}{m!} \epsilon^{a_1 \cdots a_m a_{m+1} \cdots a_D} \gamma_{a_D \cdots a_{m+1}}\)
                \State \textit{sub\_rule} \(\gets\) ``\textit{exstr} \(\rightarrow\) \textit{res\_str\_base}'' \Comment{i.e., \eqref{gamma_red:2} with \(r = m\)\footnote{Substitution rules in pseudo-code in this paper will be written in quotes.}}
                \State \textbf{append} \textit{sub\_rule} to \textit{sub\_exs}
            \EndIf

            \State \Return \textit{sub\_exs}
        \EndFunction
    \end{algorithmic}
\end{algorithm}

The SusyPy function \mintinline{python}{susy_env()} introduced in \S\ref{using_susypy} employs the foregoing procedures to determine fix symmetries and identities for the chosen dimension and representation. Since the \mintinline{python}{susy_env()} procedure is long and highly technical, dealing largely with Cadabra idiosyncrasies, detailed pseudo-code for it is omitted here.

\subsection{Algorithm for Spinor-Index Canonicalization}\label{spinor_canon}
The algorithm ``SpinorCombine,'' previewed earlier, transforms an expression \(\mathcal{E}\), involving terms which are products of elements with spinor indices, into a canonical expression which maximally combines factors by exploiting (monoterm) symmetries, the spinor metric, and sorting. ``Combining'' refers to multiplication of the form \((B)_{\alpha}{}^{\beta} (C)_{\beta}{}^{\gamma} = (BC)_{\alpha}{}^{\gamma}\) which attaches spinor indices to the product of pure Lorentz tensors following the NW-SE convention\footnote{Strictly speaking, NW-SE convention is only necessary if the spinor metric is antisymmetric, i.e., in dimensions with \(t_0 = 1\), but there is no error and relatively little computational overhead in forcing NW-SE convention anyway, and the spinor metric is antisymmetric in most dimensions of interest for supersymmetry.} that this combination can only occur given a dummy index pair in which the first index is contravariant and the second covariant. ``SpinorCombine'' loops through the terms in \(\mathcal{E}\) and breaks each down into a coefficient \(c_0\), a set \(S\) of spinor-indexed factors (i.e., indexbrackets or spinor metrics in the term), a set \(I\) of underlying Lorentz tensors of the latter (e.g., if \((B)_{\alpha \beta} \in S\), then \(B \in I\)), and the remaining factors \(S^c\) (the complement of \(S\) in the set of factors of the term). Each element \(s \in S\) has at most two spinor indices,\footnote{\(s \in S\) with two spinor indices (and \(s \in S^c\)) is bosonic and \(s \in S\) with one spinor index is fermionic. Hence, an element with \(3\) or more indices cannot carry any additional physical significance and is absent in any calculation of interest.} so the set \(O\) of indices of the \(s \in S\) is recorded as an array of 2-tuples \(o(s)\) for all \(s \in S\), where the second entry of \(o(s)\) is left null if \(s\) carries only one spinor index. For example, if the term is \(2 (A)_{\alpha \beta} (B)_{\gamma}\), then \(c_0 = 2\), \(S = \{(A)_{\alpha \beta}, (B)_{\gamma}\}\), \(I = \{A, B\}\), and \(O = \{({}_{\alpha}, {}_{\beta}), ({}_{\gamma}, \textnormal{None})\}\).\footnote{In the next few sections, even when indices are given in isolation, we will frequently impart them the appropriate parity.}

If \(|S| \leq 1\), then the term is already reduced, but if \(|S| \geq 2\), then the algorithm proceeds by augmenting every 2-tuple \(o \in O\) with the boolean proposition that the factor \(s(o) \in S\) corresponding to \(o\) has an (anti)symmetry. The set of \(3\)-tuples is inputted into ``FindChain'' (discussed later) to retrieve the ordering of the elements \(s \in S\) which leads to maximal chains of consecutive spinor-indexed factors which can be combined. \(S, I, O\) are imparted this ordering, and then the program loops through the consecutive pairs of factors in this ordering.\footnote{Here, pairs \((s_1, s_2), (s_3, s_4)\) of consecutive spinor-indexed factors are considered consecutive if \(s_2 = s_3\).} This first part of the procedure is shown in pseudo-code as Algorithm \ref{alg:spinorcomb1}.

\begin{algorithm}[ht!]
    \caption{Spinor Index Canonicalization Part 1}\label{alg:spinorcomb1}
    \begin{algorithmic}[1]
        \Function{SpinorCombine}{$\mathcal{E}$}
            \For{\textit{term} in \(\mathcal{E}\)}
                \State \(c_0 \gets \) multiplier of term
                \State \(S \gets\) factors of \textit{term} that are a spinor metric or a Cadabra ``indexbracket''
                \State \(S^c \gets\) other factors in \textit{term}
                \State \(I \gets \{s \textnormal{ without spinor indices} ~ | ~ s \in S\}\)
                \State \(O \gets \) [ ]

                \For{\(s \in S\)}
                    \State \(o \gets\) spinor indices of \(s\)
                    \If{\(o\) has only one element}
                        \State \textbf{append} [\(o\)[0], None] to \(O\)
                    \Else
                        \State \textbf{append} \(o\) to \(O\)
                    \EndIf
                \EndFor

                \State \(n \gets |S|\)

                \If{\(n \geq 2\)}
                    \State \textit{indices\_with\_syms} \(\gets \{(o[0], o[1], \textnormal{symmetry of } s(o)) ~ | ~ o \in O\}\)
                    \State \textit{order} \(\gets\) \textbf{FindChain}(\textit{indices\_with\_syms})
                    \State \(S \gets\) reorder \(S\) based on \(order\)
                    \State \(I \gets\) reorder \(I\) based on \(order\)
                    \State \(O \gets\) reorder \(O\) based on \(order\)

                    \For{\(i \in\) 1..\(n-1\)}

                        \State continued...
        \algstore{bkbreak}
    \end{algorithmic}
\end{algorithm}

Next is an accounting of all minus signs introduced when combining the pair of factors, i.e., when casting the pair in NW-SE convention. Let \((s_{i-1}, s_{i})\) be the current pair and \(o_{i-1} = o(s_{i-1}), o_{i} = o(s_{i})\). There are only two ways to get a minus sign, namely, by swapping indices on a spinorially antisymmetric tensor-spinor and (if \(t_0 = 1\)) by swapping parities of dummy indices. (This follows from antisymmetry of the spinor metric, as can be easily verified.) Hence, we deal with seven cases. 1) If one of \(s_{i-1}, s_{i}\) has only one index (\(\textnormal{None} \in o_{i-1} \cup o_{i}\)), then the pair is skipped over. (A tensor-spinor with a single spinor index is a fermionic field, which cannot be meaningfully contracted away.) 2) If each factor has two spinor indices, the second index in \(o_{i-1}\) matches the first in \(o_{i}\), and the former is contravariant while the latter is covariant, then NW-SE convention is satisfied, whence nothing is done. 3) If the conditions in (2) hold except for the dummy index parities, then the parities must be flipped. 4) If the first index in \(o_{i-1}\) matches the first index in \(o_{i}\) and \(s_{i-1}\) is (anti)symmetric, then the indices on \(s_{i-1}\) must be swapped and the parities of dummy indices flipped if necessary. For example, if \(s_{i-1} = (A)_{\beta \alpha}\) and the spinor metric are antisymmetric, then \((A)_{\beta \alpha} (B)^{\beta \gamma} \longmapsto -(A)_{\alpha \beta} (B)^{\beta \gamma} \longmapsto (A)_{\alpha}{}^{\beta} (B)_{\beta}{}^{\gamma}\), which is in NW-SE convention. (5) and (6) are analogous but exploiting (anti)symmetry of the second factor and both factors, respectively. 7) There is no exploitable symmetry or dummy indices. This second part of the procedure is shown in pseudo-code as Algorithm \ref{alg:spinorcomb2}.

\begin{algorithm}[ht!]
    \caption{Spinor Index Canonicalization Part 2}\label{alg:spinorcomb2}
    \begin{algorithmic}[1]
        \algrestore{bkbreak}
            \State \(c \gets 1\)

            \If{\(S[i-1], S[i]\) are Cadabra ``indexbrackets'' but None \(\in O[i-1] \cup O[i]\)} \Comment{Case (1)}
                \State \textbf{continue} to next iteration

            \ElsIf{\(O[i-1][1], O[i][0]\) match \& \(O[i-1][1]\) is covariant \& \(O[i][0]\) is contravariant} \Comment{Case (3)}
                \State \(c \gets -c\) \textbf{if} \(t_0 = 1\)
                \State \textbf{make} \(O[i-1][1]\) contravariant
                \State \textbf{make} \(O[i][0]\) covariant

            \ElsIf{\(O[i-1][0], O[i][0]\) match \& have opposite parity \& \(S[i-1]\) is (anti)sym} \Comment{Case (4)}
                \State\textbf{reverse} order of \(O[i-1]\)
                \State \(c \gets -c\) \textbf{if} \(S[i-1]\) is antisymmetric
                \State \(c \gets -c\) \textbf{if} \(O[i-1][1]\) is covariant \& \(O[i][0]\) is contravariant \& \(t_0 = 1\)
                \State \textbf{make} \(O[i-1][1]\) contravariant
                \State \textbf{make} \(O[i][0]\) covariant

            \ElsIf{\(O[i-1][1], O[i][1]\) match \& have opp. parity \& \(S[i]\) is (anti)sym} \Comment{Case (5)}
                \State\textbf{reverse} order of \(O[i]\)
                \State \(c \gets -c\) \textbf{if} \(S[i]\) is antisymmetric
                \State \(c \gets -c\) \textbf{if} \(O[i-1][1]\) is covariant \& \(O[i][0]\) is contravariant \& \(t_0 = 1\)
                \State \textbf{make} \(O[i-1][1]\) contravariant
                \State \textbf{make} \(O[i][0]\) covariant

            \ElsIf{\(O[i-1][0], O[i][1]\) match \& have opp. parity \& \(S[i-1], S[i]\) are (anti)sym} \Comment{Case (6)}
                \State\textbf{reverse} order of \(O[i-1]\)
                \State\textbf{reverse} order of \(O[i]\)
                \State \(c \gets -c\) \textbf{if} \(S[i-1]\) is antisymmetric
                \State \(c \gets -c\) \textbf{if} \(S[i]\) is antisymmetric
                \State \(c \gets -c\) \textbf{if} \(O[i-1][1]\) is covariant \& \(O[i][0]\) is contravariant \& \(t_0 = 1\)
                \State \textbf{make} \(O[i-1][1]\) contravariant
                \State \textbf{make} \(O[i][0]\) covariant
            \EndIf

            \State continued...
        \algstore{bkbreak}
    \end{algorithmic}
\end{algorithm}

Finally, if NW-SE convention was achieved, then the product tensor-spinor is constructed. This means wrapping the underlying Lorentz tensors in a common indexbracket with the dummy indices eliminated, multiplying by the original coefficients, replacing the old second factor with this new one, and deleting the first factor. (If one of the factors is a spinor metric, then it has no underlying Lorentz tensor, but this is handled in an analogous way.) After processing every consecutive pair of factors in this manner,\footnote{Notice that the result of an iteration is included in the next pair.} the term being considered in \(\mathcal{E}\) is replaced by the symbolic product of the modified and combined elements of \(S\), the elements of \(S^{c}\), and the coefficient \(c_0\) of the term recorded at the beginning. After looping through every term, the program returns the resulting maximally combined, spinor-index-canonicalized \(\mathcal{E}\). This third part of ``SpinorCombine'' is shown as pseudo-code in Algorithm \ref{alg:spinorcomb3}.

\begin{algorithm}[ht!]
    \caption{Spinor Index Canonicalization Part 3}\label{alg:spinorcomb3}
    \begin{algorithmic}[1]
        \algrestore{bkbreak}
                        \If{\(O[i-1][1], O[i][0]\) match \& \(O[i-1][1]\) is contravariant \& \(O[i][0]\) is covariant}
                            \State \(c_1 \gets\) multiplier of \(S[i-1]\)
                            \State \(c_2 \gets\) multiplier of \(S[i]\)

                            \If{\(S[i-1]\) and  \(S[i]\) are Cadabra ``indexbrackets''}
                                \State \(S[i] \gets c \cdot c_1 \cdot c_2 \cdot (I[i-1] \cdot I[i])_{O[i-1][0] O[i][1]}\)
                                \State \(I\)[i] \(\gets I\)[i-1]\(\cdot I\)[i]

                            \ElsIf{\(S[i-1]\) and  \(S[i]\) are spinor metrics}
                                \State \(S[i] \gets c \cdot c_1 \cdot c_2 \cdot C_{O[i-1][0] O[i][1]}\)
                                \State \(I[i] \gets \mathds{1}\)

                            \ElsIf{\(S[i-1]\) is a Cadabra ``indexbracket'' \& \(S[i]\) is a spinor metric}
                                \State \(S[i-1] \gets c \cdot c_1 \cdot c_2 \cdot (I[i-1])_{O[i-1][0] O[i][1]}\)
                                \State \(S[i] \gets S[i-1]\)
                                \State \(I[i] \gets I[i-1]\)

                            \ElsIf{\(S[i-1]\) is a spinor metric \& \(S[i]\) is a Cadabra ``indexbracket''}
                                \State \(S[i] \gets c \cdot c_1 \cdot c_2 \cdot (I[i])_{O[i-1][0] O[i][1]}\)

                            \EndIf
                            
                            \State \(O[i][0] \gets O[i-1][0]\)
                            \State \textbf{delete} \(S[i-1]\) from \(S\)
                            \State \textbf{delete} \(O[i-1]\) from \(O\)
                            \State \textbf{delete} \(I[i-1]\) from \(I\)
                        \EndIf
                    \EndFor

                    \State \textbf{replace} \textit{term} in \(\mathcal{E}\) with \(c_0~\cdot\) \textbf{prod}(\(S^c\)) \(\cdot\) \textbf{prod}(\(S\)) \Comment{\textbf{prod} is symbolic product}
                \EndIf
            \EndFor
            \State \Return \(\mathcal{E}\)
        \EndFunction
    \end{algorithmic}
\end{algorithm}

We now consider ``FindChain.'' Recall that its input is an array \(A\) of 3-tuples whose entries are the two indices of a factor (one of which may be ``None'') and the boolean proposition that the corresponding \(s_{i} \in S\) is (anti)symmetric with respect to these indices. ``FindChain'' augments each 3-tuple with a number \(i\): the location of the 3-tuple in \(A\).\footnote{In ordinary programming jargon, \(i\) would be called the \textit{index} of the 3-tuple in \(A\), but we refer to it as the \textit{location} to prevent confusion.} Let \(B\) be the resulting array of 4-tuples; the goal of the algorithm is essentially to order these 4-tuples into optimal chains and read off the indices to obtain the best reordering of the elements of \(S\). In particular, the goal will be to create chains of spinor factors ordered so that exploitation of symmetries and parity flips can combine the factors according to the NW-SE convention without commuting (reordering) the factors.

We start by sorting \(B\) so that elements with no symmetries appear first, then elements with symmetries, and finally elements with one spinor index. This sets up the forthcoming loop: elements without symmetries cannot be manipulated, so they form rigid chains; elements with symmetries can be manipulated to fit these chains, and at the end, elements with single spinor indices terminate ends of these chains (e.g., if a chain has beginning index \(\alpha\) and ending index \(\beta\), and an element with sole index \(\beta\) is added to the chain, then the resulting chain has only one end to which further elements may be joined, viz., \(\alpha\)). We consider an evolving array \(C\) of chains. Each chain is a length-5 array consisting of the array of locations of the chain's elements as elements of \(A\), the first spinor index in the chain, the last spinor index in the chain, the boolean proposition that exploitation of symmetries can enable the first and last indices to be swapped, and the location of the chain as an element of \(C\) (e.g., if the third chain in \(C\) consists of the second, eighth, and tenth elements of \(A\), has first index \(\alpha\), has last index \(\beta\), and these two indices may be swapped by exploiting symmetries of the components of the chain, then this chain is represented by \([[2,8,10], \alpha, \beta, \textnormal{True}, 2]\)). Notice that only the first and last indices in the chain are relevant, because all other indices are paired within the chain. Also, the boolean proposition that exploitation of symmetries can swap the first and last indices of the chain is the same as the boolean proposition that every component of the chain bears a symmetry in its spinor indices.

The first chain created is a seed for the procedure, viz., a length-5 array representing the 1-element chain consisting of only the first element \(B_0\) of \(B\). This length-5 array consists of the 1-element array with the location of \(B_0\) in \(A\), the first index of \(B_0\), the second index of \(B_0\), the boolean proposition that \(B_0\) has a symmetry,\footnote{To say that an element \(B_0\) of \(B\) has a symmetry is technically an abuse of notation, since it is the element of \(S\) corresponding to \(B_0\) which has a symmetry, and \(B_0\) simply carries the boolean of the existence of such a symmetry. Nevertheless, this minor misnomer improves brevity and should cause no confusion.} and the number 0 reflecting that this is the first chain.\footnote{We follow the programming convention that the first element of an array has location 0.} With this seed in place, the algorithm proceeds to loop through the remaining elements in \(B\). Let \(B_{j}\) be the \(j\)th element of \(B\), \(\alpha_{j}\) be its first index, \(\beta_{j}\) be its second index, \(\mathfrak{s}_{j}\) be the boolean proposition of symmetry, and \(i_{j}\) be the location of \(B_{j}\) in \(A\), so that \(B_{j} = (\alpha_{j}, \beta_{j}, \mathfrak{s}_{j}, i_{j})\). In round \(j\), the algorithm loops through all elements of \(C\), i.e., all stored chains. For each chain \(c \in C\), let ``chain'' be the array of components of the chain, \(\alpha_{c}\) be the first index, \(\beta_{c}\) be the last index, \(\mathfrak{s}_{c}\) be the boolean proposition of overall symmetry, and \(i_{c}\) be the location of \(c\) in \(C\), so that \(c = [\textnormal{chain}, \alpha_{c}, \beta_{c}, \mathfrak{s}_{c}, i_{c}]\). If \(B_{j}\) has not yet been connected to any chain by the time \(c\) is reached in the loop, then there are three possibilities. 1) If \(\alpha_{c}\) is not null and equals the second index \(\beta_{j}\) of \(B_{j}\),\footnote{By equality, we mean equality up to parity. Obviously, two spinor indices of the same parity cannot exist in the same product.} then \(B_{j}\) can be prepended to the chain by prepending \(i_{j}\) to ``chain'' (e.g., if \(c\) represents the chain \((D_1)_{\alpha}{}^{\gamma_1} \cdots (D_{r+1})_{\gamma_r}{}^{\beta}\) and \(B_{j}\) represents \((E)_{\eta}{}^{\alpha}\), then the latter can be prepended to the former to create the chain \((E)_{\eta}{}^{\alpha} (D_1)_{\alpha}{}^{\gamma_1} \cdots (D_{r+1})_{\gamma_r}{}^{\beta}\)). The first index and symmetry of the new chain derive from the first index of \(B_{j}\) and the symmetries of the \(c\) and \(B_{j}\), viz., \(\alpha_{c}\) is replaced by \(\alpha_{j}\) and \(\mathfrak{s}_{c}\) by the boolean \(\mathfrak{s}_{c} \& \mathfrak{s}_{j}\). An equivalent situation can occur if \(\alpha_{c}\) is equal to \(\alpha_{j}\) and \(B_{j}\) bears a symmetry in its indices (e.g., \(B_{j}\) represents \((E)^{\alpha}{}_{\eta}\)). Of course, the indices of \(B_{j}\) take opposite roles here. Either way, the number of the chain is recorded for later use (see the next paragraph). 2) If \(\beta_{c}\) is not null and equals \(\alpha_{j}\) (or  if \(\beta_{c}\) equals \(\beta_{j}\) and \(B_{j}\) bears a symmetry), then \(B_{j}\) can be appended to the chain in an analogous way. 3) If neither of the two foregoing conditions holds, then \(B_{j}\) cannot be linked to \(c\).

The loop through the chains \(c \in C\) does not terminate when \(B_{j}\) is connected to a chain, for it may have two spinor indices and therefore connect to an additional chain. The result would be combining the two chains to which \(B_{j}\) is linked, which entails a distinct set of operations from those executed when \(B_{j}\) is first linked to a chain. There are six ways that \(B_{j}\) can linked to another \(c \in C\) if \(B_{j}\) has already been connected to the \(k\)th chain \(C_{k}\). If the first index of \(C_{k}\) equals \(\alpha_{c}\), they are not null, and \(c\) bears a symmetry in its indices (\(\mathfrak{s}_{c} = \textnormal{True}\)), then \(c\) can be reversed and prepended to \(C_{k}\) by reversing ``chain'' and prepending it to the first entry in the length-5 array \(C_{k}\)\footnote{Of course, formally, attaching \(c \in C\) to an earlier \(C_k \in C\) requires deleting \(c\) from \(C\) as a separate chain.} (e.g., if \(C_{k}\) represents the chain \((D_1)_{\alpha}{}^{\gamma_1} \cdots (D_{r+1})_{\gamma_r}{}^{\beta}\) and \(c\) represents the chain \((E_1)^{\alpha \eta_1} \cdots (E_{s+1})_{\eta_s \rho}\), then the two are replaced by the chain \((E_{s+1})_{\rho}{}^{\eta_s} \cdots (E_1)_{\eta_1}{}^{\alpha} (D_1)_{\alpha}{}^{\gamma_1} \cdots (D_{r+1})_{\gamma_r}{}^{\beta}\), up to a sign). The indices of the combined chain are derived in the natural way. The remaining cases are analogous, except that \(C_{k}\) may be the one bearing a symmetry, or different indices may be equal. In all of cases, the symmetry boolean of the combined chain is the boolean \(\mathfrak{s}_{c} \& C_{k}[3]\) which is true if and only if both \(c\) and \(C_{k}\) bear symmetry in their indices.

After looping through every \(c \in C\), if \(B_{j}\) could not connect to any existing chain, it is rendered its own 1-element chain \([[i_{j}], \alpha_{j}, \beta_{j}, \mathfrak{s}_{j}, |C|]\) at the end of \(C\). After looping through every \(B_{j} \in B\), the only remaining task is to sort the resulting chains by indices in an extended lexicographic order in which null is greater than all other possibilities for the indices. In particular, chains are reversed when symmetry is available to ensure that the first index is lower lexicographically,\footnote{Notice that this puts any single-indexed factor at the end of the chain.} and then chains are ordered lexicographically by first index. From each chain, the first entry, containing the locations of the components of the chain in \(A\), is drawn, and the ordered union of these first entries is returned as the output of the function.\footnote{Specifically, the strict total order on the union is defined by setting \(i < j\) if either there exists \(c \in C\) such that \(i, j \in c[0]\) and \(i < j\), or given \(c \ni i, c' \ni j\), \(c\) precedes \(c'\) in lexicographic order of their first index, i.e., \(\alpha_{c}\) is lower lexicographically than \(\alpha_{c'}\).} From the foregoing, this outputted ordering of the elements in \(A\) is both canonical and representative of the maximum chaining of elements adhering to the NW-SE convention. Algorithm \ref{alg:findchain} shows detailed pseudo-code for the described procedure.

\begin{algorithm}[ht!]
    \caption{Sorting Objects with Spinor Indices}\label{alg:findchain}
    \begin{algorithmic}[1]

        \Function{FindChain}{\(A\)[\(0..n\)]}
            \State \(B \gets \{(A_i[0], A_i[1], A_i[2], i) \mid i \in 0..n\}\) \Comment{add the location \(i\) of every \(a \in A\) to the end of \(a\)}
            \State \textbf{move} all \(b \in B\) with symmetries to the end of \(B\)
            \State \textbf{move} all \(b \in B\) with only one spinor index to the end of \(B\)

            \State \(C \gets [[[B_0[3]], B_0[0], B_0[1], B_0[2], 0]]\) \Comment{initialize our chain set with the first element of \(B\)}

            \For{\(j \in 1..n\)}
                \State \(\alpha_j := B_j[0], ~ \beta_j := B_j[1], ~ \mathfrak{s}_j := B_j[2], ~ i_j := B_j[3]\)
                \State \textit{connected} \(\gets\) False
                \State \(k = -1\) \Comment{location of chain to which \(B_j\) is joined (N.B. \(-1\) means not joined)}
                
                \For{\(c \in C\)}
                    \State \textit{chain} \(:= c[0], ~ \alpha_c := c[1], ~ \beta_c := c[2], ~ \mathfrak{s}_c := c[3], ~ i_c := c[4]\)
                    
                    \If{not \textit{connected} \& \(\alpha_c \ne\) None \& (\(\alpha_c = \beta_j\) or (\(\mathfrak{s}_j\) \& \(\alpha_c = \alpha_j\)))} \Comment{\(B_j\) can join \(c\) on the left}
                        \State \textbf{prepend} \(i_j\) to \textit{chain}
                        \State \(\alpha_c \gets \alpha_j\) if \(\alpha_c \ne \alpha_j\) else \(\beta_j\)
                        \State \(\mathfrak{s}_c \gets \mathfrak{s}_c\) \& \(\mathfrak{s}_j\) \Comment{the combination has symmetry \(\Leftrightarrow\) both \(B_j\) \& \(c\) have symmetry}
                        \State \textit{connected} \(\gets\) True
                        \State \(k \gets i_c\)
                    \ElsIf{not \textit{connected} \& \(\beta_c \ne\) None \& (\(\beta_c = \alpha_j\) or (\(\mathfrak{s}_j\) \& \(\beta_c = \beta_j\)))} \Comment{\(B_j\) can join \(c\) on the right}
                        \State \textbf{append} \(i_j\) to \textit{chain}
                        \State \(\beta_c \gets \alpha_j\) if \(\beta_c \ne \alpha_j\) else \(\beta_j\)
                        \State \(\mathfrak{s}_c \gets \mathfrak{s}_c\) \& \(\mathfrak{s}_j\)
                        \State \textit{connected} \(\gets\) True
                        \State \(k \gets i_c\)
                    \ElsIf{\textit{connected} \& \(C_k[1] = \alpha_c\) \& \(C_k[1], \alpha_c \ne\) None \& \(\mathfrak{s}_c\)} \Comment{\(c\) flipped can join \(C_k\) on the left}
                        \State \(C_k \gets [\textbf{reverse}(chain) + C_k[0], \beta_c, C_k[2], \mathfrak{s}_c\) \& \(C_k[3]]\)
                        \State \textbf{delete} \(c\)
                    \ElsIf{\textit{connected} \& \(C_k[1] = \alpha_c\) \& \(C_k[1], \alpha_c \ne\) None \& \(C_k[3]\)} \Comment{\(C_k\) flipped can join \(c\) on the left}
                        \State \(C_k \gets [\textbf{reverse}(C_k[0]) + chain, C_k[2], \beta_c, \mathfrak{s}_c\) \& \(C_k[3]]\)
                        \State \textbf{delete} \(c\)
                    \ElsIf{\textit{connected} \& \(C_k[2] = \beta_c\) \& \(C_k[2], \beta_c \ne\) None \& \(\mathfrak{s}_c\)} \Comment{\(c\) flipped can join \(C_k\) on the right}
                        \State \(C_k \gets [C_k[0] + \textbf{reverse}(chain), C_k[1], \alpha_c, \mathfrak{s}_c\) \& \(C_k[3]]\)
                        \State \textbf{delete} \(c\)
                    \ElsIf{\textit{connected} \& \(C_k[2] = \beta_c\) \& \(C_k[2], \beta_c \ne\) None \& \(C_k[3]\)} \Comment{\(C_k\) flipped can join  \(c\) on the right}
                        \State \(C_k \gets [chain + \textbf{reverse}(C_k[0]), \alpha_c, C_k[1], \mathfrak{s}_c\) \& \(C_k[3]]\)
                        \State \textbf{delete} \(c\)
                    \ElsIf{\textit{connected} \& \(C_k[1] = \beta_c\) \& \(C_k[1], \beta_c \ne\) None} \Comment{\(c\) can join \(C_k\) on the left}
                        \State \(C_k \gets [chain + C_k[0], \alpha_c, C_k[2], \mathfrak{s}_c\) \& \(C_k[3]]\)
                        \State \textbf{delete} \(c\)
                    \ElsIf{\textit{connected} \& \(C_k[2] = \alpha_c\) \& \(C_k[2], \alpha_c \ne\) None} \Comment{\(c\) can join \(C_k\) on the right}
                        \State \(C_k \gets [C_k[0] + chain, C_k[1], \beta_c, \mathfrak{s}_c\) \& \(C_k[3]]\)
                        \State \textbf{delete} \(c\)
                    \EndIf

                \EndFor
                \If{not \textit{connected}}
                    \State \textbf{append} \([[i_j], \alpha_j, \beta_j, \mathfrak{s}_j, |C|]\) to \(C\)
                \EndIf
            \EndFor

            \State \textbf{reverse} \(c[0]\) and \textbf{swap} \(\alpha_c, \beta_c\) for any \(c \in C\) such that \(\mathfrak{s}_c\) and \(\beta_c < \alpha_c\) \Comment{by convention, ``None'' \(>\) any value}
            \State \textbf{sort} \(c \in C\) by \(\alpha_c\)
            \State \Return \(\bigcup_{c \in C} c[0]\)
        \EndFunction
    \end{algorithmic}
\end{algorithm}

We try to illuminate the foregoing via an illustrated example in Fig. \ref{fig:findchain}. Consider the ostensibly unwieldy expression
\begin{equation}
    A_{\alpha \beta} B_{\gamma}{}^{\delta} \bar{C}_{\epsilon}{}^{\zeta} D_{\eta}{}^{\theta} \bar{E}_{\zeta}{}^{\beta} \bar{F}_{\iota}{}^{\kappa} G^{\lambda \eta} \bar{H}_{\kappa}{}^{\gamma} I_{\theta \delta} J_{\lambda} \bar{K}^{\alpha \iota},
\end{equation}
where the tensors with bars, viz., \(\bar{C}_{\epsilon}{}^{\zeta}, \bar{E}_{\zeta}{}^{\beta}, \bar{F}_{\iota}{}^{\kappa}, \bar{H}_{\kappa}{}^{\gamma}, \bar{K}_{\alpha}{}^{\iota}\) lack symmetries; the tensors without bars, viz., \(A_{\alpha \beta}, B_{\gamma}{}^{\delta}, D_{\eta}{}^{\theta}, G^{\lambda \eta}, I_{\theta \delta}\), are symmetric;\footnote{Here, we mean that these tensor-spinors not only bear a symmetry, but that no negative is introduced by swapping spinor indices. This simplifies the discussion.} and only \(J_{\lambda}\) is single-indexed. In Fig. \ref{fig:findchain}(a), it is recognized the \(\bar{C}_{\epsilon}{}^{\zeta}, \bar{E}_{\zeta}{}^{\beta}\) share the index \(\zeta\) and are already in NW-SE form, so they can be combined into a chain. Similarly, \(\bar{K}^{\alpha \iota}, \bar{F}_{\iota}{}^{\kappa}\) share \(\iota\) and \(\bar{F}_{\iota}{}^{\kappa}, \bar{H}_{\kappa}{}^{\gamma}\) share \(\kappa\) in NW-SE form, so \(\bar{K}^{\alpha \iota}, \bar{F}_{\iota}{}^{\kappa},  \bar{H}_{\kappa}{}^{\gamma}\) can be combined into a chain. These two chains consist entirely of tensors without symmetries (indicated by the solid circles), so they are rigid and the positions of the outer indices cannot be flipped; rather, the remaining tensors with symmetries may need to be flipped to accommodate these fixed chains. In Fig. \ref{fig:findchain}(b), it is noticed that if one swaps the indices on \(A_{\alpha \beta}\), which is symmetric and therefore manipulable (indicated by the dashed circle), then it will share \(\beta\) with the \(\bar{C}\)-\(\bar{E}\) chain in NW-SE form. The swapping of the indices is indicated by the positions of \(\alpha, \beta\) to the right and left of \(A\), respectively, rather than their original positions to the left and right of \(A\), respectively. Notice that since the resulting chain still includes tensors without symmetry, i.e., solid circles, this chain is still rigid. In Fig. \ref{fig:findchain}(c), it is noticed that despite this rigidity, the \(\bar{C}\)-\(A\) chain already shares \(\alpha\) with the \(\bar{K}\)-\(\bar{H}\) chain in NW-SE form (up to parity), so these two chains are concatenated.\footnote{N.B. We just as well could have made \(\bar{C}\), or any other tensor, the top node. The choice of top node is purely aesthetic and has no interpretation with respect to the nodes' ordering in ``FindChain.''} In Fig. \ref{fig:findchain}(d), the algorithm moves on to the other 2-indexed symmetric tensors and finds that \(D_{\eta}{}^{\theta}, G^{\lambda \eta}\) share \(\eta\) in NW-SE form and form a chain. Since \(D_{\eta}{}^{\theta}, G^{\lambda \eta}\) are symmetric, so is the \(G\)-\(D\) chain. The algorithm also finds that \(B_{\gamma}{}^{\delta}\) shares \(\gamma\) with the \(\bar{C}\)-\(\bar{H}\) chain in NW-SE form, forming a \(\bar{C}\)-\(B\) chain, and that if one swaps the indices on \(I_{\theta \delta}\), then it will share \(\delta\) with the \(\bar{C}\)-\(B\) chain in NW-SE form, forming a \(\bar{C}\)-\(I\) chain. The algorithm discovers the latter before it finds the shared index with the smaller chain because the larger chain was defined earlier, and therefore appears earlier when looping through the array of chains. In Fig. \ref{fig:findchain}(e), the program continues looping through the array of chains while considering \(I\) and finds that if the \(G\)-\(D\) chain is flipped, then it will share \(\theta\) with the \(\bar{C}\)-\(I\) chain in NW-SE form. Finally, in Fig. \ref{fig:findchain}(f), it is recognized that the single-indexed tensor \(J_{\lambda}\) shares its one index \(\lambda\) with the \(\bar{C}\)-\(G\) chain, and the algorithm attaches the former to the latter, terminating one end of the chain. There are no symmetry issues in this attachment, as \(J\) has only a single spinor index and is therefore effectively symmetric. Reading off of Fig. \ref{fig:findchain}(f), the optimal ordering which will enable all factors to be combined while adhering to NW-SE convention is
\begin{equation}
    \bar{C}_{\epsilon}{}^{\zeta} \bar{E}_{\zeta}{}^{\beta} A_{\beta \alpha} \bar{K}^{\alpha\iota} \bar{F}_{\iota}{}^{\kappa} \bar{H}_{\kappa}{}^{\gamma} B_{\gamma}{}^{\delta}  I_{\delta \theta} D^{\theta}{}_{\eta} G^{\eta \lambda} J_{\lambda}.
\end{equation}
Notice that the expression is not yet in NW-SE convention in terms of parity. Since it is always possible to raise and lower the indices in a dummy-index pair at the cost of a minus sign (if \(t_0 = 1\)), this is not an ordering issue and is consigned to ``SpinorCombine,'' which incidentally gives the completed expression (assuming no underlying Lorentz tensors which would enable contraction of dummy spinor indices)
\begin{equation}
    -\bar{C}_{\epsilon}{}^{\zeta} \bar{E}_{\zeta}{}^{\beta} A_{\beta}{}^{\alpha} \bar{K}_{\alpha}{}^{\iota} \bar{F}_{\iota}{}^{\kappa} \bar{H}_{\kappa}{}^{\gamma} B_{\gamma}{}^{\delta}  I_{\delta}{}^{\theta} D_{\theta}{}^{\eta} G_{\eta}{}^{\lambda} J_{\lambda}.
\end{equation}

\begin{figure}[ht!]
    \centering
    \includegraphics[width=0.99\textwidth]{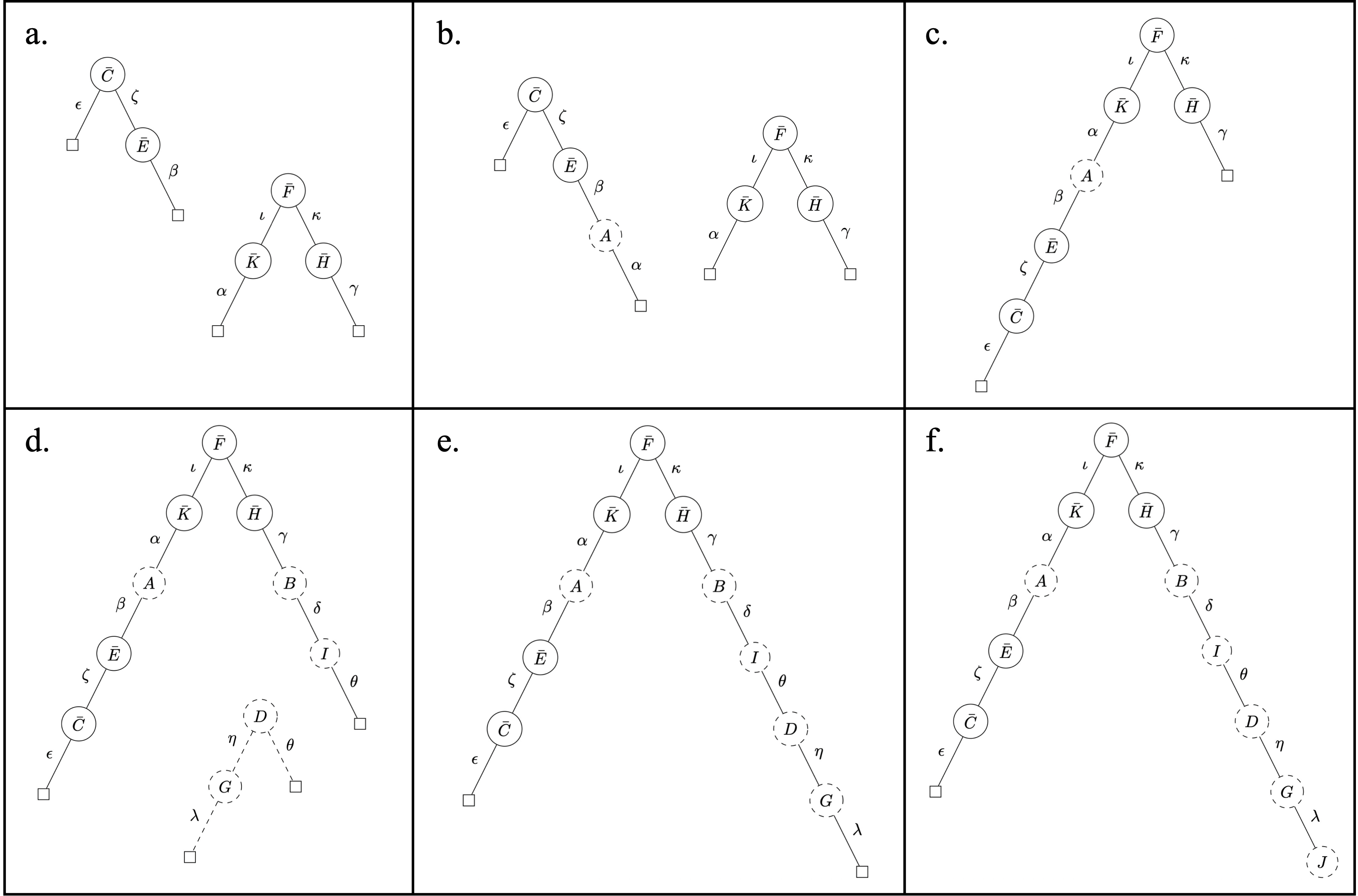}
    \caption{A graphical illustration of the algorithm ``FindChain'' applied to the ordering of a product \(A_{\alpha \beta} B_{\gamma}{}^{\delta} \bar{C}_{\epsilon}{}^{\zeta} D_{\eta}{}^{\theta} \bar{E}_{\zeta}{}^{\beta} \bar{F}_{\iota}{}^{\kappa} G^{\lambda \eta} \bar{H}_{\kappa}{}^{\gamma} I_{\theta \delta} J_{\lambda} \bar{K}^{\alpha \iota}\) of tensors, some with symmetries, some without, and one with only a single spinor index. a) Constructing chains from tensors without symmetries. b) Adding a tensor with symmetry. c) Combining chains. d) Introducing a separate chain of tensors with symmetry. e) Combining chains by reversing one. f) Adding a single-indexed tensor.}
    \label{fig:findchain}
\end{figure}

\subsection{Evaluating Expressions}\label{evaluate}
The main simplification algorithm ``Evaluate'' (see Algorithm \ref{alg:eval}) combines the typical post-processing Cadabra functions with ``SpinorCombine,'' ``GenSubs,'' and a few specialized subprocedures. We suppress the technical details of the subprocedures but briefly outline their principal functions here. ``EvaluateTraces'' evaluate all implicit traces \((B)_{\alpha}{}^{\alpha} = \tr{B}\) which arise after the operations of ``SpinorCombine.'' ``EpsilonToDelta'' computes products of Levi-Civita tensors.\footnote{Since Cadabra's native \mintinline{python}{epsilon_to_delta()} cannot handle products involving more than two Levi-Civita tensors, the wrapper ``EpsilonToDelta'' computes such products pairwise.} ``IndexBracketHex'' enables canonicalization of Lorentz indices through indexbrackets (which is typically impossible in Cadabra, since Cadabra's \mintinline{python}{canonicalise()} function cannot jump between branches of an expression's tree structure) by temporarily replacing a tensor-spinor factor with a pure Lorentz tensor whose spinor indices are encoded by a hexadecimal name. ``SubstituteSpinorZeros'' converts all terms involving a factor of \(0\) within indexbrackets (e.g., \((0)_{\alpha}{}^{\beta}\)) to zero. Finally, ``SpinorExpand'' and ``FourierExpand'' move coefficients, Kronecker deltas, and Levi-Civita tensors outside of indexbrackets and Fourier transforms, and the former also distributes indexbrackets among terms in a sum. The sequence is applied repeatedly until the result reached is no longer unique.\footnote{The simplification sequence is not quite stable. That is, letting \(\mathcal{E}_{n}\) be the outcome of the \(n\)th loop, the sequence \(\{\mathcal{E}_{n}\}_{n=1}^{\infty}\) might not be eventually constant, but rather eventually periodic with period greater than 1. Hence, ``Evaluate'' records all \(\mathcal{E}_{n}\) and terminates after loop \(n = r\) if \(\mathcal{E}_{r} \in \{\mathcal{E}_{i}~|~i < r\}\).} Note that gamma matrix multiplication\footnote{The gamma-splitting procedure used to identify off-shell equation-of-motion terms in \S\ref{multiplet_solver1} intentionally breaks apart gamma matrices into distinct factors which must not be multiplied in later steps. The optional parameter enables these split products to be left alone for the duration of the solver's procedure.} and reduction to lower-indexed gamma matrices\footnote{Applying substitution rules introduces dummy indices via \eqref{gamma_red:1} and \eqref{gamma_red:2}, which is undesirable in the multiplet solvers of \S\ref{multiplet_solver1} and \S\ref{multiplet_solver2}, as it increases canonicalization time. (These substitutions are reserved for the end of ``Evaluate'' for the same reason.) The optional parameter enables these substitutions to be withheld until virtually all canonicalization procedures in those solvers are complete.} may be prevented via optional parameters.

\begin{algorithm}[ht!]
    \caption{Evaluate}\label{alg:eval}
    \begin{algorithmic}[1]
        \Function{Evaluate}{\(\mathcal{E}\), \textit{to\_perform\_subs} = True, \textit{to\_join\_gamma} = True}

            \State \(\mathcal{E} \gets\) \textbf{distribute}(\(\mathcal{E}\))

            \State \textit{prior\_exs} \(\gets\) [ ]
            \While{\(\mathcal{E}\) not in \textit{prior\_exs}}
                \State \textbf{append} \(\mathcal{E}\) to \textit{prior\_exs}

                \State \(\mathcal{E} \gets\) \textbf{product\_rule}(\(\mathcal{E}\)) \Comment{Leibniz rule}
                \State \(\mathcal{E} \gets\) \textbf{SpinorCombine}(\(\mathcal{E}\))
                \State \(\mathcal{E} \gets\) \textbf{EvaluateTraces}(\(\mathcal{E}\))
                \If{\textit{to\_join\_gamma}}:
                    \State \(\mathcal{E} \gets\) \textbf{join\_gamma}(\(\mathcal{E}\)) \Comment{Multiply a pair of gamma matrices}
                \EndIf
                \State \(\mathcal{E} \gets\) \textbf{distribute}(\(\mathcal{E}\))
                \State \(\mathcal{E} \gets\) \textbf{unwrap}(\(\mathcal{E}\)) \Comment{Move coefficients out of derivatives}
                \State \(\mathcal{E} \gets\) \textbf{sort\_product}(\(\mathcal{E}\))
                \State \(\mathcal{E} \gets\) \textbf{sort\_sum}(\(\mathcal{E}\))
                \State \(\mathcal{E} \gets\) \textbf{EpsilonToDelta}(\(\mathcal{E}\))
                \State \(\mathcal{E} \gets\) \textbf{eliminate\_kronecker}(\(\mathcal{E}\)) \Comment{Contract out Kronecker deltas}
                \State \(\mathcal{E} \gets\) \textbf{rename\_dummies}(\(\mathcal{E}\))
                \State \(\mathcal{E} \gets\) \textbf{IndexBracketHex}(\(\mathcal{E}\))
                \State \(\mathcal{E} \gets\) \textbf{canonicalise}(\(\mathcal{E}\))
                \State \(\mathcal{E} \gets\) \textbf{SubstituteSpinorZeros}(\(\mathcal{E}\))
                \State \(\mathcal{E} \gets\) \textbf{substitute}(\(\mathcal{E}\), \(\gamma_* \gamma_* \rightarrow \mathds{1}\))
                \State \(\mathcal{E} \gets\) \textbf{collect\_terms}(\(\mathcal{E}\))
                \State \(\mathcal{E} \gets\) \textbf{SpinorExpand}(\(\mathcal{E}\)) \Comment{unwrap()/distribute() \textnormal{for indexbrackets}}
                \State \(\mathcal{E} \gets\) \textbf{FourierExpand}(\(\mathcal{E}\)) \Comment{unwrap() \textnormal{for Fourier transforms}}
            \EndWhile

            \If{\textit{to\_perform\_subs}}
                \State \(\mathcal{E} \gets\) \textbf{substitute}(\(\mathcal{E}\), \textbf{GenSubs}())
                \State \(\mathcal{E} \gets\) \textbf{SpinorExpand}(\(\mathcal{E}\))
                \State \(\mathcal{E} \gets\) \textbf{canonicalise}(\(\mathcal{E}\))
                \State \(\mathcal{E} \gets\) \textbf{SpinorCombine}(\(\mathcal{E}\))
                \State \(\mathcal{E} \gets\) \textbf{substitute}(\(\mathcal{E}\), \(\gamma_* \gamma_* \rightarrow \mathds{1}\))
            \EndIf

            \State \(\mathcal{E} \gets\) \textbf{collect\_factors}(\(\mathcal{E}\))

            \State \Return \(\mathcal{E}\)

        \EndFunction
    \end{algorithmic}
\end{algorithm}

\subsection{Algorithm for Two-Index Fierz Expansion}\label{fierz_2ind}
While there are a few Fierz transformation/expansion functions in the literature, e.g., Cadabra2's \cite{Peeters2018b} \mintinline{python}{fierz()} and FieldsX's \cite{Frob2021} \mintinline{python}{FierzExpand()}, we need a new algorithm for our purposes. Here, we present a Fierz expansion algorithm distinguished in three ways: 1) it operates on spinor-indexed expressions rather than spinor bilinears, 2) it automatically Fierz-expands every term in an inputted sum, and 3) it only requires knowledge of one pair of spinor indices (which are to be positioned on the same factors in the output). The third requirement is crucial for the multiplet-solving algorithms, as it is necessary to Fierz-expand an expression arising from the anticommutator of supercovariant derivatives, and only the spinor indices on the derivatives is known. (At least one other index is an unknown dummy spinor index.)

The algorithm ``FierzExpand2Index'' (see Algorithm \ref{alg:FierzExpand2Index}) we use is guided by the idea that the coefficients in the expansion of an expression as a linear combination of basis elements need not be scalars. Consider an orthogonal basis \(\Gamma^{A}\) for a Clifford algebra, and suppose that the expression \(\mathcal{E}_{\alpha \beta}\) that we wish to expand has only two known spinor indices (but typically has other unknown spinor indices). We wish to solve for a decomposition of the form
\begin{equation}\label{pseudo_fierz2}
    (\mathcal{E})_{\alpha \beta} = \sum_{B \in \Gamma^{A}} c_{B} (B^{A})_{\alpha \beta}.
\end{equation}
Restricting to one element \((B^{A})_{\alpha \beta} \in \Gamma^{A}\), we can consider the four variants \((B^{A})_{\alpha \beta}\), \((B^{A'})_{\alpha \beta}\), \((B_{A'})^{\beta \alpha}\), and \((B_{A})_{\alpha \beta}\). Multiplying the first and the third gives
\begin{equation}
    (B^{A})_{\alpha \beta} (B_{A'})^{\beta \alpha} = - (B_{A'} B^{A})_{\beta}{}^{\beta} = - \tr(B_{A'} B^{A}) = c'_{B} \delta_{A'}{}^{A}
\end{equation}
for some \(c'_{B}\). The sole purpose of the fourth variant \((B_{A})_{\alpha \beta}\), referred to as the dummy field in Algorithm \ref{alg:FierzExpand2Index}, is to multiply with the above trace so that \(c'_{B}\) can be isolated, for the product is \(c'_{B} (B_{A'})_{\alpha \beta}\).\footnote{\((B_{A})_{\alpha \beta}\) may include a coefficient by which the coefficient of \(c'_{B} (B_{A'})_{\alpha \beta}\) must be divided in order to truly reflect the value \(c'_{B}\) by which the basis element was multiplied.} It follows from the foregoing and orthogonality of the basis elements that multiplying the right-hand side of \eqref{pseudo_fierz2} by \((1/c'_{B}) (B_{A'})^{\beta \alpha}\) yields \(c_{B} \delta_{A'}{}^{A}\), so that multiplying by \((B^{A'})_{\alpha \beta}\) gives the desired projection \(c_{B} (B^{A})_{\alpha \beta}\). Hence, multiplying the left-hand side, viz., \((\mathcal{E})_{\alpha \beta}\), by the same factors yields the projection, and looping through the basis to add up the projections gives the Fierz expansion.

\begin{algorithm}[ht!]
  \caption{Two Spinor Index Fierz Expansion}\label{alg:FierzExpand2Index}
  \begin{algorithmic}[1]
    \Function{FierzExpand2Index}{\(\mathcal{E}\), \(\Gamma^A\), \textit{Inds}[\(1..2\)], \textit{to\_perform\_subs} = True}
      \State \textit{projections} \(\gets\) [ ]
      \For{\(B\in\Gamma^A\)}
        \State \(B\) \(\gets\) \textbf{Evaluate}(\(B\), \textit{to\_perform\_subs} = False)

        \State \textit{primed\_inds} \(\gets\) Lorentz indices of \(B\) primed \Comment{e.g., \({}^{A'}\) if \(B = B^{A}\)}
        \State \textit{unconj\_inds} \(\gets\) \textit{Inds} with parity as in \(\mathcal{E}\)  \Comment{e.g., \([{}_{\alpha}, {}_{\beta}]\) if \textit{Inds} = \([{}_{\alpha}, {}_{\beta}]\) and \(\mathcal{E} = \mathcal{E}_{\alpha \beta}\)}
        \State \textit{conj\_inds} \(\gets\) \textit{Inds} with opposite the parity in \(\mathcal{E}\) \Comment{e.g., \([{}^{\alpha}, {}^{\beta}]\) if \textit{Inds} = \([{}_{\alpha}, {}_{\beta}]\) and \(\mathcal{E} = \mathcal{E}_{\alpha \beta}\)}
        \State \textit{primed\_unconj\_indices} \(\gets\) \textit{unconj\_inds} primed

        \State \textit{element} \(\gets\) \(B\) with spinor indices replaced by \textit{unconj\_inds} \Comment{e.g., \((B^{A})_{\alpha \beta}\)}
        \State \textit{element\_primed} \(\gets\) \textit{element} with Lorentz indices replaced by \textit{primed\_inds} \Comment{e.g., \((B^{A'})_{\alpha \beta}\)}
        \State \textit{element\_primed\_conj} \(\gets\) \textit{element\_primed} with Lorentz indices lowered and spinor indices replaced with \textit{conj\_inds} in reverse order \Comment{e.g., \((B_{A'})^{\beta \alpha}\)}
        \State \textit{element\_squared} \(\gets\) \textit{element} \(\cdot\) \textit{element\_primed\_conj}
        \State \textit{dummy\_field} \(\gets\) \textit{element} with Lorentz indices lowered \Comment{e.g., \((B_{A})_{\alpha \beta}\)}

        \State \textit{element\_squared} \(\gets\) \textbf{SpinorCombine}(\textit{element\_squared}) \Comment{e.g., \(- (B_{A'} B^{A})_{\beta}{}^{\beta}\)}
        \State \textit{elsq\_trace} \(\gets\) \textbf{EvaluateTraces}(\textit{element\_squared}) \Comment{e.g., \(c'_{B} \delta_{A'}{}^{A}\)}
        \State \textit{rhs} \(\gets\) \textbf{Evaluate}(\textit{elsq\_trace} \(\cdot\) \textit{dummy\_field}, \textit{to\_perform\_subs} = False) \Comment{e.g., \(c'_{B} (B_{A'})_{\alpha \beta}\)}
        \State \textit{inv\_rhs\_const} \(\gets\) inverse, of constant coefficient of \textit{rhs} divided by initial coefficient of \(B\) \Comment{i.e., \(c'_{B}\)}

        \State \textit{lhs} \(\gets\) \textbf{Evaluate}(\textit{inv\_rhs\_const} \(\cdot\) \textit{element\_primed\_conj} \(\cdot~\mathcal{E}\), \textit{to\_perform\_subs} = False)
        \State \textit{lhs} \(\gets\) \textit{lhs} with dummy spinor indices replaced by \textit{primed\_unconj\_indices} \Comment{to prevent index name conflicts}
        \State \textit{result} \(\gets\) \textbf{Evaluate}(\textit{lhs} \(\cdot\) \textit{element\_primed}, \textit{to\_perform\_subs} = \textit{to\_perform\_subs})
        
            \State \textbf{append} \textit{result} to \textit{projections}
            \EndFor
            \State \Return symbolic sum of \textit{projections}
        \EndFunction
    \end{algorithmic}
\end{algorithm}

As an example, one can check that the below code (which uses the basis \eqref{11d_basis_redux})
\begin{minted}{python}
>>> fierz_expand_2index(Ex(r'(\Gamma_{a})_{\eta \alpha} (\Gamma^{b c})_{\beta}^{\gamma} + (\Gamma_{a})_{\eta \beta} (\Gamma^{b c})_{\alpha}^{\gamma}'), [Ex(r'C_{\alpha \beta}'), Ex(r'(\Gamma^{a})_{\alpha \beta}'), Ex(r'(\Gamma^{a b})_{\alpha \beta}'), Ex(r'(\Gamma^{a b c})_{\alpha \beta}'), Ex(r'(\Gamma^{a b c d})_{\alpha \beta}'), Ex(r'(\Gamma^{a b c d e})_{\alpha \beta}')], [r'_{\alpha}', r'_{\beta}'])
\end{minted}
verifies the 11D Fierz identity
\begin{align*}
    (\gamma_{a})_{\eta (\alpha} (\gamma^{b c})_{\beta)}{}^{\gamma} &= \frac{1}{16}\bigg\{\left(\gamma_{[1]}\right)_{\alpha \beta} \left(\gamma_{a}{}^{b c [1]}\right)_{\eta}{}^{\gamma} + \left(\gamma_{a}\right)_{\alpha \beta} \left(\gamma^{b c}\right)_{\eta}{}^{\gamma} - \left(\gamma^{[b}\right)_{\alpha \beta} \left(\gamma^{c]}{}_{a}\right)_{\eta}{}^{\gamma} \numberthis \\
    &\quad\quad\quad +\delta_{a}{}^{[b} \left(\gamma_{[1]}\right)_{\alpha \beta} \left(\gamma^{c] [1]}\right)_{\eta}{}^{\gamma} - \delta_{a}{}^{[b} \left(\gamma^{c]}\right)_{\alpha \beta} \delta_{\eta}{}^{\gamma} - \frac{1}{2}\left(\gamma_{[2]}\right)_{\alpha \beta} \left(\gamma_{a}{}^{b c [2]}\right)_{\eta}{}^{\gamma} \\
    &\quad\quad\quad +\left(\gamma_{[1] a}\right)_{\alpha \beta} \left(\gamma^{b c [1]}\right)_{\eta}{}^{\gamma} - \left(\gamma^{[1] [b}\right)_{\alpha \beta} \left(\gamma^{c]}{}_{a [1]}\right)_{\eta}{}^{\gamma} \\
    &\quad\quad\quad -\frac{1}{2}\delta_{a}{}^{[b} \left(\gamma_{[2]}\right)_{\alpha \beta} \left(\gamma^{c] [2]}\right)_{\eta}{}^{\gamma} + \delta_{a}{}^{[b} \left(\gamma^{c] [1]}\right)_{\alpha \beta} \left(\gamma_{[1]}\right)_{\eta}{}^{\gamma} + \left(\gamma^{b c}\right)_{\alpha \beta} \left(\gamma_{a}\right)_{\eta}{}^{\gamma} \\
    &\quad\quad\quad -\left(\gamma_{a}{}^{[b}\right)_{\alpha \beta} \left(\gamma^{c]}\right)_{\eta}{}^{\gamma} + \frac{1}{5!4!}\epsilon^{b c [5] [4]} \left(\gamma_{[5]}\right)_{\alpha \beta} \left(\gamma_{a [4]}\right)_{\eta}{}^{\gamma} \\
    &\quad\quad\quad -\frac{1}{6!}\epsilon_{a}{}^{b c [5] [3]} \left(\gamma_{[5]}\right)_{\alpha \beta} \left(\gamma_{[3]}\right)_{\eta}{}^{\gamma} - \frac{1}{4!}\delta_{a}{}^{[b} \left(\gamma^{c] [4]}\right)_{\alpha \beta} \left(\gamma_{[4]}\right)_{\eta}{}^{\gamma} \\
    &\quad\quad\quad -\frac{1}{3!}\left(\gamma^{[3] b c}\right)_{\alpha \beta} \left(\gamma_{a [3]}\right)_{\eta}{}^{\gamma} + \frac{1}{3!}\left(\gamma_{[3] a}{}^{[b}\right)_{\alpha \beta} \left(\gamma^{c] [3]}\right)_{\eta}{}^{\gamma} \\
    &\quad\quad\quad - \frac{1}{2}\left(\gamma_{a}{}^{b c [2]}\right)_{\alpha \beta} \left(\gamma_{[2]}\right)_{\eta}{}^{\gamma} - \frac{1}{5!4!} \epsilon^{[4] [5]}{}_{a}{}^{[b} \left(\gamma_{[5]}\right)_{\alpha \beta} \left(\gamma^{c]}{}_{[4]}\right)_{\eta}{}^{\gamma} \\
    &\quad\quad\quad - \frac{1}{5!5!}\delta_{a}{}^{[b} \epsilon^{c] [5] [\overline{5}]} \left(\gamma_{[5]}\right)_{\alpha \beta} \left(\gamma_{[\overline{5}]}\right)_{\eta}{}^{\gamma}\bigg\}.
\end{align*}

\subsection{Symbolic Fourier Transform}\label{symbolic_fourier}
It is necessary in the later algorithms of this chapter to transition to and from momentum-space, so we present here algorithms for computing symbolic Fourier and inverse Fourier transforms. Our procedure for the former is based on a draft transform sketched in \cite{Peeters2018a}. We assume that every field depends on an independent spatial parameter, unless several fields lie within the operand of the same partial derivative, in which case the fields share the spatial parameter with respect to which the operand is being differentiated.\footnote{Note that the Fourier transform algorithm does not convert a field \(A\) to a marked Fourier transform \(\mathcal{F}_{N}(A)\) if \(A\) is not differentiated, for there is little purpose for the Fourier transform in that situation, and \(A\) might as well be taken to refer both to the field and its Fourier transform.} The algorithm ``Fourier'' (see Algorithm \ref{alg:fourier}) replaces every term in an inputted position-space expression \(\mathcal{E}\) with a new term, consisting of the coefficients of the old term, a momentum parameter \(i k_{N_{a}}\) for every partial derivative \(\partial_{a}\) in the old term, and a Fourier-transformed field \(\mathcal{F}_{N}(A)\) for every field \(A\) in the old term. Here, \(N\) is an integer identifying the corresponding partial derivative, and it identifies every momentum parameter with the Fourier transforms of the fields in the operand of the corresponding partial derivative.

\begin{algorithm}[ht!]
    \caption{Symbolic Fourier Transform}\label{alg:fourier}
    \begin{algorithmic}[1]
        \Function{Fourier}{\(\mathcal{E}\), \textit{fields}[\(1..n\)]}
            \For{\textit{term} in \(\mathcal{E}\)}
                \State \(N \gets 1\)

                \For{\textit{factor} in \textit{term}}

                    \If{\textit{factor} is a partial derivative}

                        \For{\textit{A} in factors of partial's operand in \textit{factor}}
                            \If{\textit{A} in \textit{fields}}
                                \State replace \textit{A} in \textit{factor} with \(\mathcal{F}_N\)(\textit{A})
                            \EndIf
                        \EndFor

                        \State \textit{ks} \(\gets\) [ ]
                        \For{index \(a\) on partial of \textit{factor}}
                            \If{\(a\) is contravariant}
                                \State \textbf{append} \(i k_N{}^a\) to \textit{ks}
                            \Else
                                \State \textbf{append} \(i k_{N_a}\) to \textit{ks}
                            \EndIf   
                        \EndFor

                        \State \textit{new\_prod} = coefficient of factor \(\cdot\) \textbf{prod}(\textit{ks}) \(\cdot\) partial's new operand
                        \State \textbf{insert} \textit{new\_prod} after \textit{factor} \Comment{Introduce \textit{new\_prod} as a factor of \textit{term}}
                        \State \textbf{erase} \textit{factor}
                        \State \(N \gets N + 1\)
                    \EndIf

                \EndFor
            \EndFor

            \State \(\mathcal{E} \gets\) \textbf{FlattenProd}(\(\mathcal{E}\)) \Comment{Remove unnecessary nesting of products}
            \State \Return \(\mathcal{E}\)

        \EndFunction
    \end{algorithmic}
\end{algorithm}

The use of this numbering enables inversion in ``InverseFourier'' (see Algorithm \ref{alg:invfourier}). If a set of Fourier-transformed fields have an identifier \(N\), then the (untransformed) fields are wrapped in a partial derivative with indices the same as those of the momentum parameters with identifier \(N\), those momentum parameters are deleted, and we multiply by \(-i\) for each such parameter.

\begin{algorithm}[ht!]
    \caption{Symbolic Inverse Fourier Transform}\label{alg:invfourier}
    \begin{algorithmic}[1]
        \Function{InverseFourier}{$\hat{\mathcal{E}}$}

            \For{\textit{term} in \(\hat{\mathcal{E}}\)}
                \State \(N \gets 1\)

                \While{term has \(k_N{}^\bullet\) or \(k_{N_\bullet}\) factors}
                    \State \(c \gets 1\)
                    \State \textit{operand} \(\gets 1\)
                    \State \textit{indices} \(\gets\) [ ]

                    \For{\textit{factor} in \textit{term}}
                        \If{factor has form \(k_N{}^\bullet\) or \(k_{N_\bullet}\)}
                            \State \(c \gets c \cdot (-i)~\cdot\) coefficient of \textit{factor}
                            \State \textbf{append} Lorentz index of \textit{factor} to \textit{indices}

                            \State \textbf{erase} \textit{factor}
                        \ElsIf{factor has form \(\mathcal{F}_N(\cdot)\)}
                            \State \(c \gets c~\cdot\) coefficient of \textit{factor}
                            \State \textit{operand} \(\gets\) \textit{operand} \(\cdot\) \(factor\)'s operand
                            \State \textbf{erase} \textit{factor}
                        \EndIf
                    \EndFor

                    \State \textbf{insert} \(c\cdot\partial_{\textnormal{\textit{indices}}}(\textnormal{\textit{operand}})\) to \textit{term}'s factors \Comment{individual indices can be co- or contravariant}

                    \State \(N \gets N + 1\)
                \EndWhile
            \EndFor

            \State \Return \(\hat{\mathcal{E}}\)
        \EndFunction
    \end{algorithmic}
\end{algorithm}

As an example, consider the position-space expression \(\partial_{a}(A B) \partial_{b c}C + \partial_{a b c}(A B C)\), where \(A, B, C\) are fields. The code below runs ``Fourier'' to convert to momentum space and canonicalizes the result to clean up coefficients.
\begin{minted}{python}
>>> ex = Ex(r'''\partial_{a}(A B) \partial_{b c}(C) + \partial_{a b c}(A B C)''')
>>> fourier(ex, [Ex('A'), Ex('B'), Ex('C')])
>>> canonicalise(ex)
\end{minted}
The result is
\begin{equation}
    -i k_{1_{a}} \mathcal{F}_{1}(A) \mathcal{F}_{1}(B) k_{2_{b}} k_{2_{c}} \mathcal{F}_{2}(C) - i k_{1_{a}} k_{1_{b}} k_{1_{c}} \mathcal{F}_{1}(A) \mathcal{F}_{1}(B) \mathcal{F}_{1}(C).
\end{equation}
One can check that the below code reverts back to the original expression.
\begin{minted}{python}
>>> inverse_fourier(ex)
>>> canonicalise(ex)
\end{minted}

\subsection{Identifying Non-Gauge-Invariant Quantities}\label{non_gauge_inv}
In the next section, we will find it necessary to eliminate gauge-invariant quantities from an expression. We present here an algorithm ``FindNonGaugeInv'' that, given an expression \(\mathcal{E}\), returns an expression, no combination of terms of which is gauge-invariant, that is gauge-equivalent\footnote{We say that an expression \(A\) is ``gauge-equivalent'' to another expression \(B\) if \(A - B\) is gauge-invariant.} to the set of non-gauge-invariant terms originally appearing in \(\mathcal{E}\) and \textit{with the same coefficients}. For example, consider the Kalb-Ramond field \(B_{a b}\) from \cite{Kalb1974} with gauge transformation\footnote{See \cite{Heisenberg2020} for a modern reference.}
\begin{equation}\label{kalb_ramond_gauge_trans}
    \delta_{G} B_{a b} = \partial_{a} \xi_{b} - \partial_{b} \xi_{a}.
\end{equation}
It is well-known and easy to verify that the quantity
\begin{equation}\label{kalb_ramond_field_strength}
    H_{a b c} = \frac12 \partial_{[a} B_{b c]},
\end{equation}
known as the Kalb-Ramond field strength, is gauge-invariant. Let
\begin{equation}\label{find_non_gauge_inv_ex:2}
    \begin{split}
        \mathcal{E} &= (u + v) \left(\gamma^{a}\right)_{\alpha \beta} \partial_{a} B_{b c} - v \left(\gamma^{a}\right)_{\alpha \beta} \partial_{b} B_{a c} + (u + v) \left(\gamma^{a}\right)_{\alpha \beta} \partial_{c} B_{a b} \\
        &= u \left(\gamma^{a}\right)_{\alpha \beta} \partial_{a} B_{b c} + u \left(\gamma^{a}\right)_{\alpha \beta} \partial_{c} B_{a b} + v H_{a b c}.
    \end{split}
\end{equation}
Then an acceptable answer is\footnote{In general, there may be several gauge-equivalent expressions satisfying the given desiderata. It will suffice for our purposes to output any one of these.}
\begin{equation}\label{find_non_gauge_inv_ex_sol:2}
    \mathcal{E}' = u \left(\gamma^{a}\right)_{\alpha \beta} \partial_{b} B_{a c}.
\end{equation}
Indeed, this is gauge-equivalent to the non-gauge-invariant part of \(\mathcal{E}\) in \eqref{find_non_gauge_inv_ex:2} with the same coefficient \(-u\),\footnote{In general, by ``having the same coefficient,'' we mean in the sense that the coefficients are alike, so that the difference between \(\mathcal{E}'\) and the non-gauge-invariant part of \(\mathcal{E}\) is the multiple of a gauge-invariant expressions, in this case \(H_{a b c}\).} for
\begin{equation}
    u \left(\gamma^{a}\right)_{\alpha \beta} \partial_{b} B_{a c} - \left[u \left(\gamma^{a}\right)_{\alpha \beta} \partial_{a} B_{b c} + u \left(\gamma^{a}\right)_{\alpha \beta} \partial_{c} B_{a b}\right] = -u H_{a b c}.
\end{equation}

How does one go about computing \(\mathcal{E}'\)? We can eliminate gauge-invariant combinations of terms by merely replacing every instance of our desired field with its gauge transformation.\footnote{``FindNonGaugeInv'' also takes as input a full list of fields in the expression, in addition to the desired field. This is merely in order to apply the Fourier transform correctly.} It is convenient to move to momentum space to obviate subtleties concerning differentiation. (This is really more of a technical programmatic matter than a theoretical one.) In the Kalb-Ramond example, one finds
\begin{equation}\label{fourier_gauge_E}
    \mathcal{F}(\delta_{G} \mathcal{E}) = -u \left(\gamma^{a}\right)_{\alpha \beta} k_{a} k_{b} \mathcal{F}(\xi_{c}) + u \left(\gamma^{a}\right)_{\alpha \beta} k_{b} k_{c} \mathcal{F}(\xi_{a}).
\end{equation} Programmatically reversing the substitution of the gauge transformation is a bit subtle and comprises much of the algorithm. Let \(\mathcal{B}\) be the set of expressions obtained by taking the momentum-space equivalent of \(-i\) times the partial derivative of the desired field and giving it each permutation of Lorentz indices of each term in \(\mathcal{E}\), and let \(\hat{\mathcal{B}}\) be the analogous set of expressions obtained from negative the partial derivative of the gauge transformation of the desired field. In the Kalb-Ramond example,
\begin{equation}
    \begin{split}
        \mathcal{B} &= \{k_{\sigma_1} \mathcal{F}(B_{\sigma_2 \sigma_3}) ~|~ \sigma \in \mathfrak{S}(\{a,b,c\})\} \\
        \hat{\mathcal{B}} &= \{k_{\sigma_1} k_{\sigma_2} \mathcal{F}(\xi_{\sigma_3}) - k_{\sigma_1} k_{\sigma_3} \mathcal{F}(\xi_{\sigma_2}) ~|~ \sigma \in \mathfrak{S}(\{a, b, c\})\},
    \end{split}
\end{equation}
where \(\mathfrak{S}(\cdot)\) is the set of permutations. We may enumerate \(\mathcal{B}, \hat{\mathcal{B}}\) as \(\mathcal{B} = \{\mathcal{G}'_j\}_{j \in J}\), \(\hat{\mathcal{B}} = \{\mathcal{H}'_j\}_{j \in J}\). By construction, \(\mathcal{F}(\delta_G \mathcal{E})\) is in the span of \(\hat{\mathcal{B}}\), i.e., there are coefficients \(c_j\), \(j \in J\), in the Clifford algebra such that
\begin{equation}\label{zero_ex}
    \mathcal{F}(\delta_G \mathcal{E}) + \sum_{j \in J} c_j \mathcal{H}'_j = 0.
\end{equation}
``FindNonGaugeInv'' begins (see Algorithm \ref{alg:fngi1}) by computing \(\mathcal{B},\hat{\mathcal{B}}\) and constructing the left-hand side of \eqref{zero_ex} with appropriate indices on the coefficients.

\begin{algorithm}[ht!]
    \caption{Find Non Gauge Invariant Part of \(\mathcal{E}\) Part 1}\label{alg:fngi1}
    \begin{algorithmic}[1]
        \Function{FindNonGaugeInv}{\(\mathcal{E}\), \(A\), \textit{fields}[\(1..n\)], $\delta_G \mathcal{A}$}

            \State \(\mathcal{E}\gets\) \textbf{substitute}(\(\mathcal{E}\), \(A \rightarrow \delta_G \mathcal{A}\))
            \State \(\mathcal{E}\gets\) \textbf{Evaluate}(\(\mathcal{E}\), \textit{to\_perform\_subs} = False)
            \State \(\zeta \gets\) gauge parameter in \(\delta_G A\) \Comment{e.g., \(\xi\) in \(\delta_G B_{a b}\)}
            \State \(\mathcal{E}\gets\) \textbf{Fourier}(\(\mathcal{E}\), \(fields + [\zeta]\))

            \State \textit{Inds}\({}_{kf} \gets\) [ ]
            \For{\textit{term} in \(\mathcal{E}\)}
                \State \textit{LInds}\({}_t \gets\) lexicographically sorted list of Lorentz indices in \textit{term}'s ``k1''s and \(\zeta\)s
                \State \textit{SInds}\({}_t \gets\) spinor index on \(\zeta\)
                \State \textbf{append} (\textit{LInds}\({}_t\), \textit{SInds}\({}_t\)) to \textit{Inds}\({}_{kf}\)
            \EndFor

            \State \textbf{remove} duplicates in \textit{Inds}\({}_{kf}\)

            \State \(\mathcal{G} \gets -i \partial_{a'} A\) \Comment{e.g., \(-i \partial_{a'} B_{a b}\)}
            \State \(\mathcal{H} \gets -\partial_{a'}(\delta_{G} A)\) \Comment{e.g., \(-\partial_{a'} (\partial_{a} \xi_{b} - \partial_{b} \xi_{a})\)}
            \State \textbf{distribute}(\(\mathcal{E}\))

            \State \(\mathcal{G} \gets\) \textbf{Fourier}(\(\mathcal{G}\), [\(A\)]) \Comment{i.e., \(k_{1_{a'}} \mathcal{F}_1 (A)\), e.g., \(k_{1_{a'}} \mathcal{F}_1(B_{a b})\)}
            \State \(\mathcal{H} \gets\) \textbf{Fourier}(\(\mathcal{H}\), [\(\zeta\)]) \Comment{e.g., \(k_{1_{a'}} k_{1_{a}} \mathcal{F}_1(\xi_{b}) - k_{1_{a'}} k_{1_{b}} \mathcal{F}_1(\xi_{a})\)}

            \State \(\mathcal{E'} \gets 0\)
            \State \textit{zero\_ex} \(\gets \mathcal{E}\)
            \State \(C \gets\) [ ]

            \State \(j \gets 0\)
            \For{\textit{LInds}\({}_t\), \textit{SInds}\({}_t\) in \textit{\textit{Inds}}\({}_{kf}\)}
                \State \textit{c\_linds} \(\gets\) Lorentz indices in \(\mathcal{E}\) that are not in \textit{LInds}\({}_t\)
                \State \textit{c\_sinds} \(\gets\) spinor indices in \(\mathcal{E}\) that are not in \textit{SInds}\({}_t\)

                \For{\(\sigma \in \mathfrak{S}(\)\textit{LInds}\({}_t)\)}
                    \State \(\mathcal{G}'_{j} \gets \mathcal{G}\) with its Lorentz indices replaced by \(\sigma\) and its spinor index replaced by \textit{SInds}\({}_t\)
                    \State \(\mathcal{H}'_{j} \gets \mathcal{H}\) with its Lorentz indices replaced by \(\sigma\) and its spinor index replaced by \textit{SInds}\({}_t\)

                    \State \(c_j \gets (c_{j_{\textit{LInds}_{t}}})_{\textit{SInds}_{t}}\)
                    \State \(\mathcal{E}' \gets \mathcal{E}'+c_j \mathcal{G}'_{j}\)
                    \State \textit{zero\_ex} \(\gets\) \textit{zero\_ex} \(+c_j \mathcal{H}'_{j}\) \Comment{Introduce variable linear combination of the \(\mathcal{H}'_{j}\)s}
                    \State \textbf{append} \(c_j\) to \(C\) \Comment{Store the variables for which to solve}
                    
                    \State \(j \gets j + 1\)
                \EndFor
            \EndFor
            \State continued...
            \algstore{bkbreak}
    \end{algorithmic}
\end{algorithm}

It remains to solve for the coefficients. We set up \eqref{zero_ex} as a linear system \(A\mathbf{x} = \mathbf{b}\). For example, a term\footnote{Note that we must factor out every instance of the Fourier transform of the gauge parameter and every corresponding momentum parameter to obtain terms of this form.}
\begin{equation}
    \left(5c_0 + c_2 - \left(\gamma^{a}\right)_{\alpha \beta}\right) k_{a} k_{b} \mathcal{F}(\xi)
\end{equation}
contributes a row \((5,0,1,0,\ldots,0)\) to \(A\) and an element \(\left(\gamma^{a}\right)_{\alpha \beta}\) to \(\mathbf{b}\). Notice that although \(\mathbf{b}\) has entries that are general Clifford algebra elements, \(A\), which consists of scalar multipliers of the coefficients \(c_j\), has entries in \(\mathbb{Q}\). Hence, we may proceed largely as in ordinary linear algebra: if \([A ~|~ I]\) (\(I\) the appropriate identity matrix) row-reduces to \([R ~|~ P]\), then we may find a solution for the coefficients by the standard pivot analysis of \(R\mathbf{x} = P\mathbf{b}\). It is then easy to verify that our desired \(\mathcal{E}'\) satisfies
\begin{equation}
    -i \mathcal{F}(\mathcal{E}') = \sum_{j \in J} c_j \mathcal{G}'_j.
\end{equation}
``FindNonGaugeInv'' computes \(\mathcal{E}'\) via this linear algebra and an inverse Fourier transform (see Algorithm \ref{alg:fngi2}).

\begin{algorithm}[ht!]
    \caption{Find Non Gauge Invariant Part of \(\mathcal{E}\) Part 2}\label{alg:fngi2}
    \begin{algorithmic}[1]
            \algrestore{bkbreak}
            \State \textit{zero\_ex} \(\gets\) \textbf{distribute}(\textit{zero\_ex})
            \State \textit{zero\_ex} \(\gets\) \textbf{canonicalise}(\textit{zero\_ex})
            \State \textit{zero\_ex} \(\gets\) \textbf{factor\_out}(\textit{zero\_ex}, all ``\(k_1\)''s \& ``\(\mathcal{F}_1\)''s, \textit{right} = True)

            \State \(A \gets\) [ ]
            \State \(\mathbf{b} \gets\) [ ]
            \For{\textit{term} in \textit{zero\_ex}}
                \For{\textit{factor} in \textit{term}}
                    \If{\textit{factor} is a sum or in \(C\)}
                        \State \(\textbf{a} \gets [0, \cdots, 0]\)
                        \For{\textit{sub\_term} in terms of \textit{factor}}
                            \If{\textit{sub\_term} \(\in C\)}
                                \State \(\mathbf{a}\)[index on \textit{sub\_term}] = multiplier of \textit{sub\_term} \Comment{e.g., \(\mathbf{a}[0] = 5\) if \textit{sub\_term} \(= 5c_0\)}
                                \State \textbf{erase} \textit{sub\_term} from \textit{factor} \Comment{Remove the discovered \(c_{j}\)}
                            \EndIf
                        \EndFor
                        \State \textbf{append} \(\mathbf{a}\) to \(A\) \Comment{e.g., \(\mathbf{a} = (5, 0, 1, 0, \dotsc, 0)\) if \(5c_0, c_2\) are among the subterms}
                        \State \textbf{append} [\(-\)\textit{factor}] to \(\mathbf{b}\) \Comment{e.g., \((\gamma^{a})_{\alpha \beta}\) if \(factor = -(\gamma^{a})_{\alpha \beta}\) after removing the \(c_{j}\)s}
                    \EndIf
                \EndFor
            \EndFor

            \State \([R \mid P] \gets\) RREF of \([A \mid I]\)
            \State \(P\mathbf{b} \gets P * \mathbf{b}\) \Comment{\(*\) denotes matrix multiplication}

            \State \textit{subs} \(\gets\) [ ]
            \For{\textit{row} in \(R\)}
                \State \textbf{append} ``C[index of pivot in \textit{row}] \(\rightarrow P\mathbf{b}\)(\textit{row})''
            \EndFor
            \State \(\mathcal{E}'\gets\) \textbf{substitute}(\(\mathcal{E}'\), \textit{subs}) \Comment{Replace the \(c_{j}\)s in \(\mathcal{E}'\) with their values}
            \State \textbf{set} all remaining \(c_j\)s in \(\mathcal{E}'\) to \(0\)
            \State \(\mathcal{E}'\gets i\mathcal{E}'\) \Comment{The correction for the sign and one fewer partial derivative}
            \State \(\mathcal{E}'\gets\) \textbf{distribute}(\(\mathcal{E}'\))
            \State \(\mathcal{E}'\gets\) \textbf{InverseFourier}(\(\mathcal{E}'\))
            \State \(\mathcal{E}'\gets\) \textbf{Evaluate}(\(\mathcal{E}'\), \textit{to\_perform\_subs} = False)

            \State \Return \(\mathcal{E}'\)

        \EndFunction
    \end{algorithmic}
\end{algorithm}

As an example, the below code finds the non-gauge-invariant part of \eqref{find_non_gauge_inv_ex:2}.
\begin{minted}{python}
>>> field = Ex(r'B_{a b}')
>>> gauge_trans = Ex(r'\partial_{a}(\zeta_{b}) - \partial_{b}(\zeta_{a})')

>>> ex = Ex(r'(u + v) (\Gamma^{a})_{\alpha \beta} \partial_{a}(B_{b c}) - v (\Gamma^{a})_{\alpha \beta} \partial_{b}(B_{a c}) + (u + v) (\Gamma^{a})_{\alpha \beta} \partial_{c}(B_{a b})')
>>> find_non_gauge_inv(ex, field, [field], gauge_trans)
>>> canonicalise(ex)
\end{minted}
The result is precisely \eqref{find_non_gauge_inv_ex_sol:2}.

\subsection{Multiplet Solver 1: Brute-Force Gamma-Matrix Splitting}\label{multiplet_solver1}
We can now begin treating the solution of multiplets. A large part of this is obtaining constraints from closure of the algebra, for which we provide a first algorithm in this section. (An alternative algorithm is provided in \S\ref{multiplet_solver2}.) Recall that by ``closure of the algebra,'' we mean the requirement that the anticommutator of supercovariant derivatives\footnote{We occasionally call the anticommutator of supercovariant derivatives applied to a field the ``closure" on that field.} be equal to a special translation up to ignorable terms, viz.,\footnote{The coefficient of the translation varies with convention and dimension. In 4D multiplets, it is typical to choose \(c = 2\), while in 11D, we take \(c = 1\).}
\begin{equation}\label{closure_condition}
    \{{\rm D}_{\alpha}, {\rm D}_{\beta}\} = c \cdot i \left(\gamma^{a}\right)_{\alpha \beta} \partial_{a} + \textnormal{other terms}.
\end{equation}
If the field to which the supercovariant-derivative anticommutator is applied is a gauge field, ``other terms'' will generally include gauge terms, which can be identified by the presence of a free index on the partial derivative in the term. As summarized in \cite{Gates2002}, if the multiplet is off-shell, then these are the only additional terms (as auxiliary fields do away with the equations of motion), while if the multiplet is on-shell, then the extra terms may also include terms arising from the equations of motion of the field.\footnote{It is perhaps more sound to frame this in the opposite way, namely, that if non-closure terms exist in the anticommutator of supersymmetry transformations, then the requirement that they identically vanish imposes equations of motion on the field, rendering the theory on-shell. These equation-of-motion terms are occasionally referred to as ``central charges,'' as in \cite{Gates2020a}, reflecting in spirit the additional possible transformations proven (originally in 4D) permissible in \cite{Haag1975}. However, often in the literature, the term ``central charge'' is reserved for non-closure terms involving fields other than the field on which the anticommutator of supercovariant derivatives is being calculated. Hence, to eliminate any ambiguity, we simply refer to these non-closure terms as ``equation-of-motion terms.''} In the case of fermions, which we consider in this paper, equation-of-motion terms are precisely terms which vanish under the equations of motion, which we call ``off-shell equation-of-motion terms.''\footnote{In \cite{Gates2020a}, such terms are referred to as ``off-shell central charges.''}

In particular, we assume in the tools in this paper that any equation-of-motion terms are on the closure on Lorentz-symmetric fermions, which already offers flexibility much broader than is generally needed, since fermions of interest in \(\mathcal{N} = 1\) supersymmetry are generally spin-\(1/2\) particles and spin-\(3/2\) particles, which have zero and one vector indices, respectively, and are thus vacuously Lorentz-symmetric. The study of the equations of motion for (symmetric) fermions of arbitrary spin (i.e., arbitrary numbers of vector indices) can be traced to \cite{Fang1978}, but we rely on the more recent formulation in \cite{Miyamoto2011}. Let the rank-\(s\) tensor-spinor \(\Psi_{a_1 \cdots a_{s}}{}^{\gamma}\) be a fermion of spin \(s + 1/2\), and let \(U_{i} = \{a_1, \dotsc, a_{i-1}, a_{i+1}, \dotsc, a_{s}\}\), so that \(\Psi_{b U_{i}}{}^{\gamma} = \Psi_{b a_1 \cdots a_{i-1} a_{i+1} \cdots a_{s}}{}^{\gamma}\). The gauge transformation of the fermion is ((3.3.1) in \cite{Miyamoto2011})
\begin{equation}\label{sym_fermion_gauge}
    \delta_{G} \Psi_{a_1 \cdots a_{s}}{}^{\gamma} = \sum_{i=1}^{s} \partial_{a_{i}} \kappa_{U_{i}}{}^{\gamma},
\end{equation}
where the gauge parameter \(\kappa_{U_{i}}{}^{\gamma}\) is a rank-(\(s-1\)) tensor-spinor assumed to satisfy ((3.3.2) in \cite{Miyamoto2011})
\begin{equation}\label{gauge_param_trace}
    \left(\gamma^{b}\right)_{\eta}{}^{\gamma} \kappa_{b b_3 \cdots b_{s}}{}^{\eta} = 0.
\end{equation}
The equation of motion of the spin-(\(s + 1/2\)) fermion is ((3.3.4) in \cite{Miyamoto2011})
\begin{equation}\label{sym_fermion_eqom}
    \left(\gamma^{b}\right)_{\eta}{}^{\gamma} \partial_{b} \Psi_{a_1 \cdots a_{s}}{}^{\eta} - \sum_{i = 1}^{s} \left(\gamma^{b}\right)_{\eta}{}^{\gamma} \partial_{a_{i}} \Psi_{b U_{i}}{}^{\eta} = 0.
\end{equation}
Notice that the \(s = 0\) and \(s = 1\) cases are precisely the Dirac and Rarita-Schwinger equations, respectively. \eqref{sym_fermion_eqom} is, of course, gauge-invariant, and proposition \ref{arb_spin_ferm_prop:1} verifies that every gauge-invariant combination of terms of the form
\begin{equation}\label{lorentz_proper_form:1}
    \left(\gamma^{b}\right)_{\eta}{}^{\gamma} \partial_{a'} \Psi_{E}{}^{\eta}, \quad b \in \{a'\} \cup E,
\end{equation}
is proportional to the left-hand side of \eqref{sym_fermion_eqom}. Also, proposition \ref{arb_spin_ferm_prop:2} and corollary \ref{arb_spin_ferm_cor} show that
\begin{equation}\label{lorentz_proper_form:2}
    \left(\gamma^{a' E}\right)_{\eta}{}^{\gamma} \partial_{a'} \Psi_{E}{}^{\eta}
\end{equation}
can be considered an equation-of-motion term. While it could be reduced to a multiple of \eqref{sym_fermion_eqom}, its empirically frequent occurrence when considering spin-\(3/2\) fermions merits the efficiency gain of skipping the reduction and considering it separately.

The general idea of the multiplet solver is to use the foregoing considerations to sift through the terms in the Fierz expansion of the anticommutator of supercovariant derivatives applied to each field;\footnote{The Fierz expansion is needed to ensure that the spinor indices of the supercovariant derivatives are attached to the same gamma matrices, as will be discussed in a moment.} recognize the desired translation, gauge terms, (off-shell) equation-of-motion terms, and other ``undesired terms''; record the equation setting the coefficient of the desired translation equal to \(c \cdot i\) (see \eqref{closure_condition}); and record the equations arising from setting the coefficients of the undesired terms equal to zero. Independent constraints are then culled from the equations produced while sifting through the terms.\footnote{Notice that the sole purpose of recognizing gauge and equation-of-motion terms is to allow the coefficients of these terms to be anything, thereby preventing the addition of unwarranted equations to the ultimate list of constraints.} The task of sifting through the terms is handled, in part, by an algorithm ``FilterTerms'' (see Algorithm \ref{alg:filterterms}). The algorithm takes as input the expression \(\mathcal{E}\) for the Fierz-expanded closure on a field, the field \(A\) in question (which may be bosonic or fermionic), the indices ``Inds'' of the supercovariant derivatives, and an optional boolean parameter, by default false, dictating whether to look for potential off-shell equation-of-motion terms. The algorithm loops through an evolving array of terms, initially those in \(\mathcal{E}\), and for each term, proceeds according to five possible cases. 1) If the term is of the form of the desired translation (up to coefficients and indices), then it is classified as a desired term. 2) If the procedure is permitted to identify potential equation-of-motion terms and the term is of the form \((\gamma^{D})_{Inds} (\gamma^{a' E F})_{\alpha'}{}^{\beta'} \partial_{a'}(A_{E})_{\beta'}\) for some \(\alpha'\),\footnote{N.B. Cases (2) and (3) assume that \(A\) is a fermion. \(\alpha'\) is the free spinor index of the fermion before the supercovariant derivatives are applied, but this fact on its own has no utility in ``FilterTerms,'' with the fermion's free spinor index not even inputted. In the algorithm for ``SUSYSolve,'' whenever \(A\) is a boson, ``identify\_lorentz\_proper'' will be set false so that cases (2) and (3) are not considered.} then the second gamma matrix can be decomposed\footnote{The procedure for executing this decomposition is actually quite intricate and technical and is suppressed here. Fundamentally, the procedure relies on the Cadabra2 \cite{Peeters2018b} function \mintinline{python}{split_gamma()} (which is incidentally similar to the FieldsX \cite{Frob2021} function \mintinline{python}{SplitGammaMatrix()}), which splits off a 1-gamma matrix from a multi-indexed gamma matrix. The subtlety is in splitting gamma matrices, and then joining them, in a consistent way to yield \(\gamma^{F} \gamma^{a' E}\) given the desired indices \(\{a'\} \cup E\) on one factor.\label{split_gamma_footnote}} to give \(\gamma^{F} \gamma^{a' E}\) and some remainder, yielding a term that is a multiple of \eqref{lorentz_proper_form:2} and is classified as ``Lorentz proper'' (in the sense of having the appropriate Lorentz-index structure), and a series of residual terms that are merely thrown back into the loop for further evaluation. 3) Same as (2) but we have the form \((\gamma^{D})_{Inds} (\gamma^{E})_{\alpha'}{}^{\beta'} \partial_{a'}(A_{F})_{\beta'}\), where \(E \cap (\{a'\} \cup F) \neq 0\), so we sort this intersection lexicographically, draw the last index \(g\),\footnote{The only reason to choose \(g\) in this way is to ensure that it is canonical} decompose the second gamma matrix as \(\gamma^{E \setminus \{g\}} \gamma^{g}\) and some remainder, and proceed as in (2). Here, the ``Lorentz proper'' term is instead a multiple of \eqref{lorentz_proper_form:1}. 4) If conditions (1)-(3) are not met but the term's partial derivative carries a free index, then the term is classified as a gauge term. 5) If none of (1)-(4) are met, then the term is classified as an undesired term. The loop eventually terminates because gamma matrices may be split only finitely many times, so cases (2) and (3) may occur only finitely many times. The isolation of (potential) equation-of-motion terms in cases (2) and (3) by means of explicit gamma-matrix decompositions is what we dub ``brute-force gamma-matrix splitting,'' in contrast to the method in \S\ref{multiplet_solver2}.

\begin{algorithm}[ht!]
    \caption{Filter Closure/Non-Closure Terms}\label{alg:filterterms}
    \begin{algorithmic}[1]
        \Function{FilterTerms}{\(\mathcal{E}\), \(A\), \textit{Inds}[\(1..2\)], \textit{identify\_lorentz\_proper} = False}
            \State \textit{desired\_terms} \(\gets\) [ ]
            \State \textit{gauge\_terms} \(\gets\) [ ]
            \State \textit{undesired\_terms} \(\gets\) [ ]
            \State \textit{lorentz\_proper\_terms} \(\gets\) [ ]

            \State \(T \gets\) terms of \(\mathcal{E}\)

            \State \(n = 0\)
            \While{\(n < |T|\)}
                \State \(t \gets T[n]\)

                \If{\(t\) has form \((\gamma^{a'})_{\textit{Inds}} \partial_{a'} A\)}
                    \State \textbf{append} \(t\) to \textit{desired\_terms}

                \ElsIf{\textit{identify\_lorentz\_proper} \& \(t\) has form \((\gamma^D)_{\textit{Inds}} (\gamma^{a' E F})_{\bullet}{}^{\beta'} \partial_{a'} (A_E)_{\beta'}\)}
                    \State \(\gamma^{a' E F} \rightarrow \gamma^{F}\gamma^{a' E} + R\)
                    \State \(\mathcal{R} \gets (\gamma^D)_{\textit{Inds}} (R)_{\bullet}{}^{\beta'} \partial_{a'} (A_E)_{\beta'}\)
                    \State  \(\mathcal{R} \gets\) \textbf{Evaluate}(\(\mathcal{R}\), \textit{to\_perform\_subs} = False)
                    \State \textbf{append} \((\gamma^D)_{\textit{Inds}} (\gamma^{F}\gamma^{a' E})_{\bullet}{}^{\beta'} \partial_{a'} (A_E)_{\beta'}\) to \textit{lorentz\_proper\_terms}
                    \State \textbf{append} terms of \(\mathcal{R}\) to \(T\)

                \ElsIf{\textit{identify\_lorentz\_proper} and \(t\) has form \((\gamma^D)_{\textit{Inds}} (\gamma^{E})_{\bullet}{}^{\beta'} \partial_{a'} (A_F)_{\beta'}\) \& \(E \cap (\{a'\} \cup F) \ne \emptyset\)}

                    \State \(G \gets \{\textnormal{last index of the sorted array } E \cap (\{a'\} \cup F)\}\) \Comment{\(G = \{g\}\), with \(g\) from the discussion above}
                    \State \(\gamma^{E} \rightarrow \gamma^{E \setminus G}\gamma^{G} + R\)
                    \State \(\mathcal{R} \gets (\gamma^D)_{\textit{Inds}} (R)_{\bullet}{}^{\beta'} \partial_{a'} (A_F)_{\beta'}\)
                    \State  \(\mathcal{R} \gets\) \textbf{Evaluate}(\(\mathcal{R}\), \textit{to\_perform\_subs} = False)
                    \State \textbf{append} \((\gamma^D)_{\textit{Inds}} (\gamma^{E \setminus G}\gamma^{G})_{\bullet}{}^{\beta'} \partial_{a'} (A_F)_{\beta'}\) to \textit{lorentz\_proper\_terms}
                    \State \textbf{append} terms of \(\mathcal{R}\) to \(T\)

                \ElsIf{the \(\partial\) in \(t\) has a free index}
                    \State \textbf{append} \(t\) to \textit{gauge\_terms}

                \Else:
                    \State \textbf{append} \(t\) to \textit{undesired\_terms}

                \EndIf
                \State \(n \gets n + 1\)
            \EndWhile

            \State \Return \{``desired\_terms'': \textit{desired\_terms}, ``gauge\_terms'': \textit{gauge\_terms}, ``undesired\_terms'': \textit{undesired\_ terms}, ``lorentz\_proper\_terms'': \textit{lorentz\_proper\_terms}\}

        \EndFunction
    \end{algorithmic}
\end{algorithm}

We can now treat the (first) multiplet solver algorithm ``SUSYSolve.'' The algorithm takes as input the array ``bosons'' of bosonic fields in the multiplet, the array ``fermions'' of fermionic fields in the multiplet, the corresponding array ``gauge\_transs'' of gauge transformations for the fermions, the supersymmetry transformation rules ``susy'' (e.g., \eqref{11d_sugra}) for whose coefficients it is desired to solve, a basis \(\Gamma^{A}\) for the Clifford algebra (e.g., \eqref{11d_basis_redux}),\footnote{For consistency of notation, we always write \(\Gamma^{A}\) for the basis of the Clifford algebra. Note that this basis is always unrelated to any field that might be called \(A\).} an array ``consts'' of the unknown coefficients, a length-2 array ``Inds'' of indices for the supercovariant derivatives, and optionally the desired value of \(c\) in \eqref{closure_condition}. (By default, \(c = 1\).) The algorithm is largely a loop through the bosonic and fermionic fields. For a particular field field \(A\), the closure on \(A\) is evaluated via a subprocedure ``SusyExpand''\footnote{``SUSYExpand'' is a function which given an expression and supersymmetry transformation rules expands every instance in the expression of the supercovariant derivative applied to a field. Since the algorithm behind it is of little conceptual interest, and most of its contents consider rather technical complexities in the passing of supercovariant derivatives past partial derivatives in Cadabra, detailed discussion of this algorithm is suppressed.}, it is Fierz-transformed via ``FierzExpand2Index" to place the supercovariant derivatives' indices on the same factors for comparison with \eqref{closure_condition}, and its terms are filtered via ``FilterTerms."\footnote{Notice in Algorithm \ref{alg:susysolve1} that ``FilterTerms'' is set to look for potential equation-of-motion terms if and only if \(A\) is a fermion, i.e., \(A\) has a spinor index, or equivalently, it \textit{is} an indexbracket. \label{filterterms_fermion_expl}} This first part of the procedure (together with unpacking of the filtered terms)\footnote{Notice that in the unpacking of ``'Lorentz proper'' terms, multiplication of gamma matrices is not allowed, in order to prevent the reversal of the gamma-matrix decompositions in cases (2) and (3) of ``FilterTerms.''} is shown in pseudo-code as Algorithm \ref{alg:susysolve1}.

\begin{algorithm}[ht!]
    \caption{Find Coefficients That Enforce Multiplet Closure Part 1}\label{alg:susysolve1}
    \begin{algorithmic}[1]
        \Function{SUSYSolve}{\textit{bosons}[\(1..n\)], \textit{fermions}[\(1..m\)], \textit{gauge\_transs}[\(1..m\)], \textit{susy}, \(\Gamma^A\), \textit{consts}[\(1..p\)], \textit{Inds}[\(1..2\)], \(c\) = 1}

            \State \textit{susy\_dict} \(\gets\) empty hash table
            \State \(S \gets\) [ ]

            \For{\(A \in \) \textit{fermions} \(\cup\) \textit{bosons}}
                \State \(\delta_G A \gets \) \textit{gauge\_transs}\((A)\) \textbf{if} \(A \in \) \textit{fermions}
                \State \(\mathcal{E} \gets {\rm D}_{\textit{Ind}[0]}{\rm D}_{\textit{Ind}[1]}A + {\rm D}_{\textit{Ind}[1]}{\rm D}_{\textit{Ind}[0]}A\)
                \State \(\mathcal{E} \gets\) \textbf{SUSYExpand}(\(\mathcal{E}\), \textit{susy})
                \State \(\mathcal{E} \gets\) \textbf{Evaluate}(\(\mathcal{E}\), \textit{to\_perform\_subs} = False)
                \State \(\mathcal{E} \gets\) \textbf{FierzExpand2Index}(\(\mathcal{E}\), \(\Gamma^A\), \textit{Inds}, \textit{to\_perform\_subs} = False)

                \State \(T_0 \gets\) \textbf{FilterTerms}(\(\mathcal{E}\), \(A\), \textit{Inds}, \textit{identify\_lorentz\_proper} = \(A \in fermions\))

                \State \textit{desired\_terms} \(\gets T_0\)[``desired\_terms'']
                \State \textit{gauge\_terms} \(\gets T_0\)[``gauge\_terms'']
                \State \textit{undesired\_terms} \(\gets T_0\)[``undesired\_terms'']
                \State \textit{lorentz\_proper\_terms} \(\gets T_0\)[``lorentz\_proper\_terms'']

                \State \textit{desired\_terms} \(\gets\) \textbf{Evaluate}(\textbf{sum}(\textit{desired\_terms}))
                \State \textit{pre\_gauge\_terms} \(\gets\) \textbf{Evaluate}(\textbf{sum}(\textit{gauge\_terms}))
                \State \textit{pre\_undesired\_terms} \(\gets\) \textbf{Evaluate}(\textbf{sum}(\textit{undesired\_terms}))
                \State \textit{lorentz\_proper\_exp} \(\gets\) \textbf{Evaluate}(\textbf{sum}(\textit{lorentz\_proper\_terms}), \textit{to\_perform\_subs} = False, \textit{to\_join \_gamma} = False)
                
                \State \textit{pre\_desired\_terms} \(\gets\) \textbf{factor\_in}(\textit{pre\_desired\_terms}, \textit{consts})
                \State \textit{pre\_gauge\_terms} \(\gets\) \textbf{factor\_in}(\textit{pre\_gauge\_terms}, \textit{consts})
                \State \textit{pre\_undesired\_terms} \(\gets\) \textbf{factor\_in}(\textit{pre\_undesired\_terms}, \textit{consts})
                \State continued...
                \algstore{bkbreak}
    \end{algorithmic}
\end{algorithm}

If \(A\) is a fermion, so that we have ``Lorentz proper'' terms (see footnote \ref{filterterms_fermion_expl}), more processing is necessary. Recall that these terms are multiples of the forms \eqref{lorentz_proper_form:1} and \eqref{lorentz_proper_form:2}; gauge-invariant combinations of terms of the former type comprise equation-of-motion terms, while terms of the latter type are equation-of-motion terms and are already gauge-invariant. Hence, ``FindNonGaugeInv" eliminates the equation-of-motion terms, and the remaining ``Lorentz proper" terms may be refiltered (without searching for equation-of-motion terms, so no ``Lorentz proper'' terms remain).

The classified terms of the field are adjoined to a hash table with the field as the key. Further, the terms of each type are added together and factored in terms of ``consts,'' and constraints are drawn from coefficients of the desired and undesired terms in the manner noted earlier (e.g., if \(c = 1\), our sole desired term is \((u v - x z) (\gamma^{a})_{\alpha \beta} \partial_{a} A_{b}\), and our sole undesired term is \(5y z (\gamma^{b c})_{\alpha \beta} (\gamma_{a})_{\eta}{}^{\gamma} \partial_{b} \Psi_{c}{}^{\eta}\), then we have constraints \(u v - x z = i, 5y z = 0\)). After compiling the constraints from all fields, a subprocedure ``DistillConstrs''\footnote{``DistillConstrs'' is too technical to be described in detail here, but to give a word on its functioning, the algorithm is a linear solver which solves for the values of the pairwise products of elements in ``consts.'' Note that this procedure places a tacit requirement on the inputted multiplet that each term in each part of the supersymmetry transformation rule has exactly one unknown variable in its coefficient, so that every term in the closure includes the product of exactly two elements from ``consts.'' This, of course, is not at all restrictive, for any additional unknown variable in the coefficient of a term in the supersymmetry transformation rules would necessarily be redundant.} is applied to this system to reduce to independent constraints, which are outputted together with the earlier-mentioned hash table. This second part of the procedure is shown in pseudo-code as Algorithm \ref{alg:susysolve2}.

\begin{algorithm}[ht!]
    \caption{Find Coefficients That Enforce Multiplet Closure Part 2}\label{alg:susysolve2}
    \begin{algorithmic}[1]
                \algrestore{bkbreak}
                \If{\(A \in fermions\)}
                    \State \textit{added\_terms} \(\gets\) \textbf{FindNonGaugeInv}(copy of \textit{lorentz\_proper\_exp}, \(A\), \textit{bosons} \(\cup\) \textit{fermions}, \(\delta_G A\)) \textbf{if} \(\delta_G A \ne 0\) \textbf{else} 0
                    \State \(T_1 \gets\) \textbf{FilterTerms}(\textit{added\_terms}, \textit{field}, \textit{Inds}, \textit{identify\_lorentz\_proper} = False)

                    \State \textit{gauge\_terms} \(\gets\) \textit{gauge\_terms} + \(T_1\)[``gauge\_terms'']
                    \State \textit{undesired\_terms} \(\gets\) \textit{undesired\_terms} + \(T_1\)[``undesired\_terms'']
    
                    \State \textit{gauge\_terms} \(\gets\) \textbf{Evaluate}(\textbf{sum}(gauge\_terms))
                    \State \textit{undesired\_terms} \(\gets\) \textbf{Evaluate}(\textbf{sum}(undesired\_terms))
    
                    \State \textit{gauge\_terms} \(\gets\) \textbf{factor\_in}(\textit{gauge\_terms}, \textit{consts})
                    \State \textit{undesired\_terms} \(\gets\) \textbf{factor\_in}(\textit{undesired\_terms}, \textit{consts})
                \Else
                    \State \textit{gauge\_terms} \(\gets\) copy of \textit{pre\_gauge\_terms}
                    \State \textit{undesired\_terms}  \(\gets\) copy of \textit{pre\_undesired\_terms}
                \EndIf
                
                \State \textit{lorentz\_proper\_exp} \(\gets\) \textbf{Evaluate}(\textit{lorentz\_proper\_exp}, \textit{to\_join\_gamma} = False)
                \State \textit{lorentz\_proper\_exp} \(\gets\) \textbf{factor\_in}(\textit{lorentz\_proper\_exp}, \textit{consts})

                \State \textbf{append} ``coefficient of \textit{desired\_terms} in terms of \textit{consts} \(= c \cdot i\)'' to \(S\)

                \For{coefficient \(k\), in terms of \textit{consts}, of each term in \textit{undesired\_terms}}
                    \State \textbf{append} ``\(k = 0\)'' to \(S\)
                \EndFor

                \State susy\_dict[\(A\)] \(\gets\) \{``desired\_terms'': \textit{desired\_terms}, ``gauge\_terms'': \textit{pre\_gauge\_terms}, ``undesired\_terms'': \textit{pre\_undesired\_terms}, ``lorentz\_proper\_terms'': \textit{lorentz\_proper\_exp}\}

            \EndFor

            \State \textit{sol} \(\gets\) \textbf{DistillConstrs}(\(S\), \textit{consts})
            \State \Return \textit{sol}, \textit{susy\_dict}

        \EndFunction
    \end{algorithmic}
\end{algorithm}

As an example, we apply ``SUSYSolve'' to the on-shell 4D vector multiplet, whose supersymmetry rules, with variable coefficients, are
\begin{subequations}
    \begin{align}
        {\rm D}_{\alpha} A_{a} &= u (\gamma_{a})_{\alpha}{}^{\beta} \lambda_{\beta} \\
        {\rm D}_{\alpha} \lambda_{\beta} &=  v (\gamma^{a b})_{\alpha \beta} \partial_{a} A_{b}.
    \end{align}
\end{subequations}
The code below finds the constraints that can be obtained from closure. The basis of the Clifford algebra is
\begin{equation}\label{4d_basis}
    \Gamma^A = \left\{C_{\alpha \beta}, \left(\gamma^{a}\right)_{\alpha \beta}, \left(\gamma^{a b}\right)_{\alpha \beta}, \left(\gamma_{*}\right)_{\alpha \beta}, \left(\gamma_{*}\gamma^{a}\right)_{\alpha \beta}\right\}
\end{equation}
\begin{minted}{python}
>>> bosons = [Ex('A_{a}')]
>>> fermions = [Ex(r'(\lambda)_{\gamma}')]
>>> gauge_transs = [Ex(r'0')]
>>> susy = r'''D_{\alpha}(A_{a}) -> u (\Gamma_{a})_{\alpha}^{\beta} (\lambda)_{\beta}, D_{\alpha}((\lambda)_{\beta}) -> v (\Gamma^{a b})_{\alpha \beta} \partial_{a}(A_{b})'''
>>> consts = ['u', 'v']
>>> indices = [r'_{\alpha}', r'_{\beta}']
>>> susy_solve(bosons, fermions, gauge_transs, susy, basis, consts, indices, comm_coef=2)
\end{minted}
The sole resulting constraint is
\begin{equation}\label{on_shell_vect_closure}
    u v = -i.
\end{equation}
Fixing \(u = 1\) gives \(v = -i\), agreeing with (22) of \cite{Gates2009} and (2.2) of \cite{Gates2020a}.\footnote{The results in \cite{Gates2009, Gates2020a} are written in terms of 1-index gamma matrices, but it is a trivial exercise to convert those to 2-index gamma matrices where appropriate, and the outcomes match precisely the results of ``SUSYSolve.''} Also, ``SUSYSolve'' gives the closure structure as
\begin{subequations}
    \begin{align}
        \{{\rm D}_{\alpha}, {\rm D}_{\beta}\} A_{a} &= -2u v \left(\gamma^{b}\right)_{\alpha \beta} \partial_{b} A_{a} \textcolor{blue}{~+ 2u v \left(\gamma^{b}\right)_{\alpha \beta} \partial_{a} A_{b}} \numberthis \\
        \begin{split}
            \{{\rm D}_{\alpha}, {\rm D}_{\beta}\} \lambda_{\beta} &= -2u v \left(\gamma^{a}\right)_{\alpha \beta} \partial_{a} \lambda_{\gamma} \textcolor{DarkGreen}{~+ \frac12 u v \left(\gamma^{a}\right)_{\alpha \beta} \left(\gamma_{a} \gamma^{b}\right)_{\gamma}{}^{\eta} \partial_{b} \lambda_{\eta}} \\
            &\quad\textcolor{DarkGreen}{+  \frac14 u v \left(\gamma^{a b}\right)_{\alpha \beta} \left(\gamma_{a b} \gamma^{c}\right)_{\gamma}{}^{\eta} \partial_{c} \lambda_{\eta}},
        \end{split}
    \end{align}
\end{subequations}
where the terms in black are of the form of the desired translation, the term in blue is a gauge term, and the terms in green are equation-of-motion terms.\footnote{Formally, they are classified as ``Lorentz'' proper terms, but they are indeed equation-of-motion terms, as they vanish under the Dirac equation.} Inputting \eqref{on_shell_vect_closure} shows that this matches the closure structure in (3.9)-(3.11) of \cite{Gates2020a}.

We now note one limitation of ``SUSYSolve.'' In particular, not all off-shell equation-of-motion terms can be identified and isolated by the gamma-splitting procedure of ``FilterTerms,'' as they may be multiples of the equation of motion by tensor-spinors that contract indices with the equation of motion. The result is that equation-of-motion terms are misclassified as undesired terms,\footnote{The problem is actually a bit subtler: terms which might be decomposed into undesired terms and equation-of-motion terms are simply classified as undesired.} imposing additional constraints on the coefficients of the multiplet that make the system of constraints inconsistent. It can be verified that this occurs when applying ``SUSYSolve'' to the on-shell 4D matter-gravitino and supergravity multiplets in \cite{Gates2020a}, and it is for this reason that we present an alternative solver using Feynman propagators in \S\ref{multiplet_solver2} which does not suffer of this issue and is therefore of wider applicability. Note, however, that this limitation only affects applicability, not accuracy: whenever the outputted constraints are consistent (together with those from other physical requirements, e.g., SUSY-invariance of the action), they are valid. Further, this issue is by definition absent in off-shell multiplets. The advantages of ``SUSYSolve'' over the solver in \S\ref{multiplet_solver2} are 1) that the hash table outputted by ``SUSYSolve'' enables quick computation of the non-closure geometry of a multiplet, and 2) that by avoiding Feynman propagators, ``SUSYSolve'' uses fewer indices in internal expressions, reducing canonicalization time.

\subsection{Procedure for Computing Feynman Propagators}\label{compute_props}
The multiplet solver of \S\ref{multiplet_solver2} will require as input (momentum-space) Feynman propagators for half-integer-spin particles. Hence, in this section, we present functions used to calculate these propagators for arbitrary-spin fields. Since the construction of Feynman propagators for arbitrary spin in 4D is well-documented, we provide algorithms for handling the 4D case to improve ease of use of our suite of tools. As the computation of boson and fermion propagators is rather intertwined, we include the integer-spin case for completeness.\footnote{Note a few conditions on the bosonic field here. Like in the fermionic case (as mentioned in \S\ref{multiplet_solver1}), the boson is assumed to be a symmetric tensor, which has rank equal to its spin. In addition, letting \(A^{a_1 \cdots a_s}\) be a boson of spin \(s\), it is assumed that \(k_{a_1} A^{a_1 \cdots a_s} = 0\) and \(A^{a_1 a_1 a_3 \cdots a_s} = 0\).} We follow \cite{Huang2005, Miyamoto2011} for the relevant theory. At the end of the section, we will introduce the more limited function we use for our 11D computations.

Before calculating propagators, one must first consider projection operators. These can be constructed recursively, starting with the (momentum-space) definition of the spin-1 projection operator as
\begin{equation}\label{spin_1_proj}
    P^{a b} = \eta^{a b} - k^{a} k^{b} \hat{\square}^{-1},
\end{equation}
where
\begin{equation}
    \hat{\square}^{-1} = -\mathcal{F}(\hat{\square})^{-1} = \frac{1}{k^2}
\end{equation}
is (negative) the inverse of the Fourier transform of the d'Alembert operator.\footnote{We feel justified in adopting this notation for the inverse momentum-space d'Alembert operator, in spite of its simple formula, because of its role as the ``spin-0 Feynman propagator,'' or more accurately, the (Euclidean) Green's function of the Klein-Gordon equation for a (massless) scalar field. The inverse Fourier transform of \(\hat{\square}^{-1}\) is often denoted by the symbol chosen for arbitrary integer-spin propagators, \(\Delta_{F}\) in \cite{Huang2005} and \(G\) in \S4.1.3 of \cite{Freedman2013}, but without Lorentz indices, corresponding to spin 0.} If \(s\) is an integer, then let \(A = \{a_1, \dotsc, a_s\}, B = \{b_1, \dotsc, b_s\}\). Given the projection operators for integer spin \(\leq s - 1\), the projection operator for spin \(s\) is computed via ((8-9) in \cite{Huang2005} and (6.1.18-6.1.19) in \cite{Miyamoto2011})\footnote{This projection operator construction appears in other incarnations in older work on higher-spin particles. See the references in \cite{Huang2005}.}
\begin{equation}\label{int_spin_proj}
    P^{A B} = \left(\frac{1}{s!}\right)^2 \sum_{(\sigma, \sigma') \in \mathfrak{S}(A) \times \mathfrak{S}(B)} \sum_{r = 0}^{\lfloor s/2 \rfloor} \frac{(-1)^r s!}{2^r r! (s - 2r)! q_{r}} \prod_{i=1}^{r} P^{\sigma_{2i-1} \sigma_{2i}} P^{\sigma'_{2i-1} \sigma'_{2i}} \prod_{i=2r+1}^{s} P^{\sigma_{i} \sigma'_{i}},
\end{equation}
where \(\mathfrak{S}(A), \mathfrak{S}(B)\) are the sets of permutations of \(A, B\), respectively, and
\begin{equation}
    q_{r} = \begin{cases} 1, & \textnormal{if } r = 0 \\ \prod_{i=1}^{r} (2s - 2i + 1), & \textnormal{if } r \geq 1. \end{cases}
\end{equation}
The projection operator for half-integer spin \(s\) can then be computed from that for spin \(\ceil{s}\) via ((23b) in \cite{Huang2005} and (6.1.59-6.1.60) in \cite{Miyamoto2011})\footnote{Notice that the projection operator has spinor indices, as it must. A particle of half-integer spin \(s\) is a tensor-spinor of (Lorentz) rank \(\floor{s}\) with one spinor index. Note that when discussing fermions in this section, for brevity, we break slightly from our notation elsewhere and use \(s\), not \(s + 1/2\), to denote a half-integral spin.}
\begin{equation}\label{half_int_spin_proj}
    \left(Q^{a_1 \cdots a_{\floor{s}} b_1 \cdots b_{\floor{s}}}\right)_{\alpha \beta} = \frac{\ceil{s}}{2s + 2} \left(\gamma_{a} \gamma_{b}\right)_{\alpha \beta} P^{a a_1 \cdots a_{\floor{s}} b b_1 \cdots b_{\floor{s}}}.
\end{equation}
The algorithm ``SymProjDecomp'' (see Algorithm \ref{alg:symprojdecomp}) computes the projection operator for arbitrary spin \(s\) via the foregoing recursive method. It takes as input the spin (and an optional parameter dictating whether to perform a gamma-matrix substitution in the calculation of a half-integer-spin projection operator, which is not relevant for most dimensions of interest) and returns a substitution rule taking \(P^{a_1 \cdots a_s b_1 \cdots b_s}\) or \(\left(Q^{a_1 \cdots a_{\floor{s}} b_1 \cdots b_{\floor{s}}}\right)_{\alpha \beta}\) to its decomposition in terms of spin-1 projection operators \(P^{a b}\). Note that we contract indices via the easily verified identities
\begin{subequations}
    \begin{align}
        P_{a}{}^{a} = 3 \\
        P_{a b} P^{b c} = P_{a}{}^{c}.
    \end{align}
\end{subequations}
Note also that the symmetrizations are weight-1, yielding the needed \(1/s!\) factor.

\begin{algorithm}[ht!]
    \caption{Compute the Projection Operator for a Field of Spin \(s\)}\label{alg:symprojdecomp}
    \begin{algorithmic}[1]
        \Require \(s\) is an integer or half-integer
        \Function{SymProjDecomp}{\(s\), \textit{to\_perform\_subs} = False}

            \If{\(s\) is an integer}
                \State \(P_0 \gets 0\)
                \For{\(r \in 0..\floor{\frac{s}{2}}\)} \Comment{Compute the inner sum in \eqref{int_spin_proj}}
                    \State \(q \gets \prod_{i \in 1..r} (2s - 2i + 1)\) \textbf{if} \(r \geq 1\) \textbf{otherwise} 1 \Comment{Compute \(q_{r}\)}
                    \State \(P_0 \gets P_0 + \frac{(-1)^r s!}{2^r r! (s - 2r)! q} \prod_{i \in 1..r} P^{\sigma_{2i-1} \sigma_{2i}} P^{\sigma'_{2i-1} \sigma'_{2i}} \prod_{i \in (2r+1)..s} P^{\sigma_{i} \sigma'_{i}}\) \Comment{Add a summand}
                \EndFor

                \If{\(s > 1\)}
                    \State \textbf{symmetrise} \(P_0\) in \(a_1...a_s\) \Comment{Sum over \(\mathfrak{S}(A)\)}
                    \State \textbf{symmetrise} \(P_0\) in \(b_1...b_s\) \Comment{Sum over \(\mathfrak{S}(B)\)}
                \EndIf
                \State \(P_0 \gets\) \textbf{Evaluate}(\(P_0\), \textit{to\_perform\_subs} = False)
                \State \Return ``\(P^{a_1 \cdots a_s b_1 \cdots b_s} \rightarrow P_0\)''

            \ElsIf{\(s\) is a half-integer}
                \State \(Q_0 \gets \frac{\ceil{s}}{2s+2} (\gamma_{a} \gamma_{b})_{\alpha \beta} P^{a a_1 \cdots a_{\floor{s}} b b_1 \cdots b_{\floor{s}}}\) \Comment{Write \eqref{half_int_spin_proj}}
                \State \(Q_0 \gets\) \textbf{substitute}(\(Q_0\), \textbf{SymProjDecomp}(\(\ceil{s}\))) \Comment{Decompose \(P^{a a_1 \cdots a_{\floor{s}} b b_1 \cdots b_{\floor{s}}}\)}
                \State \(Q_0 \gets\) \textbf{Evaluate}(\(Q_0\), \textit{to\_perform\_subs} = False)
                \State \(Q_0 \gets\) \textbf{substitute}(\(Q_0\), \(P_{a}{}^{a} \rightarrow 3\))
                \State \(Q_0 \gets\) \textbf{substitute}(\(Q_0\), \(P_{a b} P^{b c} \rightarrow P_{a}{}^{c}\))
                \State \(Q_0 \gets\) \textbf{Evaluate}(\(Q_0\), \textit{to\_perform\_subs} = \textit{to\_perform\_subs})

                \State \Return ``\((Q^{a_1 \cdots a_{\floor{s}} b_1 \cdots b_{\floor{s}}})_{\alpha \beta} \rightarrow Q_0\)''

            \EndIf
        \EndFunction
    \end{algorithmic}
\end{algorithm}

As an example, the single line of code below computes the spin-3 projection operator,
\begin{minted}{python}
>>> sym_proj_decomp(3)
\end{minted}
verifying the decomposition ((7a) of \cite{Huang2005} and (6.1.28) of \cite{Miyamoto2011})
\begin{equation}
    \begin{split}
        P^{a_1 a_2 a_3 b_1 b_2 b_3} &\to \frac16 \bigg\{P^{a_1 b_1} P^{a_2 b_2} P^{a_3 b_3} + P^{a_1 b_1} P^{a_2 b_3} P^{a_3 b_2} + P^{a_1 b_2} P^{a_2 b_1} P^{a_3 b_3} \\
        &\quad~+ P^{a_1 b_2} P^{a_2 b_3} P^{a_3 b_1} + P^{a_1 b_3} P^{a_2 b_2} P^{a_3 b_1} + P^{a_1 b_3} P^{a_2 b_1} P^{a_3 b_2}\bigg\} \\
        &\quad~- \frac{1}{15} \bigg\{P^{a_1 a_2} P^{b_1 b_2} P^{a_3 b_3} + P^{a_1 a_3} P^{b_1 b_2} P^{a_2 b_3} + P^{a_2 a_3} P^{b_1 b_2} P^{a_1 b_3} \\
        &\quad~+ P^{a_1 a_2} P^{b_1 b_3} P^{a_3 b_2} + P^{a_1 a_3} P^{b_1 b_3} P^{a_2 b_2} + P^{a_2 a_3} P^{b_1 b_3} P^{a_1 b_2} \\
        &\quad~+ P^{a_1 a_2} P^{b_2 b_3} P^{a_3 b_1} + P^{a_1 a_3} P^{b_2 b_3} P^{a_2 b_1} + P^{a_2 a_3} P^{b_2 b_3} P^{a_1 b_1}\bigg\},
    \end{split}
\end{equation}
For a fermionic example, the code below computes the spin-\(3/2\) projection operator,
\begin{minted}{python}
>>> sym_proj_decomp(3/2)
\end{minted}
verifying the decomposition (equivalent to (21b) in \cite{Huang2005})
\begin{equation}
    \left(Q^{a_1 b_1}\right)_{\alpha \beta} \to \frac23 C_{\alpha \beta} P^{a_1 b_1} - \frac13 \left(\gamma_{a b}\right)_{\alpha \beta} P^{a a_1} P^{b b_1}.
\end{equation}

Now, for a field of integral spin \(s\), the propagator is\footnote{This is essentially (6.2.2) in \cite{Miyamoto2011}, but we omit the coefficient of \(-i\) to more closely reflect the conventions in \cite{Freedman2013}. There can also be an additional function of momentum added to the expression for the propagator; see (6.2.1) of \cite{Miyamoto2011} or (73) of \cite{Huang2005}. We follow \cite{Miyamoto2011} and assume this added function is zero. One might argue that we are abusing notation by calling the factor multiplying the projection operator in \eqref{int_spin_prop} ``\(\hat{\square}^{-1}\)''; one could perhaps attempt to view this rigorously as a limit of masslessness (although this is not quite the massless expression), but this level of rigor is superfluous, and the symbolic algebra procedures are simpler without a separate treatment of this factor. \label{prop_ambig}}
\begin{equation}\label{int_spin_prop}
    D_{F}^{a_1 \cdots a_s b_1 \cdots b_s} = \hat{\square}^{-1} P^{a_1 \cdots a_{s} b_1 \cdots b_{s}},
\end{equation}
while for a field of half-integer spin \(s\), the propagator is\footnote{This corresponds to (6.2.4) of \cite{Miyamoto2011}, which again assumes no added function of momentum.}
\begin{equation}\label{half_int_spin_prop}
    \left(S_{F}^{a_1 \cdots a_{\floor{s}} b_1 \cdots b_{\floor{s}}}\right)_{\alpha \beta} = i \hat{\square}^{-1} \left(\gamma^b\right)_{\alpha \beta} k_{b} Q^{a_1 \cdots a_{\floor{s}} b_1 \cdots b_{\floor{s}}}.
\end{equation}
The algorithm ``SymProp'' (see Algorithm \ref{alg:symprop}) computes the Feynman propagator via the foregoing equations, applying ``SymProjDecomp'' to reduce to spin-\(1\) projection operators and then \eqref{spin_1_proj} to reduce to Kronecker deltas, momentum parameters, and d'Alembert operators. Note the suppressed subprocedure ``SubstituteKleinGordon,'' which cleans up all instances of \(\hat{\square}^{-1} k^2\).

\begin{algorithm}[ht!]
    \caption{Compute the Feynman Propagator for a Field of Spin \(s\)}\label{alg:symprop}
    \begin{algorithmic}[1]
        \Require \(s\) is an integer or half-integer
        \Function{SymProp}{\(s\), \textit{to\_perform\_subs} = False}

            \If{\(s\) is an integer}
                \State \(G \gets \hat{\square}^{-1} P^{a_1 \cdots a_s b_1 \cdots b_s}\) \Comment{Write \eqref{int_spin_prop}}

            \ElsIf{\(s\) is a half-integer}
                \State \(G \gets i \hat{\square}^{-1} (\gamma^b)_{\alpha}{}^{\gamma} k_{1_b} (Q^{a_1 \cdots a_{\floor{s}} b_1 \cdots b_{\floor{s}}})_{\gamma \beta}\) \Comment{Write \eqref{half_int_spin_prop}}
            \EndIf

            \State \(G \gets\) \textbf{substitute}(\(G\), \textbf{SymProjDecomp}(\(s\))) \Comment{Decompose \(P^{a_1 \cdots a_s b_1 \cdots b_s}\) or \((Q^{a_1 \cdots a_{\floor{s}} b_1 \cdots b_{\floor{s}}})_{\alpha \beta}\)}
            \State \(G \gets\) \textbf{substitute}(\(G\), \(P^{a_1 b_1} \rightarrow \delta^{a_1 b_1} - (k_1{}^{a_1} k_1{}^{b_1}) \hat{\square}^{-1}\)) \Comment{Apply \eqref{spin_1_proj}}
            \State \(G \gets\) \textbf{Evaluate}(\(G\), \textit{to\_perform\_subs} = \textit{to\_perform\_subs})
            \State \(G \gets\) \textbf{SubstituteKleinGordon}(\(G\)) \Comment{Replace \(\hat{\square}^{-1} k^2\) with \(1\)}
            \State \(G \gets\) \textbf{collect\_terms}(\(G\))

            \State \Return \(G\)

        \EndFunction
    \end{algorithmic}
\end{algorithm}

As an example, the code below computes the Feynman propagator for a spin-\(2\) particle, e.g., a graviton,\footnote{This seemed to the authors the most interesting demonstration within the ambient domain, supergravity, of our work. We will otherwise not be making much use of bosonic propagators.}
\begin{minted}{python}
>>> sym_prop(2)
\end{minted}
verifying the expression\footnote{Note that the momentum-space inverse d'Alembert operator \(\hat{\square}^{-1}\) is represented in the code output by the text \mintinline{python}{KleinGordon}.} (equivalent to that obtained from Feynman diagrams in \S2 of \cite{Veltman1981})\footnote{Conventions differ on whether this quantity should be called the graviton propagator. This essentially reflects the arbitrary term which may be added: see footnote \ref{prop_ambig}.}
\begin{equation}
    \begin{split}
        D_{F}^{a_1 a_2 b_1 b_2} &= \hat{\square}^{-1} \bigg\{\frac12 (\delta^{a_1 b_1} \delta^{a_2 b_2} + \delta^{a_1 b_2} \delta^{a_2 b_1} - \delta^{a_1 a_2} \delta^{b_1 b_2}) \\
        &\quad\quad\quad -  \frac12 \hat{\square}^{-1} (\delta^{a_1 b_1} k^{a_2} k^{b_2} + \delta^{a_2 b_2} k^{a_1} k^{b_1} + \delta^{a_1 b_2} k^{a_2} k^{b_1} + \delta^{a_2 b_1} k^{a_1} k^{b_2}) \\
        &\quad\quad\quad + \frac23 \left(\frac12 \delta^{a_1 a_2} + \hat{\square}^{-1} k^{a_1} k^{a_2}\right) \left(\frac12 \delta^{b_1 b_2} + \hat{\square}^{-1} k^{b_1} k^{b_2}\right)\bigg\}.
    \end{split}
\end{equation}

We now present one more function for general dimension \(D\), particularly of use for supergravity. The function \mintinline{python}{rarita_schwinger_prop()} returns the massless Rarita-Schwinger Feynman propagator in a given dimension \(D\) using the equation
\begin{equation}\label{rarita_schwinger_prop}
    S_{F_{a b \alpha \beta}} = i \hat{\square}^{-1} \eta_{a b} \left(\gamma^{c}\right)_{\alpha \beta} k_{c} + \frac{i}{D-2} \hat{\square}^{-1} \left(\gamma_{a} \gamma^{c} \gamma_{b}\right)_{\alpha \beta} k_{c},
\end{equation}
which is (5.31) in \cite{Freedman2013} (with a particular choice of gauge terms). We omit pseudo-code here.

\subsection{Multiplet Solver 2: Feynman Propagator Substitution}\label{multiplet_solver2}
We now treat our second algorithm for obtaining constraints on the coefficients of a multiplet from closure of the algebra. The principal difference from the multiplet solver in \S\ref{multiplet_solver1} is the manner of ferreting out off-shell equation-of-motion terms in the closure for an on-shell multiplet. Our general procedure is to move to momentum space and use a Green's function to identify occurrences of the equation of motion. Recall that for a spin-(\(s + 1/2\)) fermion \(\Psi^{U}{}_{\alpha}\) with Feynman propagator \((S_{F}^{U V})_{\alpha \beta}\), the coupling current \(J^{U}{}_{\alpha}\) is defined so that
\begin{equation}\label{coupling_current}
    \mathcal{F}(\Psi^{U}{}_{\alpha}) = - S_{F_{U V \alpha}}{}^{\beta} \mathcal{F}(J^{V}{}_{\beta}),
\end{equation}
with the additional current conservation condition
\begin{equation}\label{current_conserv}
    k_{a} \mathcal{F}(J^{a a_2 \cdots a_s}{}_{\beta}) = 0.
\end{equation}
In position-space, this means that the fermion is (negative) the convolution of the propagator and the coupling current; see the spin-\(3/2\) case in (5.26) of \cite{Freedman2013}. Since the Feynman propagator is fundamentally a Green's function for the relevant equation of motion, it follows that inputting this convolution into the left-hand side of \eqref{sym_fermion_eqom} yields a scalar multiple of the coupling current.\footnote{The convention for the coefficient of the propagator can be chosen so that this scalar coefficient equals \(1\), but it need not be for our purposes.} In momentum space, this means that
\begin{equation}\label{mom_space_eqom}
    \left(\gamma^{b}\right)_{\eta}{}^{\gamma} k_{b} \mathcal{F}(\Psi_{a_1 \cdots a_{s}}{}^{\eta}) - \sum_{i = 1}^{s} \left(\gamma^{b}\right)_{\eta}{}^{\gamma} k_{a_{i}} \mathcal{F}(\Psi_{b U_{i}}{}^{\eta}) = \alpha \mathcal{F}(J_{a_1 \cdots a_s}{}^{\gamma})
\end{equation}
for some \(\alpha \in \mathbb{C}\). Notice that all occurrences of \(\hat{\square}^{-1}\) are furtively contracting away: substituting the right-hand side of \eqref{coupling_current} into the field equation effects contractions between momentum parameters which cancel the instances of \(\hat{\square}^{-1}\). A moment's thought reveals that this uniquely determines the equation of motion: a momentum-space expression in terms of the fermion of interest (and not itself including instances of \(\hat{\square}^{-1}\)) is a (not necessarily scalar) multiple of the left-hand side of \eqref{mom_space_eqom} if and only if the result of substituting the right-hand side of \eqref{coupling_current} for the fermion lacks any occurrence of \(\hat{\square}^{-1}\). When considering the anticommutator of supercovariant derivatives applied to the fermion, such multiples of the equation of motion are precisely the equation-of-motion terms.

We adapt these ideas to an analogue of ``FilterTerms.'' The algorithm ``FilterTermsMomentumSpace'' has the same inputs as ``FilterTerms'' but for an additional substitution rule \(\tilde{G}_{A}\) of the form of \eqref{coupling_current} above, taking the Fourier transform of the field \(A\) to its expression as the product of a (momentum-space) Feynman propagator and  the Fourier transform of a coupling current \(J_{A}\). Two assumptions are made as to the input. First, the inputted expression \(\mathcal{E}\) is assumed to represent the Fourier transform of the Fierz-transformed closure, and is assumed to be flat (i.e., a mere sum of products). Second, the propagator is assumed to have a distinguished term of the form \(\hat{\square}^{-1} \eta_{a_1 b_1} \cdots \eta_{a_s b_s} (\gamma^{a'})_{\gamma}{}^{\eta} k_{a'}\), which when contracted with \(\mathcal{F}(J_{A}^{b_1 \cdots b_s}{}_{\eta})\) gives rise to a term of the form \(\hat{\square}^{-1} (\gamma^{a'})_{\gamma}{}^{\eta} k_{a'} \mathcal{F}(J_{A_{a_1 \cdots a_s \eta}})\) in \(\tilde{G}_{A}\), with a single 1-index gamma matrix contracted with a momentum parameter, and with all indices \(a_1, \dotsc, a_s, \gamma\) free. This is a reasonable assumption: all fermionic propagators of interest can be made to meet this requirement via the manipulation of gauge and equation-of-motion terms, and propagators are typically written with such a term, as seen in the Rarita-Schwinger case in \eqref{rarita_schwinger_prop}. The significance of this singular term is to serve as a hook when reconstructing the left-hand side of \eqref{coupling_current} from the right-hand side. We set the stage for this reconstruction by drawing the right-hand side\footnote{Note that a selection of the procedure of ``Evaluate'' is applied without the renaming of dummy indices in order to preserve the Lorentz-index structure for ease of filtration (this is suppressed in the pseudo-code of Algorithm \ref{alg:ftms1}), and then the dummy indices are renamed with an added prime, e.g., \(a\) is replaced with \(a'\).} \(\mathcal{E}'\) of \eqref{coupling_current} via \(\tilde{G}_{A}\) and recording the coefficient \(\ell\) of the aforementioned term, together with the set \(\mathcal{R}\) of remaining terms.\footnote{It is worth making a technical note about isolating the coefficient of a term in the loop in Algorithm \ref{alg:ftms1}. Cadabra does not treat the imaginary unit \(i\) like a real scalar, and ordinary means of drawing the coefficient of a term in Cadabra only return the real multiplier, e.g., \(5\) even if the constant coefficient is \(5i\). There is a bit of manipulation and substitution involved in drawing the full coefficient, including \(i\), and in inverting the result to get \(\ell\), but this is a purely technical concern and is suppressed here.\label{Im_foot}} This first part of the procedure is shown in pseudo-code as Algorithm \ref{alg:ftms1}.

\begin{algorithm}[ht!]
    \caption{Filter Closure/Non-Closure Terms Involving Coupling Currents in Momentum Space, Part 1}\label{alg:ftms1}
    \begin{algorithmic}[1]
        \Function{FilterTermsMomentumSpace}{\(\mathcal{E}\), \(A\), \(\tilde{G}_A\), \textit{Inds}[1..2]}

            \State \textit{desired\_terms} \(\gets\) [ ]
            \State \textit{ignored\_terms} \(\gets\)  [ ]
            \State \textit{gauge\_terms} \(\gets\)  [ ]
            \State \textit{eqom\_terms} \(\gets\)  [ ]
            \State \textit{undesired\_terms} \(\gets\)  [ ]
            \State \textit{normalized\_schouten\_terms} \(\gets\) [ ]
            \State \(\ell \gets\) 1

            \State \(\mathcal{E}' \gets \mathcal{F}_1((A_G)_{\bullet})\)
 
            \State \(\mathcal{E}' \gets\) \textbf{substitute}(\(\mathcal{E}'\), \(\tilde{G}_A\))
            \State \textbf{simplify} \(\mathcal{E}'\) without renaming indices
            \State \textbf{prime} dummy indices in \(\mathcal{E}'\)

            \State \(\mathcal{R} \gets \) [ ]
            
            \For{\textit{term} in \(\mathcal{E}'\)}
                \If{\textit{term} has form \(\hat{\square}^{-1} (\gamma^{a'})_{\bullet}{}^{\beta'} k_{1_{a'}} \mathcal{F}_1((J_{A_G})_{\beta'})\)}
                    \State \(\ell \gets\) inverse of the coefficient of \textit{term} \Comment{including imaginary factors in the coefficient of \textit{term} (see footnote \ref{Im_foot})}
                \Else
                    \State \textbf{append} \textit{term} to \(\mathcal{R}\)
                \EndIf
            \EndFor
            \State continued...
        \algstore{bkbreak}
    \end{algorithmic}
\end{algorithm}

We now loop through an evolving array of terms, initially those in \(\mathcal{E}\), and proceed according to six possible cases, where \(t\) is the term at hand. 1) If \(t\) has the form of the distinguished term \(\hat{\square}^{-1} (\gamma^{a'})_{\gamma}{}^{\eta} k_{a'} \mathcal{F}(J_{A_{a_1 \cdots a_s \eta}})\) multiplied by \((\gamma^{a''})_{\textit{Inds}} k_{a''}\), then letting \(\ell'\) be the scalar coefficient of \(t\), we have
\begin{equation}
    t + \ell \ell' \left(\gamma^{a''}\right)_{\textit{Inds}} k_{a''} \mathcal{R} = \ell \ell' \left(\gamma^{a''}\right)_{\textit{Inds}} k_{a''} \mathcal{E}' = \ell \ell' \left(\gamma^{a''}\right)_{\textit{Inds}} k_{a''} \mathcal{F}(A_{G \gamma}),
\end{equation}
where the last expression is of the form of the Fourier transform of the desired translation in \eqref{closure_condition} applied to \(A\). Hence, the algorithm replaces \(t\) with \(\ell \ell' (\gamma^{a''})_{\textit{Inds}} k_{a''} \mathcal{F}(A_{G \gamma}) - \ell \ell' (\gamma^{a''})_{\alpha \beta} k_{a''} \mathcal{R}\), classifies the first term as a desired term, and throws the remaining terms back into the loop for further evaluation. None of the terms in the latter remainder can be of the same form as case (1), so no infinite loop results. 2) If in \(t\), some momentum parameters share indices with the coupling current, then by \eqref{current_conserv}, \(t = 0\), so \(t\) is classified as an ``ignored'' term. 3) If in \(t\), at least one momentum parameter has a free index, then \(t\) is classified as a gauge term. 4) If in \(t\), \(\hat{\square}^{-1}\) is absent, then by the discussion at the beginning of this section, \(t\) is classified as an equation-of-motion term. 5) If \(t\) is of the form \(\hat{\square}^{-1} (\gamma^{E})_{\textit{Inds}} (\gamma^{F})_{\gamma}{}^{\eta} k_{a_1} \cdots k_{a_n} \mathcal{F}(J_{A_{G \eta}})\), where the second gamma matrix has \(|F| = D\) indices and where at least one of the \(a_{i}\) is not in \(F\), then Schouten identities must be handled. 6) If \(t\) matches none of the criteria (1)-(5), then it is classified as an undesired term (i.e., a term which is eventually set to zero to get a constraint). This second part of the procedure is shown in pseudo-code as Algorithm \ref{alg:ftms2}.

The details of case (5) require further explanation. The crux of the matter is that a term might not fit the criteria for any of (1)-(4) but still not be a simple undesired term  because it is actually related to terms fitting (1)-(4) by a Schouten identity. For example, in 4D, suppose that\footnote{This particular term comes from an intermediate step in the solution for the coefficients of the 4D supergravity multiplet.}
\begin{equation}\label{pre_asym_t_ex}
    t = -  \frac32 u v \hat{\square}^{-1} \left(\gamma^{c d}\right)_{\alpha \beta} \left(\gamma_{a b c}{}^{e}\right)_{\gamma}{}^{\eta} k_{d} k_{e} \mathcal{F}\left(J_{A}{}^{b}{}_{\eta}\right).
\end{equation}
This is ostensibly an undesired term not permitted in \eqref{closure_condition}, and assuming no like terms, one might jump to the conclusion that \(u v = 0\). However, notice that the second gamma matrix involves \(D = 4\) vector indices. Since a Schouten identity arises from the fact that a totally antisymmetric tensor of rank \(D + 1\) must be zero, it seems reasonable to draw another index in the expression, \(d\), and antisymmetrize to find an identity. Antisymmetrizing (without the coefficient) in the indices \(a, b, c, d, e\) yields\footnote{Note that \eqref{asym_tp_ex} is a weight-1 antisymmetrization of \eqref{pre_asym_t_ex} (without the coefficient of the latter). That is, a normalization factor \(1/120\) is introduced. This reflects the normalization of antisymmetrization in Cadabra.}
\begin{equation}\label{asym_tp_ex}
    \begin{split}
        t' &= \frac25 u v \hat{\square}^{-1} \left(\gamma^{c d}\right)_{\alpha \beta} \left(\gamma_{a b c}{}^{e}\right)_{\gamma}{}^{\eta} k_{d} k_{e} \mathcal{F}\left(J_{A}{}^{b}{}_{\eta}\right) -  \frac15 u v \left(\gamma^{c d}\right)_{\alpha \beta} \left(\gamma_{a b c d}\right)_{\gamma}{}^{\eta} \mathcal{F}\left(J_{A}{}^{b}{}_{\eta}\right) \\
        &\quad+  \frac15 u v \hat{\square}^{-1} \left(\gamma^{c d}\right)_{\alpha \beta} \left(\gamma_{a c d}{}^{e}\right)_{\gamma}{}^{\eta} k_{b} k_{e} \mathcal{F}\left(J_{A}{}^{b}{}_{\eta}\right) -  \frac15 u v \hat{\square}^{-1} \left(\gamma^{c d}\right)_{\alpha \beta} \left(\gamma_{b c d}{}^{e}\right)_{\gamma}{}^{\eta} k_{a} k_{e} \mathcal{F}\left(J_{A}{}^{b}{}_{\eta}\right).
    \end{split}
\end{equation}
The first term, which we denote \(t'_0\), is a scalar multiple of \(t\), while the other terms are respectively an equation-of-motion term (no \(\hat{\square}^{-1}\)), an ``ignored term'' (\(b\) is shared by the current and a momentum parameter), and a gauge term (the free index \(a\) is on a momentum parameter). Remembering that \(t' = 0\), we can solve for \(t\) to find
\begin{equation}
    \begin{split}
        t &= -  \frac34 u v \left(\gamma^{c d}\right)_{\alpha \beta} \left(\gamma_{a b c d}\right)_{\gamma}{}^{\eta} \mathcal{F}\left(J_{A}{}^{b}{}_{\eta}\right) +  \frac34 u v \hat{\square}^{-1} \left(\gamma^{c d}\right)_{\alpha \beta} \left(\gamma_{a c d}{}^{e}\right)_{\gamma}{}^{\eta} k_{b} k_{e} \mathcal{F}\left(J_{A}{}^{b}{}_{\eta}\right) \\
        &\quad-  \frac34 u v \hat{\square}^{-1} \left(\gamma^{c d}\right)_{\alpha \beta} \left(\gamma_{b c d}{}^{e}\right)_{\gamma}{}^{\eta} k_{a} k_{e} \mathcal{F}\left(J_{A}{}^{b}{}_{\eta}\right),
    \end{split}
\end{equation}
i.e., \(t\) is equivalent to a set of non-undesired terms. In order to incorporate the insights from the foregoing example, in case (5), we choose an index \(a_{i}\) and antisymmetrize \(t\) in the \(D + 1\) indices \(F \cup \{a_{i}\}\) to get \(t'\); the first term is our \(t'_0\). Let \(r\) denote the remaining terms. Then we can solve for \(t\) in terms of \(r\) can throw the resulting terms back into the loop for further evaluation. An infinite loop can arise if terms of \(r\) are also of the form (5) (reflecting that \(t\) is really an undesired term). To prevent this, we store normalized copies of the terms in \(r\), and if any term of \(r\) has appeared previously, we abort and classify \(t\) as undesired.

\begin{algorithm}[ht!]
    \caption{Filter Closure/Non-Closure Terms Involving Coupling Currents in Momentum Space, Part 2}\label{alg:ftms2}
    \begin{algorithmic}[1]
            \algrestore{bkbreak}
            \State \(T \gets\) terms of \(\mathcal{E}\)

            \State \(n = 0\)
            \While{\(n < |T|\)}
                \State \(t \gets T[n]\)

                \If{\(t\) has form \(\hat{\square}^{-1} (\gamma^{a'})_{\textit{Inds}} (\gamma^{b'})_{\bullet}{}^{\beta'} k_{1_{a'}} k_{1_{b'}} \mathcal{F}_1((J_{A_G})_{\beta'})\)}
                    \State \textbf{append} \(\ell \cdot (\)coefficient of \(t) \cdot (\gamma^{a'})_{\textit{Inds}} k_{1_{a'}} \mathcal{F}_1((A_G)_{\bullet})\) to \textit{desired\_terms} \Comment{including ``\(i\)''s in \(t\)}
                    \State \textbf{append} terms of \(-\ell \cdot (\)coefficient of \(t) \cdot (\gamma^{a'})_{\textit{Inds}} k_{1_{a'}} \mathcal{R}\) to \(T\)

                \ElsIf{the indices on the ``\(k_1\)''s have a nonempty intersection with the indices on \(J_A\)}
                    \State \textbf{append} \(t\) to \textit{ignored\_terms}

                \ElsIf{at least one ``\(k_1\)'' factor of \(t\) has a free index}
                    \State \textbf{append} \(t\) to \textit{gauge\_terms}

                \ElsIf{\(\hat{\square}^{-1}\) is not a factor of \(t\)}
                    \State \textbf{append} \(t\) to \textit{eqom\_terms}

                \ElsIf{\(t\) has form \(\hat{\square}^{-1} (\gamma^{E})_{\textit{Inds}} (\gamma^F)_{\bullet}{}^{\beta'} {k_{1_{a'_1}}} \cdots  {k_{1_{a'_n}}} \mathcal{F}_1((J_{A_G})_{\beta'})\) \& \(|F| = D\) \& \(\{a'_1, \dotsc, a'_n\} \not\subset F\)}
                    \State \textbf{append} normalized \(t\) to \textit{normalized\_schouten\_terms}
                    \State \(a_i \gets \) an element of \(\{a'_1, \dotsc, a'_n\} \setminus F\)
                    \State \(t' \gets t\) antisymmetrized in the indices \(F \cup \{a_i\}\)
                    \State \textbf{simplify} \(t'\) without sorting terms
                    \State \(t'_0 \gets\) first term of \(t'\)
                    \State \(r \gets t' - t'_0\)
                    \State \(r \gets\) \textbf{collect\_terms}(\(r\))

                    \If{terms of \(r~\cap\) \textit{normalized\_schouten\_terms} \(\ne \emptyset\)}
                        \State \textbf{append} \(t\) to \textit{undesired\_terms}
                    \Else:
                        \State \textbf{append} terms of \(-[(\textnormal{coefficient of } t)/(\textnormal{coefficient of } t'_0)] \cdot r\) to \(T\)
                    \EndIf

                \Else
                    \State \textbf{append} \(t\) to \textit{undesired\_terms}

                \EndIf

                \State \(n \gets n + 1\)
            \EndWhile

            \State \Return \{``desired\_terms'': desired\_terms, ``ignored\_terms'': ignored\_terms, ``gauge\_terms'': gauge\_terms, ``eqom\_terms'': eqom\_terms, ``undesired\_terms'': undesired\_terms\}

        \EndFunction
    \end{algorithmic}
\end{algorithm}

Note that this handling of Schouten identities is necessarily very specific. Such specificity reduces computational complexity. Indeed, Cadabra2's \cite{Peeters2018b} function \mintinline{python}{decompose_product()}, which attempts to broadly treat all Schouten identities (in the sense of proving expressions zero under Schouten identities) through decomposition of tensors into irreducible Young tableau representations, is of impractical computational complexity for higher dimensions like \(D = 11\). The difficulty is primarily that Schouten identities represent dimension-dependent multi-term symmetries. Handling such multi-term identities in symbolic algebra is a famously difficult task (see the references in \cite{Hongbo2017}), inspiring brute-force solutions like Invar's database of \(6\cdot10^5\) Riemann tensor identities \cite{MartinGarcia2008a}. Time will tell whether such a massive database will be necessary for the efficient computational treatment of all Schouten identities that might appear in the filtration of closure terms after the substitution of coupling currents. However, since Schouten identities require the complete antisymmetry of \(D + 1\) indices, virtually all of these indices must be on the gamma matrices. (The momentum parameters commute and the coupling current is symmetric, so they could take up at most two indices together.) Hence, if an ostensibly undesired term can be transformed into non-undesired terms via a Schouten identity, it must be rather close in form to case (5). Indeed, the Schouten identities treated in case (5) empirically enable the solution of all of the on-shell 4D, \(\mathcal{N} = 1\) multiplets in \cite{Gates2020a}. It seems reasonable, then, to conjecture that case (5) represents the only possible situation in which a Schouten identity could impact term classification.

We can now treat the (second) multiplet solver algorithm ``SUSYSolvePropagator'' (see Algorithm \ref{alg:susysp}). The algorithm takes the same inputs as ``SUSYSolve,'' except that instead of the array ``gauge\_transs'' of gauge transformations for the fermions, ``SUSYSolvePropagator'' takes as input the array ``fermion\_propagators'' of Feynman propagators for the fermions. The algorithm applies ``SUSYSolve'' to the bosons alone to obtain the constraints from those fields (since we assume that all equation-of-motion terms are off-shell). For each fermion, the algorithm draws the Fourier transform\footnote{The careful reader might have noticed that only the fermions are inputted into ``Fourier'' in Algorithm \ref{alg:susysp}, rather than the combined array of bosons and fermions. Recall from \S\ref{symbolic_fourier} that the input of an array of fields into `Fourier'' is supposed to represent all fields appearing in the inputted expression. In this case, simply counting spinor indices reveals that \(\mathcal{E}\) is fermionic, and since each term includes only one instance of a field, it follows that each field appearing in \(\mathcal{E}\) must be a fermion.} of the Fierz expansion of the anticommutator of supercovariant derivatives applied to the field and then substituting the Feynman propagator and coupling current for each fermionic field as in \eqref{coupling_current}.\footnote{The coupling currents corresponding to different fields are distinguished in the code via a hexadecimal based on the field's name.} Note that we apply a subprocedure ``GenPropSub,'' whose technical details we omit, to construct the substitution rule \(\tilde{G}_{A}\) corresponding to a fermion \(A\) from its Feynman propagator. The procedure filters terms via ``FilterTermsMomentumSpace''; note that in contrast to the procedure in \S\ref{multiplet_solver1}, ``FilterTermsMomentumSpace'' definitively determines equation-of-motion terms, so no further evaluation is necessary. Constraints are then obtained as in ``SUSYSolve,'' except that the coefficient of the desired translation is required to equal \(-c\) instead of \(c \cdot i\), to account for the factor of \(i\) introduced in the Fourier transform.

\begin{algorithm}[ht!]
    \caption{Find Coefficients That Enforce Multiplet Closure Using Propogators}\label{alg:susysp}
    \begin{algorithmic}[1]
        \Function{SUSYSolvePropagator}{\textit{bosons}[1..n], \textit{fermions}[1..m], \textit{fermion\_propagators}[1..m], \textit{susy}, \(\Gamma^A\), \textit{consts}[1..p], Inds[1..2], \(c\) = 1}

            \State \(S\), \_ \(\gets\) \textbf{SUSYSolve}(\textit{bosons}, [~], [~], \textit{susy}, \(\Gamma^A\), \textit{consts}, \textit{Inds}, \(c = c\))

            \For{\(A\) in \textit{fermions}}

                \State \(\mathcal{E} \gets {\rm D}_{\textit{Ind}[0]}{\rm D}_{\textit{Ind}[1]}A + {\rm D}_{\textit{Ind}[1]}{\rm D}_{\textit{Ind}[0]}A\)
                \State \(\mathcal{E} \gets\) \textbf{SUSYExpand}(\(\mathcal{E}\), \textit{susy})
                \State \(\mathcal{E} \gets\) \textbf{Evaluate}(\(\mathcal{E}\), \textit{to\_perform\_subs} = False)
                \State \(\mathcal{E} \gets\) \textbf{FierzExpand2Index}(\(\mathcal{E}\), \(\Gamma^A\), \textit{Inds}, \textit{to\_perform\_subs} = False)

                \State \(\mathcal{E} \gets\) \textbf{Fourier}(\(\mathcal{E}\), \textit{fermions}) \Comment{Differs from ``SUSYSolve'' starting here}

                \State \(\tilde{G}_A \gets\) \textbf{GenPropSub}(\(A\), element of \textit{fermion\_propagators} corresponding to \(A\))

                \For{\(i \in 1..m\)} \Comment{Apply \eqref{coupling_current} to each fermion}
                    \State \(f :=\) \textit{fermions}[\(i\)], \(G_f :=\) \textit{fermion\_propagators}[\(i\)]:
                    \State \textit{sub\_rule} \(\gets\) \textbf{GenPropSub}(\(f, G_f\))
                    \State \(\mathcal{E} \gets\) \textbf{substitute}(\(\mathcal{E}\), \textit{sub\_rule})
                \EndFor

                \State \(\mathcal{E} \gets\) \textbf{Evaluate}(\(\mathcal{E}\), \textit{to\_perform\_subs} = False)
                \State \(\mathcal{E} \gets\) \textbf{SubstituteKleinGordon}(\(\mathcal{E}\))

                \State \(T \gets\) \textbf{FilterTermsMomentumSpace}(\(\mathcal{E}\), \(A\), \(\tilde{G}_A\), \textit{Inds})

                \State \textit{desired\_terms} \(\gets T\)[``desired\_terms'']
                \State \textit{undesired\_terms} \(\gets T\)[``undesired\_terms'']

                \State \textit{desired\_terms} \(\gets\) \textbf{Evaluate}(\textbf{sum}(desired\_terms))
                \State \textit{undesired\_terms} \(\gets\) \textbf{Evaluate}(\textbf{ex}\_sum(undesired\_terms))

                \State \textit{desired\_terms} \(\gets\) \textbf{factor\_in}(\textit{desired\_terms}, \textit{consts})
                \State \textit{undesired\_terms} \(\gets\) \textbf{factor\_in}(\textit{undesired\_terms}, \textit{consts})

                \State \textbf{append} ``coefficient of \textit{desired\_terms} in terms of \textit{consts} \(=-c\)'' to \(S\)

                \For{coefficient \(k\), in terms of \textit{consts}, of each term in \textit{undesired\_terms}}
                    \State \textbf{append} ``\(k = 0\)'' to \(S\)
                \EndFor

            \EndFor

            \State \Return \textbf{DistillConstrs}(\(S\), \textit{consts})

        \EndFunction
    \end{algorithmic}
\end{algorithm}

As an example, we apply ``SUSYSolvePropagator'' to the on-shell 4D axial-vector multiplet, whose supersymmetry transformation rules, with variable coefficients, are
\begin{subequations}\label{axial_vector}
    \begin{align}
        {\rm D}_{\alpha} U_{a} &= u (\gamma_* \gamma_{a})_{\alpha}{}^{\beta} \lambda_{\beta} \\
        {\rm D}_{\alpha} \lambda_\beta &= v (\gamma_* \gamma^{a b})_{\alpha \beta} \partial_{a} U_{b}.
    \end{align}
\end{subequations}
The code below finds the constraints that can be obtained from closure. The basis of the Clifford algebra is \eqref{4d_basis}. Notice the use of the function ``SymProp'' from \S\ref{compute_props}, with the input corresponding to the fact that the fermion \(\lambda_{\gamma}\) has spin \(1/2\).
\begin{minted}{python}
>>> bosons = [Ex('U_{a}')]
>>> fermions = [Ex(r'(\lambda)_{\gamma}')]
>>> fermion_propagators = [sym_prop(1/2)]
>>> susy = r'''D_{\alpha}(U_{a}) -> u (\Gamma' \Gamma_{a})_{\alpha}^{\beta} (\lambda)_{\beta}, D_{\alpha}((\lambda)_{\beta}) -> v (\Gamma' \Gamma^{a b})_{\alpha \beta} \partial_{a}(U_{b})'''
>>> basis = [Ex(r'C_{\alpha \beta}'), Ex(r'(\Gamma^{a})_{\alpha \beta}'), Ex(r'(\Gamma^{a b})_{\alpha \beta}'), Ex(r'''(\Gamma')_{\alpha \beta}'''), Ex(r'''(\Gamma' \Gamma^{a})_{\alpha \beta}''')]
>>> consts = ['u', 'v']
>>> indices = [r'_{\alpha}', r'_{\beta}']
>>> susy_solve_propagator(bosons, fermions, fermion_propagators, susy, basis, consts, indices, comm_coef=2)
\end{minted}
The result is the constraint
\begin{equation}\label{axial_vector_closure}
    u v = i.
\end{equation}
Fixing \(u = i\) gives \(v = 1\), agreeing with (2.3) of \cite{Gates2020a}.

As noted at the end of \S\ref{multiplet_solver2}, ``SUSYSolvePropagator'' does not suffer of the same limitations as ``SUSYSolve'' concerning the elimination of equation-of-motion terms. To give a picture of the difference in applicability, Table \ref{tab:solver_success} indicates on which multiplets from \cite{Gates2020a} each algorithm is successful. Note, though, that because the substitution corresponding to \eqref{coupling_current} is difficult to reverse, the hash table of classified terms from ``FilterTermsMomentumSpace'' is less tractable for othr purposes than that from ``FilterTerms.'' In particular, unlike ``SUSYSolve,'' ``SUSYSolvePropagator'' cannot offer an explicit closure structure along with the outputted set of constraints.

\begin{table}[ht!]
    \begin{center}
        \begin{tabular}{| c | c | c |}
            \hline
            \textbf{Multiplet} & \textbf{SUSYSolve} & \textbf{SUSYSolvePropagator} \\
            \hline
            \hline
            Chiral & \cmark & \cmark \\ \hline
            Vector & \cmark & \cmark \\  \hline
            Axial-Vector & \cmark & \cmark \\  \hline
            Matter-Gravitino & \xmark & \cmark \\ \hline
            Supergravity & \xmark & \cmark \\ [0.5ex]
            \hline
        \end{tabular}
        \caption{Success of ``SUSYSolve'' and ``SUSYSolvePropagator'' on a variety of 4D, \(\mathcal{N} = 1\) multiplets. A \cmark\, indicates success, while an \xmark\, indicates failure.}
        \label{tab:solver_success}
    \end{center}
\end{table}

\subsection{SUSY-Invariant Action Solver}\label{action_solver}
We now briefly treat the algorithm drawing constraints from SUSY-invariance of the action. Recall the requirement of ``SUSY-invariance'' means that the action \(S\) is stationary with respect to variation by the supersymmetry transformation,\footnote{SUSY-invariance is frequently defined as the requirement that for some quantity \({\cal J}{}_{\alpha}{}^{a}\) (the supercurrent) \({\rm D}_{\alpha}S = \int d^{D}x \big[\partial_{a}{\cal J}{}_{\alpha}{}^{a}\big]\), or equivalently, that the supercovariant derivative of the Lagrangian density is a total derivative. Since we typically assume that the supercurrent vanishes at the boundary of integration, this condition is equivalent to the one we use.} i.e., \({\rm D}_{\alpha}S = 0\), or equivalently,
\begin{equation}\label{susy_inv_def}
    \int d^{D}x \big[{\rm D}_{\alpha}\mathcal{L}\big] = 0,
\end{equation}
where \(\mathcal{L}\) is the Lagrangian density. The input to the algorithm ``MakeActionSUSYInv'' (see Algorithm \ref{alg:masusyinv}) is the integrand to that expression (with the supercovariant derivative applied symbolically; it need not be evaluated), the supersymmetry transformation rules for the multiplet, the coefficients (``consts'') of the multiplet and the action components, and the names of the fermions in the multiplet. The procedure is simply to evaluate \(D_{\alpha}\mathcal{L}\), construct the left-hand side of \eqref{susy_inv_def}, integrate by parts (here we use Cadabra's \mintinline{python}{integrate_by_parts()}) to move all partial derivatives onto the bosons so that like terms may be combined or cancelled, and factor in ``consts.'' At this point, each term is clearly independent, as the terms involve distinct fields or partial-derivative indices. Hence, the coefficient of each term can be set to zero, giving the desired system of constraints.

\begin{algorithm}[ht!]
    \caption{Make Action SUSY Invariant}\label{alg:masusyinv}
    \begin{algorithmic}[1]
        \Function{MakeActionSUSYInv}{\({\rm D}_{\alpha}\mathcal{L}\), \textit{susy}, \textit{consts}[\(1..n\)], \textit{fermion\_names}[\(1..m\)]}
            \State \({\rm D}_{\alpha}\mathcal{L} \gets\) \textbf{SUSYExpand}(\({\rm D}_{\alpha}\mathcal{L}\), \textit{susy})
            \State \({\rm D}_{\alpha}\mathcal{L} \gets\) \textbf{Evaluate}(\({\rm D}_{\alpha}\mathcal{L}\))

            \State \({\rm D}_{\alpha} S \gets \int dx^D {\rm D}_{\alpha}\mathcal{L}\) \Comment{Symbolically integrating in Cadabra}
            \For{\(i \in 1..m\)}
                \State \({\rm D}_{\alpha} S \gets\) \textbf{integrate\_by\_parts}(\({\rm D}_{\alpha} S\), \textit{fermion\_names}[\(i\)])
            \EndFor

            \State \textit{integrand} \(\gets\) \textbf{nth\_arg}(\({\rm D}_{\alpha} S\), 0).\textbf{ex}() \Comment{extract the integrand, the 0th argument of \({\rm D}_{\alpha}S\)'s ExNode structure}
            \State \textit{integrand} \(\gets\) \textbf{Evaluate}(\textit{integrand})
            \State \textit{integrand} \(\gets\) \textbf{factor\_in}(\textit{integrand}, \textit{consts})

            \State \textit{system} \(\gets\) [ ]
            \For{\textit{term} in \textit{integrand}}
                \State \textbf{append} ``coefficient of \textit{term} = 0'' to \textit{system}
            \EndFor

            \State \textit{sol} \(\gets\) \textbf{DistillConstrs}(\textit{system}, \textit{consts})

            \State \Return \textit{sol}
        \EndFunction
    \end{algorithmic}
\end{algorithm}

As an example, consider the off-shell 4D vector multiplet, whose supersymmetry rules, with variable coefficients, are
\begin{subequations}\label{off_shell_vect}
    \begin{align}
        {\rm D}_{\alpha} A_{a} &= u (\gamma_{a})_{\alpha}{}^{\beta} \lambda_{\beta} \\
        {\rm D}_{\alpha} \lambda_\beta &=  v (\gamma^{a b})_{\alpha \beta} \partial_{a} A_{b} + w C_{\alpha \beta} \partial^{a} A_{a} + x (\gamma_{*})_{\alpha \beta} d \\
        {\rm D}_{\alpha} d &= y (\gamma_{*} \gamma^a)_{\alpha}{}^{\beta} \partial_{a} \lambda_{\beta}.
    \end{align}
\end{subequations}
The corresponding action, with variable coefficients, is
\begin{equation}\label{off_shell_vect_act}
    S = \int dx^{D} \bigg\{\ell F_{a b} F^{a b} + m \left(\gamma^{a}\right)^{\alpha \beta} \lambda_{\alpha} \partial_{a}\lambda_{\beta} + n d^2\bigg\},
\end{equation}
where
\begin{equation}\label{field_strength}
    F_{a b} = \partial_{a}A_{b} - \partial_{b}A_{a}
\end{equation}
is the field strength tensor. The code below enters the symbolically SUSY-transformed Lagrangian density and the off-shell 4D vector multiplet SUSY-rule\footnote{Notice that a substitution rule for the field strength tensor has been added so that the tensor can be recognized and converted into a gauge field expression more tractable for the program.} into ``MakeActionSUSYInv.''
\begin{minted}{python}
>>> L = Ex(r'D_{\gamma}(l F_{a b} F^{a b} + m (\Gamma^{a})^{\alpha \beta} (\lambda)_{\alpha} \partial_{a}((\lambda)_{\beta}) + n d d)')
>>> susy = r'''D_{\alpha}(A_{a}) -> u (\Gamma_{a})_{\alpha}^{\beta} (\lambda)_{\beta}, D_{\alpha}((\lambda)_{\beta}) -> v (\Gamma^{a b})_{\alpha \beta} \partial_{a}(A_{b}) + w C_{\alpha \beta} \partial^{a}(A_{a}) + x (\Gamma')_{\alpha \beta} d, D_{\alpha}(d) -> y (\Gamma' \Gamma^{a})_{\alpha}^{\beta} \partial_{a}((\lambda)_{\beta}), F_{a b} -> \partial_{a}(A_{b}) - \partial_{b}(A_{a})'''
>>> make_action_susy_inv(L, susy, ['l', 'm', 'n', 'u', 'v', 'w', 'x', 'y'], [r'\lambda'])
\end{minted}
The result is the set of constraints
\begin{subequations}
    \begin{align}
        \ell u &= -\frac12 m v \\
        m x &= n y \\
        m w &= 0 \label{4d_vect_action_constr}.
    \end{align}
\end{subequations}
Now, one can verify via ``SUSYSolve'' or ``SUSYSolvePropagator'' that the constraints from closure are
\begin{subequations}\label{off_shell_vect_closure}
    \begin{align}
        u v &= -i \label{4d_vect_constr:1} \\
        u w &= 0 \label{4d_vect_constr:2} \\
        x y &= i \label{4d_vect_constr:3}.
    \end{align}
\end{subequations}
Together with the normalization conditions from \S\ref{4d_vector_norm}, the choice of scaling \(u = 1\), and the assumption that \(y\) is positive imaginary, these constraints can be solved to find
\begin{equation}
    u = 1, \quad v = -i, \quad w = 0, \quad x = 1, \quad y = i, \quad \ell = -\frac14, \quad m = \frac{i}{2}, \quad n = \frac12,
\end{equation}
agreeing with (19) of \cite{Gates2009} and (2.12) of \cite{Chappell2013}.

\subsection{Algorithm for Computing Holoraumy}\label{holoraumy_alg}
Our final algorithm computes holoraumy. Since the distinction between holoraumy and closure is simply in taking the commutator rather than the anticommutator of supercovariant derivatives, we adapt the procedure of ``SUSYSolve.'' However, now the intended output is the classification of terms, which is useful because in every single on-shell 4D, \(\mathcal{N} = 1\) multiplet studied in \cite{Gates2020a}, the fermionic holoraumy included equation-of-motion terms; in 4D supergravity, there were even gauge terms. The algorithm ``Holoraumy'' (see \ref{alg:holoraumy}) takes as input an array ``fields'' of fields, the supersymmetry transformation rules ``susy'' for the multiplet, a basis \(\Gamma^{A}\) for the Clifford algebra, substitution rules ``subs'' which take the coefficients of the multiplets to their values,\footnote{The imagined use of ``Holoraumy'' is as a calculation which follows the solution of a multiplet. It would be inconvenient to rewrite the supersymmetry transformation rules with the solved coefficients; for this reason, we have preferred to enable a separate listing of the computed values. See the end of this section for examples.} and the indices ``Inds'' to be used in the commutator of supercovariant derivatives. The procedure is to evaluate and Fierz-expand the commutator of supercovariant derivatives applied to each field and classify terms via ``FilterTerms.'' Since the classification of terms in holoraumy is between ``regular'' terms, gauge terms, and equation-of-motion terms, we compile together the desired and undesired terms from ``FilterTerms'' as ``regular.'' Note that we have chosen to leave the set of ``Lorentz proper'' terms as is without determining which among them are bona fide equation-of-motion terms. The reason is that neither approach for definitively discerning and removing equation-of-motion terms from an expression \(\mathcal{E}\) is guaranteed to return terms originally present in \(\mathcal{E}\): ``FindNonGaugeInv'' only achieves this up to unknown gauge-invariant quantities, while ``FilterTermsMomentumSpace'' reformulates \(\mathcal{E}\) in terms of coupling currents without reconstructing the expression in terms of the original fields.

\begin{algorithm}[ht!]
    \caption{Calculate the Holoraumy}\label{alg:holoraumy}
    \begin{algorithmic}[1]
        \Function{Holoraumy}{\textit{fields}[\(1..n\)], \textit{susy}, \(\Gamma^A\), \textit{subs}, \textit{Inds}[1..2]}

            \State \textit{holoraumy\_dict} \(\gets\) empty hash table

            \For{\(A \in\) fields}
                \State \(\mathcal{E} \gets {\rm D}_{\textit{Ind}[0]}{\rm D}_{\textit{Ind}[1]}A - {\rm D}_{\textit{Ind}[1]}{\rm D}_{\textit{Ind}[0]}A\)
                \State \(\mathcal{E} \gets\) \textbf{SUSYExpand}(\(\mathcal{E}\), \textit{susy})
                \State \(\mathcal{E} \gets\) \textbf{substitute}(\(\mathcal{E}\), \textit{subs})
                \State \(\mathcal{E} \gets\) \textbf{distribute}(\(\mathcal{E}\))
                \State \(\mathcal{E} \gets\) \textbf{Evaluate}(\(\mathcal{E}\), \textit{to\_perform\_subs} = False)
                \State \(\mathcal{E} \gets\) \textbf{FierzExpand2Index}(\(\mathcal{E}\), \(\Gamma^A\), \textit{Inds}, \textit{to\_perform\_subs} = False)

                \State \(T \gets\) \textbf{FilterTerms}(\(\mathcal{E}\), \(A\), \textit{Inds}, \textit{identify\_lorentz\_proper} = \(A\) is a fermion)

                \State \textit{regular\_terms} \(\gets T\)[``desired\_terms''] + \(T\)[``undesired\_terms'']
                \State \textit{gauge\_terms} \(\gets T\)[``gauge\_terms'']
                \State \textit{lorentz\_proper\_terms} \(\gets T\)[``lorentz\_proper\_terms'']

                \State \textit{regular\_terms} \(\gets\) \textbf{Evaluate}(\textbf{sum}(\textit{regular\_terms}))
                \State \textit{gauge\_terms} \(\gets\) \textbf{Evaluate}(\textbf{sum}(\textit{gauge\_terms}))
                \State \textit{lorentz\_proper\_terms} \(\gets\) \textbf{Evaluate}(\textbf{sum}(\textit{lorentz\_proper\_terms}), \textit{to\_join\_gamma} = False)

                \State \textit{holoraumy}\_dict[\(A\)] \(\gets\) \{``regular\_terms'': \textit{regular\_terms}, ``gauge\_terms'': \textit{gauge\_terms}, ``lorentz\_proper\_terms'': \textit{lorentz\_ proper\_terms}\}
            \EndFor

            \State \Return holoraumy\_dict
        \EndFunction
    \end{algorithmic}
\end{algorithm}

As an example of use, consider the on-shell 4D supergravity multiplet, whose supersymmetry rules, with variable coefficients, are
\begin{subequations}\label{4d_supergravity}
    \begin{align}
        {\rm D}_{\alpha} h_{a b} &= u (\gamma_{a})_{\alpha}{}^{\beta} \Psi_{b \beta} + u (\gamma_{b})_{\alpha}{}^{\beta} \Psi_{a \beta} \\
        {\rm D}_{\alpha} \Psi_{a \beta} &= v (\gamma^{b c})_{\alpha \beta} \partial_{b} h_{c a}.
    \end{align}
\end{subequations}
One can verify (use ``SUSYSolvePropagator''!) that we can take \(u = 1/2, v = -i\) (also, see (2.5) in \cite{Gates2020a}). The code below computes the holoraumy of the multiplet. (Again, the basis for the Clifford algebra is \eqref{4d_basis}.)
\begin{minted}{python}
>>> fields = [Ex('h_{a b}'), Ex(r'(\Psi_{a})_{\gamma}')]
>>> susy = r'''D_{\alpha}(h_{a b}) -> u (\Gamma_{a})_{\alpha}^{\beta} (\Psi_{b})_{\beta} + u (\Gamma_{b})_{\alpha}^{\beta} (\Psi_{a})_{\beta}, D_{\alpha}((\Psi_{a})_{\beta}) -> v (\Gamma^{b c})_{\alpha \beta} \partial_{b}(h_{c a})'''
>>> subs = Ex('u -> 1/2, v -> -I', False)
>>> indices = [r'_{\alpha}', r'_{\beta}']
>>> holoraumy(fields, susy, basis, subs, indices)
\end{minted}
In terms of the Rarita-Schwinger terms \eqref{rarita_schwinger}, the result is
\begin{subequations}
    \begin{align}
        [{\rm D}_{\alpha}, {\rm D}_{\beta}] h_{a b} &= \epsilon_{(a|}{}^{c d e} \left(\gamma_* \gamma_{c}\right)_{\alpha \beta} \partial_{d} h_{|b) e} \\
        \begin{split}\label{our_holo}
            [{\rm D}_{\alpha}, {\rm D}_{\beta}] \Psi_{a \gamma} &= 2i \left(\gamma^{b} \gamma_*\right)_{\alpha \beta} \left(\gamma_*\right)_{\gamma}{}^{\eta} \partial_{b} \Psi_{a \eta} \textcolor{DarkGreen}{+~  \frac{i}{4} C_{\alpha \beta} \left(\gamma_{a}\right)_{\gamma}{}^{\eta} R_{\eta}} \\
            &\quad \textcolor{DarkGreen}{+  \frac{i}{4} \left(\gamma_*\right)_{\alpha \beta} \left(\gamma_* \gamma_{a}\right)_{\gamma}{}^{\eta} R_{\eta} -  \frac18 \epsilon_{a b}{}^{e f} \left(\gamma_* \gamma^{b}\right)_{\alpha \beta} \left(\gamma_{e f}\right)_{\gamma}{}^{\eta} R_{\eta}} \\
            &\quad \textcolor{DarkGreen}{-  \frac{i}{4} \left(\gamma_* \gamma_{a}\right)_{\alpha \beta} \left(\gamma_*\right)_{\gamma}{}^{\eta} R_{\eta} +  \frac{5i}{4} C_{\alpha \beta} E_{a \gamma}} \\
            &\quad \textcolor{DarkGreen}{+  \frac{5i}{4} \left(\gamma_*\right)_{\alpha \beta} \left(\gamma_*\right)_{\gamma}{}^{\eta} E_{a \gamma} -  \frac{3i}{4} \left(\gamma_* \gamma^{b}\right)_{\alpha \beta} \left(\gamma_* \gamma_{b}\right)_{\gamma}{}^{\eta} E_{a \eta}} \\
            &\quad \textcolor{DarkGreen}{+  \frac{i}{2} \left(\gamma_* \gamma^{b}\right)_{\alpha \beta} \left(\gamma_* \gamma_{a}\right)_{\gamma}{}^{\eta} E_{b \eta} -  \frac{3i}{4} C_{\alpha \beta} \left(\gamma^{b}\right)_{\gamma}{}^{\eta} \partial_{a} \Psi_{b \eta}} \\
            &\quad \textcolor{DarkGreen}{-  \frac{3i}{4} \left(\gamma_*\right)_{\alpha \beta} \left(\gamma_* \gamma^{b}\right)_{\gamma}{}^{\eta} \partial_{a} \Psi_{b \eta} +  \frac{i}{4} \left(\gamma_* \gamma^{b}\right)_{\alpha \beta} \left(\gamma_* \gamma_{b} \gamma^{c}\right)_{\gamma}{}^{\eta} \partial_{a} \Psi_{c \eta}} \\
            &\quad \textcolor{blue}{-~i \left(\gamma^{b} \gamma_*\right)_{\alpha \beta} \left(\gamma_*\right)_{\gamma}{}^{\eta} \partial_{a} \Psi_{b \eta}},
        \end{split}
    \end{align}
\end{subequations}
where the terms in green are equation-of-motion terms and the term in blue is a gauge term. One can show that this is equivalent to
\begin{subequations}
    \begin{align}
        [{\rm D}_{\alpha}, {\rm D}_{\beta}] h_{a b} &= -2i \mathscr{B}_{a b c d e \alpha \beta} \partial^{c} h^{d e} \\
        [{\rm D}_{\alpha}, {\rm D}_{\beta}] \Psi_{a \gamma} &= -2i \mathscr{F}_{a b c \alpha \beta \gamma}{}^{\eta} \partial^{b} \Psi^{c}{}_{\eta} - \mathscr{Z}_{a \alpha \beta \gamma} - \partial_{a} \zeta_{\alpha \beta \gamma} \label{4d_sugra_holo_answer},
    \end{align}
\end{subequations}
where\footnote{We have corrected small typographical errors in the signs in this expression present in \cite{Gates2020a}.}
\begin{subequations}
    \begin{align*}
        \mathscr{B}_{a b c d e \alpha \beta} &= \frac{i}{2} \eta_{e (a} \epsilon_{b) c d f} \left(\gamma_* \gamma^{f}\right)_{\alpha \beta} \numberthis \\
        \mathscr{F}_{a b c \alpha \beta \gamma}{}^{\eta} &= -\frac34 \eta_{a [b} \left(\gamma_* \gamma_{c]}\right)_{\alpha \beta} \left(\gamma_*\right)_{\gamma}{}^{\eta} - \frac{i}{8} \epsilon_{a b c d} \left(\gamma_* \gamma^{d}\right)_{\alpha \beta} \delta_{\gamma}{}^{\eta} + \frac18 \left(\gamma_* \gamma_{[b}\right)_{\alpha \beta} \left(\gamma_* \gamma_{c] a}\right)_{\gamma}{}^{\eta} \numberthis \\
        \mathscr{Z}_{a \alpha \beta \gamma} &= i \bigg\{-C_{\alpha \beta} \left(\gamma_{a}\right)_{\gamma}{}^{\eta} - \left(\gamma_*\right)_{\alpha \beta} \left(\gamma_* \gamma_{a}\right)_{\gamma}{}{}^{\eta} + \frac12 \left(\gamma_* \gamma^{d}\right)_{\alpha \beta} \left(\gamma_* \gamma_{d} \gamma_{a}\right)_{\gamma}{}^{\eta} - \frac14 \left(\gamma_* \gamma_{a}\right)_{\alpha \beta} \left(\gamma_*\right)_{\gamma}{}^{\eta}\bigg\} \left(\gamma^{e f}\right)_{\eta}{}^{\rho} \partial_{e} \Psi_{f \rho} \\
        &\quad + \frac14 \bigg\{5 C_{\alpha \beta} \delta_{\gamma}{}^{\eta} + 5 \left(\gamma_*\right)_{\alpha \beta} \left(\gamma_*\right)_{\gamma}{}^{\eta} - 2 \left(\gamma_* \gamma^{d}\right)_{\alpha \beta} \left(\gamma_* \gamma_{d}\right)_{\gamma}{}^{\eta}\bigg\} \epsilon_{a}{}^{e f g} \left(\gamma_* \gamma_{e}\right)_{\eta}{}^{\rho} \partial_{f} \Psi_{g \rho} \numberthis \\
        \zeta_{\alpha \beta \gamma} &= \frac{3i}{4} \bigg\{C_{\alpha \beta} \left(\gamma^{b}\right)_{\gamma}{}^{\eta} + \left(\gamma_*\right)_{\alpha \beta} \left(\gamma_* \gamma^{b}\right)_{\gamma}{}^{\eta} + \left(\gamma_* \gamma^{b}\right)_{\alpha \beta} \left(\gamma_*\right)_{\gamma}{}^{\eta} - \frac13 \left(\gamma_* \gamma_{c}\right)_{\alpha \beta} \left(\gamma_* \gamma^{c b}\right)_{\gamma}{}^{\eta}\bigg\} \Psi_{b \eta}, \numberthis
    \end{align*}
\end{subequations}
agreeing with (6.38) of \cite{Gates2020a}.

\section{Conclusions}\label{discussion_sect}
We should like to conclude with a few comments about the nature of our results and future directions. In the foregoing, we have derived, in a few lines of code, a form for the linearized 11D, \(\mathcal{N} = 1\) supergravity (and its action) multiplet analogous to the forms studied for 4D theories in the Genomics project. Using this solution, we have given a comprehensive picture of the non-closure geometry of 11D, \(\mathcal{N} = 1\) in the style of \cite{Gates2020a}. Further, we have computed the holoraumy of 11D, \(\mathcal{N} = 1\) supergravity for the first time. It turns out that 11D, \(\mathcal{N} = 1\) supergravity provides the second counterexample to the conjectured ubiquity of electromagnetic-duality-type rotations in holoraumy \cite{Gates2020a}, after 10D, \(\mathcal{N} = 1\) super Maxwell theory \cite{Gates2022}.

These results were obtained using the new Cadabra module SusyPy. We have provided computer algorithms, particularly ``SpinorCombine'' and ``FindChain,'' for computer computation and canonicalization of spinor-index expressions. We have built on this tensor-spinor symbolic algebra to provide a new Fierz expansion algorithm, ``FierzExpansion2Index'' that operates in terms of spinor indices (rather than spinor bilinears) and requires a piori knowledge of only two of the four manipulated spinor indices, an essential prerequisite for automatic manipulation of (anti)commutators of supercovariant derivatives. Combining these with procedures to filter out gauge-invariant quantities or manipulate Feynman propagators, we have provided two algorithms, ``SUSYSolve'' and ``SUSYSolvePropagator,'' that obtains constraints on unknown coefficients in supersymmetry rules from closure of the algebra. The on-shell multiplets to which ``SUSYSolve'' can successfully applied are limited by the occasional presence of non-closure functions not amenable to the gamma-splitting procedure of ``SUSYSolve,'' but when it succeeds, ``SUSYSolve'' provides an explicit closure structure. While it does not provide an explicit closure structure, ``SUSYSolvePropagator'' bypasses the limitation of ``SUSYSolve'' and has been applied successfully to every multiplet in \cite{Gates2020a}. The two algorithms assume that any equation-of-motion terms of an on-shell multiplet lie on Lorentz-symmetric fermions. (They both work on any off-shell multiplet.) Finally, we have provided an algorithm ``MakeActionSUSYInv'' that obtains constraints on SUSY-transformation and action coefficients from SUSY-invariance of the action, and an algorithm ``Holoraumy'' that provides the holoraumy of a multiplet, with gauge and potential equation-of-motion terms identified.

Further work remains in terms of improving the generality of our algorithms. Three possible extensions are clear. First, our multiplet solvers ``SUSYSolve'' and ``SUSYSolvePropagator'' should be extended to the case that equation-of-motion terms do not lie on symmetric fermions, particularly to the case that they lie on bosons. This is a highly pathological situation, but it does occur, as seen in the 4D double-tensor multiplet \cite{Gates2009, Townsend1982}. The difficulty is that bosonic equations of motion are second-order, so they do not appear themselves in the closure on a bosonic field. Rather, the equation-of-motion terms give rise to terms in the anticommutator of supercovariant derivatives applied to the field strength of the bosons that are second-order and vanish under the equations of motion (\S3 of \cite{Townsend1982}). Second, our algorithms, from the solvers to ``Holoraumy,'' should be extended to higher-\(\mathcal{N}\) supermultiplets, which comes down to treating multiple supercovariant derivatives simultaneously and their relations. Third, our spinor-index arithmetic/canonicalization in ``SpinorCombine'' should be extended to treat spinor-indexed sigma matrices. In even dimensions (typically besides \(D = 4\)), and especially in dimensions \(D \equiv 2~(\textnormal{mod}~8)\) for which both the Weyl and Majorana conditions give reductions in the number of independent components, it is more natural to consider Pauli sigma matrices over Dirac gamma matrices. Previously, a Mathematica tool for sigma-matrix products \cite{Gates2001a} has been demonstrated, employing a matrix representation rather than a purely symbolic approach. Although it did not bear a robust evaluation/canonicalization system like ours, it successfully generated numerous Fierz identities. Version 1 of Cadabra \cite{Peeters2012} also seems to have included an attempt at implementing sigma matrices using the conventions of \cite{Wess1992}, but without any computer arithmetic created for these objects. An implementation with spinor indices as in Appendix B of \cite{Gates2019} could extend the applicability of our multiplet solvers to more even-dimensional multiplets.

Closely tied to these improvements are possible extensions of the Genetics project and our physical results. In particular, one should consider the non-closure geometry and holoraumy of the alternative linearized 11D, \(\mathcal{N} = 1\) supergravity multiplet obtained by replacing the three-form \(A_{a b c}\) with a six-form \(\tilde{A}_{a b c d e f}\) \cite{Nicolai1980}. This multiplet is dual to the one used here (i.e., to that from the field content of \cite{Cremmer1978}). Previous work \cite{Townsend1982} has posited that the bosonic equation-of-motion terms should be present in this dual 11D supergravity multiplet, as they are for the 4D double-tensor multiplet. Hence, an extension of our multiplet solvers to consider bosonic equation-of-motion terms should prove useful. Also, various 10D, \(\mathcal{N} = 1\) supergravity multiplets have also become of interest in the literature in the context of the off-shell SUSY problem \cite{Gates2019}, so the Genetics project should be extended to them as well. Due to the dimension, an extension of our tensor-spinor symbolic algebra to handle sigma matrices should prove useful. Beyond these extensions of Supersymmetry Genetics, now that we have computed the holoraumy of the 11D, \(\mathcal{N} = 1\) multiplet, one should explore the explicit computation of the action of the 11D Pauli-Lubanski 3-form operator on the fields of the multiplet, which is closely tied to holoraumy.\footnote{This possible future direction was pointed out to us by Konstantinos Koutrolikos.}

Finally, we briefly remark on computational complexity. We have not paid much attention to this because of the high cost of ordinary index canonicalization (in the original sense of lexicographic ordering of indices via the exploitation of tensor symmetries). Cadabra uses the Butler-Portugal algorithm \cite{Butler1991,Manssur2002} (in a form drawn from \cite{MartinGarcia2008b}) for canonicalization, which suffers cost factorial in the number of indices for certain expressions \cite{Niehoff2018}. This leads to drastically varying running time: solving the on-shell 4D chiral multiplet with ``SUSYSolve'' takes 8.72 seconds on a computer with a processor of measured average speed 4.2 GHz, while solving the 11D supergravity multiplet takes 1 hour, 31 minutes, and 16 seconds. However, there is certainly room for efficiency improvements in SusyPy itself.\footnote{It should be noted that a reduction in runtime might be achieved if Cadabra used the improved canonicalization algorithm in \cite{Niehoff2018}, which broadens the number of cases that can be treated in polynomial-time.} Future efforts should reduce calls to Cadabra's canonicalization function, as well as explore different paradigms for tensor symbolic algebra, such as the index isomorphism approach of Redberry \cite{Poslavsky2015} which demonstrated shorter runtimes than Cadabra.

In summary, we have developed a module SusyPy for multiplet computations that should prove indispensible as the Supersymmetry Genomics project explores higher-dimensional and on-shell multiplets. We have successfully demonstrated our algorithm in the case of 11D, \(\mathcal{N} = 1\) supergravity. Future work should focus on extending our software, applying our tools to further multiplets, and attempting greater efficiency. We fervently look forward to these endeavors.

\vspace{.05in}
\begin{center}
\parbox{4in}{{\it ``For truth is eternal, it is divine; and no phase in the development of truth, however small the domain it embraces, can pass away without a trace. It remains even if the garments in which feeble men clothe it fall into dust.'' \\ ${~}$ 
 ${~}$ 
\\ ${~}$ }\,\,-\,\, H. G. \, Grassmann \, \cite{Grassmann1862} \\ ${~}$ \quad\,(translated by L. C. \, Kannenberg) $~~~~~~~~~$}
 \parbox{4in}{
 $~~$}  
 \end{center}

\noindent
{\bf Acknowledgments}\\[.1in] \indent
S.J.G. was supported in this research in part by the endowment of the Ford Foundation Professorship of Physics at Brown University and the Brown Theoretical Physics Center. S.J.G. is also grateful to acknowledge the support of the endowment of the Clark Leadership Chair in Science at the University of Maryland, College Park. Additionally, I.B.H. and S.H. would like to thank S.J.G. and the other organizers of the 2020 SSPTRS (Summer Student Theoretical Physics Research Session) at Brown University for immensely valuable mentorship from which this research originated. Finally, the authors wish to express deep appreciation to Yangrui Hu and S.-N. Hazel Mak for discussions that played a central role at the inauguration of this work.

\appendix
\numberwithin{equation}{section}

\section{Notation and Conventions}\label{conventions}
Here, we briefly summarize the notation and conventions we use in this paper. Following \cite{Gates2019}, we denote Lorentz, or vector, indices by early-alphabet lowercase Latin letters, while we denote Dirac, or spinor, indices by lowercase Greek letters. Indices contained in parentheses, e.g., \(a, b, c\) in the expression \(A_{(a} B_{b} C_{c)}\), are symmetrized, i.e., expressions with indices in parentheses are summed over all permutations of those indices. Indices contained in square brackets, e.g., \(a, b, c\) in the expression \(A_{[a} B_{b} C_{c]}\), are antisymmetrized, so that the sign of a summand is determined by the parity of the corresponding permutation of the indices. We use pipes to block out indices of the same parity (co- or contravariance) from the (anti)symmetrized set of indices, e.g., in the expression \(A_{(a|} B_{b} C_{|c)}\), only \(a, c\) are symmetrized. Unlike in \cite{Freedman2013}, for simplicity, we do not normalize symmetrizations; that is, we do not divide by the number of permutations.

Moving on from indices, we define \([A, B]\) not to be the \(\mathbb{Z}_2\)-graded Lie bracket, but always to be \(AB - BA\). Similarly, \(\{A, B\} = AB + BA\) always. It is simpler to use these conventions for the commutator and anticommutator for our purposes, since both come into play when considering simultaneously calculations regarding closure of a supermultiplet and a supermultiplet's holoraumy. Finally, when considering multiplets in even dimension \(D\) in \S\ref{susypy}, we follow \cite{Freedman2013} and use the dimension-independent notation \(\gamma_{*}\) for the highest-rank element
\begin{equation}
    \gamma_{*} = (-i)^{(D/2) + 1} \gamma_0 \cdots \gamma_{D - 1},
\end{equation}
rather than a dimension-dependent notation, e.g., the common \(\gamma^5\) for \(D = 4\).

Now that we have summarized our notation, we describe our conventions for the 11D Clifford algebra and relevant tensors. We follow the conventions of \cite{Gates2019}; see Appendix A in that paper for details. We content ourselves here with sketching the most essential of these points. The gamma matrices with spinor indices are of the form \((\gamma^{a_1 \cdots a_n})_{\alpha_1 \alpha_2}\) for \(n \leq 10\). The algebra itself is defined by the usual equation
\begin{equation}
    \{\gamma^{a}, \gamma^{b}\} = 2 \eta^{a b} \mathds{1},
\end{equation}
where \(\mathds{1}\) is the \(32 \times 32\) identity matrix and \(\eta^{a b}\) is the metric, for which we choose the signature
\begin{equation}
    \eta^{a b} = \textnormal{diag}(-1, +1, +1, +1, +1, +1, +1, +1, +1, +1, +1).
\end{equation}
We take the basis
\begin{equation}\label{11d_basis_redux}
    \Gamma^{A} = \left\{C_{\alpha \beta}, \left(\gamma^{a}\right)_{\alpha \beta}, \left(\gamma^{a b}\right)_{\alpha \beta}, \left(\gamma^{a b c}\right)_{\alpha \beta}, \left(\gamma^{a b c d}\right)_{\alpha \beta}, \left(\gamma^{a b c d e}\right)_{\alpha \beta}\right\}.
\end{equation}
Of these basis elements, the elements
\begin{equation}
    \Gamma^{A}_{+} = \left\{\left(\gamma^{a}\right)_{\alpha \beta}, \left(\gamma^{a b}\right)_{\alpha \beta}, \left(\gamma^{a b c d e}\right)_{\alpha \beta}\right\}
\end{equation}
are symmetric in their spinor indices and the elements
\begin{equation}
    \Gamma^{A}_{-} = \left\{C_{\alpha \beta}, \left(\gamma^{a b c}\right)_{\alpha \beta}, \left(\gamma^{a b c d}\right)_{\alpha \beta}\right\}
\end{equation}
are antisymmetric. Since \(C_{\alpha \beta}\) is antisymmetric, we enforce the NW-SE convention for multiplication of tensor-spinors. The representation we have chosen for our basis is the ``negative'' one obtained by adjoining to the basis for the 10D Clifford algebra the negative \(-\gamma_{*}\) of the canonical highest-rank element. This implies the gamma-matrix reduction rule
\begin{equation}
        \gamma^{a_1 \cdots a_r} = \frac{1}{(11-r)!} \epsilon^{a_1 \cdots a_{11}} \gamma_{a_{11} \cdots a_{r+1}}.
\end{equation}
Finally, we define the Levi-Civita tensor as in \cite{Freedman2013}, so that the product of Levi-Civita tensors with shared indices first has a negative coefficient, viz.,
\begin{equation}
    \epsilon_{a_1 \cdots a_{n} b_1 \cdots b_{11 - n}} \epsilon^{a_1 \cdots a_{n} c_1 \cdots c_{11 - n}} = -n! \delta_{b_1}{}^{[c_1} \cdots \delta_{b_{11 - n}}{}^{c_{11 - n}]}.
\end{equation}

\section{Action Normalization}\label{action_norm}
Here, we discuss our convention for normalization of the bosonic part of the action corresponding to a multiplet. We frame this convention for actions which have no interaction terms, so that each field \(A\) has its own free-field Lagrangian density \(\mathcal{L}_{A}\). The rule essentially regards 0-brane reduction: if a dynamical field \(A\) (i.e., a field which has non-trivial equations of motion) is replaced with a field \(\hat{A}\) with only time dependence (i.e., no spatial dependence), then the Lagrangian density takes the form of a standard 1D mass-1 kinetic energy term
\begin{equation}
    \hat{\mathcal{L}}_{A} = \frac12 (\dot{\hat{A}})^2.
\end{equation}
If \(A\) is a gauge field, then the requirement is that replacement with a solely time-dependent field yields a Lagrangian density for \(A\) which is the sum of standard 1D mass-1 kinetic energy terms, plus other terms not of the form of a squared derivative. If, on the other hand, \(A\) is an auxiliary (i.e., non-dynamical) field,\footnote{All of the fields in the 11D supergravity multiplet are dynamical, but we include the non-dynamical case for completeness and for the example multiplet solution in \S\ref{action_solver}.} then we simply require that its free-field Lagrangian be of the form
\begin{equation}
    \mathcal{L}_{A} = \frac12 A^2.
\end{equation}
In the remainder of this appendix, we provide the constraints from normalization for the example from \S\ref{action_solver} and for the 11D supergravity multiplet.

\subsection{4D Vector Multiplet}\label{4d_vector_norm}
Consider the off-shell 4D vector multiplet.\footnote{See \S2 of \cite{Chappell2013} for a compilation of Lagrangian densities for 4D multiplets and their 0-brane reductions.} Using the notation of \S\ref{action_solver}, the Lagrangian densities for the dynamical vector boson and the auxiliary pseudoscalar boson are
\begin{subequations}
    \begin{align}
        \mathcal{L}_A &= \ell F_{a b} F^{a b} \\
        \mathcal{L}_d &= n d^2.
    \end{align}
\end{subequations}
Replace the vector field with a field only dependent on time, as below.
\begin{equation}
    A_{a}(t, \textbf{x}) \longmapsto \hat{A}_{a}(t).
\end{equation}
One can compute that the reduced Lagrangian is
\begin{equation}
    \hat{\mathcal{L}}_A = -2\ell (\dot{\hat{A}}_{1})^2 - 2\ell (\dot{\hat{A}}_{2})^2 - 2\ell (\dot{\hat{A}}_{3})^2.
\end{equation}
The requirement of normalization is then
\begin{equation}
    \ell = -\frac14, \qquad n = \frac12.
\end{equation}

\subsection{11D Supergravity Multiplet}\label{11d_sugra_norm}
Consider the on-shell 11D, \(\mathcal{N} = 1\) supergravity multiplet. Using the notation of \S\ref{closure_sect}, the Lagrangian densities for the graviton and three-form are
\begin{align}
    \mathcal{L}_h &= \ell \bigg\{-\frac14 c^{a b c} c_{a b c} + \frac12 c^{a b c} c_{c a b} + (c^{a}{}_{b}{}^{b})^2\bigg\} \\
    \mathcal{L}_A &= \frac{n}{48} \bigg\{\left[\partial_{[a} A_{b c d]}\right] \left[\partial^{[a} A^{b c d]}\right]\bigg\}.
\end{align}
Replace the graviton and three-form with fields only dependent on time, as below.
\begin{equation}
    \begin{split}
        h_{a b}(t,\textbf{x}) &\longmapsto \hat{h}_{a b}(t) \\
        A_{a b c}(t,\textbf{x}) &\longmapsto \hat{A}_{a b c}(t)
    \end{split}
\end{equation}
One can compute that the reduced Lagrangians are
\begin{subequations}
    \begin{align}
        \hat{\mathcal{L}}_h &= 2\ell \sum_{i=1}^{9} \sum_{j=i+1}^{10} \left(\dot{\hat{h}}_{ij}\right)^2 - 2\ell \sum_{i=1}^{9} \sum_{j=i+1}^{10} \dot{\hat{h}}_{ii} \dot{\hat{h}}_{jj} \\
        \hat{\mathcal{L}}_A &= -18n \sum_{i=1}^{8} \sum_{j=i+1}^{9} \sum_{k=j+1}^{10} (\dot{\hat{A}}_{ijk})^2.
    \end{align}
\end{subequations}
The requirement of normalization is then
\begin{equation}\label{graviton_norm}
    \ell = \frac14, \qquad n = -\frac{1}{36}.
\end{equation}
Notice that in the literature, an alternative normalization convention is sometimes used for the graviton. In particular, as seen in (3.10) and (4.10) of \cite{Chappell2013} for 4D, \(\mathcal{N} = 1\) minimal and non-minimal supergravity, respectively, the coefficient of the graviton part of the Lagrangian density is often chosen so that the coefficients of the graviton terms in the 0-brane-reduced Lagrangian all equal 1. Nevertheless, we choose to be consistent with the most broadly used convention for supersymmetry multiplets.

\section{Propositions on Arbitrary-Spin Fermions}
\begin{proposition}\label{arb_spin_ferm_prop:1}
    If \(\Psi_{a_1 \cdots a_{s}}{}^{\gamma}\) is a (Lorentz-symmetric) spin-(\(s + 1/2\)) fermion, then a combination of terms of the form \(\mathcal{E} = c_0 \left(\gamma^{b}\right)_{\eta}{}^{\gamma} \partial_{b} \Psi_{a_1 \cdots a_{s}}{}^{\eta} - \sum_{i = 1}^{s} c_{i} \left(\gamma^{b}\right)_{\eta}{}^{\gamma} \partial_{a_{i}} \Psi_{b U_{i}}{}^{\eta}\) is gauge-invariant if and only if it is proportional to the left-hand side of \eqref{sym_fermion_eqom}, i.e., if and only if \(c_0 = c_1 = \cdots = c_{s}\).
\end{proposition}
\begin{proof}
We use the notation of \S\ref{multiplet_solver1}. Letting \(V_{ij} = U_{i} \setminus {a_{j}}\), it follows from \eqref{sym_fermion_gauge} that
\begin{equation}
    \delta_{G} \mathcal{E} = \left(\gamma^{b}\right)_{\eta}{}^{\gamma} (c_0 - c_{i})\partial_{b} \left(\sum_{i=1}^{s} \partial_{a_{i}} \kappa_{U_{i}}{}^{\eta}\right) - \sum_{i = 1}^{s} \sum_{\substack{j=1 \\ j \neq i}}^{s} c_{i} \partial_{a_{i} a_{j}} \left(\gamma^{b}\right)_{\eta}{}^{\gamma} \kappa_{b V_{ij}}{}^{\eta}.
\end{equation}
By \eqref{gauge_param_trace}, the second term vanishes. Since the summands of the first term are obviously independent, we have that \(\delta_{G} \mathcal{E} = 0\) precisely when \(c_0 = c_1 = \cdots = c_{s} \equiv c\), i.e., if and only if \(\mathcal{E}\) is \(c\) times the equation of motion \eqref{sym_fermion_eqom}.
\end{proof}

\begin{proposition}\label{arb_spin_ferm_prop:2}
If \(\Psi_{a_1 \cdots a_{s}}{}^{\gamma}\) is a (Lorentz-symmetric) spin-(\(s + 1/2\)) fermion, then
\begin{equation}\label{alt_sym_fermion_eqom}
    \left(\gamma^{E}\right)_{\eta}{}^{\gamma} \partial_{a} \Psi_{E \setminus \{a\}}{}^{\eta} = 0,
\end{equation}
where \(E\) is a set of indices, \(a \in E\), and \(|E| = s + 1\).
\end{proposition}
\begin{proof}
If \(s = 0\), then \(E \setminus \{a\} = \emptyset\) and \eqref{alt_sym_fermion_eqom} is simply the Dirac equation. If \(s = 1\), then \eqref{alt_sym_fermion_eqom} follows by multiplying the Rarita-Schwinger equation, in the form
\begin{equation}
    0 = \left(\gamma^{a}\right)_{\eta}{}^{\rho} \partial_{[a} \Psi_{b]}{}^{\eta},
\end{equation}
by \(\left(\gamma^{b}\right)_{\rho}{}^{\gamma}\), where \(E = \{a, b\}\).
Finally, if \(s > 1\), then the gamma matrix and fermion in \eqref{alt_sym_fermion_eqom} share at least two dummy indices, so \eqref{alt_sym_fermion_eqom} follows from the Lorentz-symmetry of \(\Psi\) and Lorentz-antisymmetry of \(\gamma\).
\end{proof}

\begin{corollary}\label{arb_spin_ferm_cor}
The left-hand side of \eqref{alt_sym_fermion_eqom} is a Clifford-algebra multiple of the left-hand side of the equation of motion \eqref{sym_fermion_eqom} for \(\Psi_{E \setminus \{a\}}{}^{\gamma}\).
\end{corollary}
\begin{proof}
This is really more a corollary of the proof of proposition \ref{arb_spin_ferm_prop:2} than of the proposition itself. It was shown in that proof that if \(s = 0\), then the left-hand side of \eqref{alt_sym_fermion_eqom} is precisely the left-hand side of the equation of motion (the Dirac equation); that if \(s = 1\), then the left-hand side of \eqref{alt_sym_fermion_eqom} is obtained by multiplying the left-hand side of the equation of motion (the Rarita-Schwinger equation) by a 1-gamma matrix and a number; and if \(s > 1\), then the left-hand side of \eqref{alt_sym_fermion_eqom} is zero by index considerations irrespective of the equation of motion, making it a trivial multiple of the equation of motion.
\end{proof}


\bibliographystyle{plain}
\typeout{}
\bibliography{main}

\end{document}